\documentclass[format=acmsmall, review=true]{acmart}
\usepackage{acm-ec-25}
\usepackage{booktabs} % For formal tables
\usepackage[ruled]{algorithm2e} % For algorithms
\usepackage{subfigure}

\SetAlFnt{\small}
\SetAlCapFnt{\small}
\SetAlCapNameFnt{\small}
\SetAlCapHSkip{0pt}
\IncMargin{-\parindent}

\setcitestyle{authoryear}

\title[Aggregating Information and Preferences with Bounded-Size Deviations]{Aggregating Information and Preferences with Bounded-Size Deviations}

\author{Qishen Han}
\email{hnickc2017@gmail.com}
\orcid{0000-0003-0268-6918}
\affiliation{%
  \institution{Rutgers University}\city{Piscataway}\state{NJ}\country{USA}}
\author{Grant Schoenebeck}
\email{schoeneb@umich.edu}
\orcid{0000-0001-6878-0670}
\affiliation{%
  \institution{University of Michigan}\city{Ann Arbor}\state{MI}\country{USA}}
\author{Biaoshuai Tao}
\email{bstao@sjtu.edu.cn}
\orcid{0000-0003-4098-844X}
\affiliation{%
  \institution{Shanghai Jiao Tong University}\city{Beijing}\country{China}}
\author{Lirong Xia}
\email{xialirong@gmail.com}
\orcid{0000-0002-9800-6691}
\affiliation{%
  \institution{Rutgers University and DIMACS}\city{Piscataway}\state{NJ}\country{USA}}

\begin{abstract}
We investigate a voting scenario with two groups of agents whose preferences depend on a ground truth that cannot be directly observed. The majorities' preferences align with the ground truth, while the minorities disagree. Focusing on strategic behavior, we analyze situations where agents can form coalitions up to a certain capacity and adopt the concept of ex-ante Bayesian $\kd$-strong equilibrium, in which no group of at most $\kd$ agents has an incentive to deviate. Our analysis provides a complete characterization of the region where equilibria exist and yield the majority-preferred outcome when the ground truth is common knowledge. This region is defined by two key parameters: the size of the majority group and the maximum coalition capacity. When agents cannot coordinate beyond a certain threshold determined by these parameters, a stable outcome supporting the informed majority emerges. The boundary of this region exhibits several distinct segments, notably including a surprising non-linear relationship between majority size and deviation capacity. Our results reveal the complexity of the strategic behaviors in this type of voting game, which in turn demonstrate the capability of the ex-ante Bayesian $\kd$-strong equilibrium to provide a more detailed analysis.

\end{abstract}

\begin{document}
\newtheorem{thm}{Theorem}
\newenvironment{thmnb}[1]
  {\renewcommand{\thethm}{\ref{#1}}%
   \addtocounter{thm}{-1}%
   \begin{thm}}
  {\end{thm}}

\newtheorem{claim}{Claim}
\newtheorem{prop}{Proposition}
\newtheorem{lem}{Lemma}
\newenvironment{lemnb}[1]
  {\renewcommand{\thelem}{\ref{#1}}%
   \addtocounter{lem}{-1}%
   \begin{lem}}
  {\end{lem}}

\newenvironment{propnb}[1]
  {\renewcommand{\theprop}{\ref{#1}}%
   \addtocounter{prop}{-1}%
   \begin{prop}}
  {\end{prop}}

\newenvironment{smallblock}{\scriptsize}{\par}

\newtheorem{asm}{Assumption}
\newtheorem{conj}{Conjecture}

\newtheorem{Alg}{Algorithm}
\newtheorem{prob}{Problem}
\newtheorem{prof}{Proof}
\newtheorem{coro}{Corollary}
\newenvironment{coronb}[1]
  {\renewcommand{\thecoro}{\ref{#1}'}%
   \addtocounter{coro}{-1}%
   \begin{coro}}
  {\end{coro}}
\newtheorem{remark}{Remark}

\theoremstyle{acmdefinition}
\newtheorem{dfn}{Definition}

\newtheorem{ex}{Example}

\newcommand\qishen[1]{{\color{blue} \footnote{\color{blue}Qishen: #1}} }
\newcommand{\Qishen}[1]{\textbf{\color{orange} [Qishen: #1\textbf{]}}} 
\newcommand{\Biaoshuai}[1]{\textbf{\color{pink} [Biaoshuai: #1\textbf{]}}} 
\newcommand\gs[1]{{\color{blue} \footnote{\color{blue}Grant: #1}} }

\newcommand{\blue}[1]{\textcolor{blue}{#1}}

\newcommand{\sag}{i}
\newcommand{\ag}{n}
\newcommand{\ut}{u}
\newcommand{\vt}{v}
\newcommand{\Ut}{U}
\newcommand{\stg}{\sigma}
\newcommand{\stgp}{\Sigma}

\newcommand{\acc}{A}
\newcommand{\err}{I}

\newcommand{\wos}{k} %world state
\newcommand{\Wos}{K} %world state num
\newcommand{\Wosset}{\mathcal{W}} % set of world states
\newcommand{\Wosrv}{W} %rv of world state
\newcommand{\sig}{m} %signal
\newcommand{\sigi}{s}
\newcommand{\Sigrv}{\Psi}
\newcommand{\Sig}{M} %signal num
\newcommand{\Sigset}{\mathcal{S}} % signal set

\newcommand{\bA}{\mathbf{A}}
\newcommand{\bR}{\mathbf{R}}
\newcommand{\bP}{\mathbf{P}}

\newcommand{\Thd}{\mu}
\newcommand{\thd}{f}
\newcommand{\hthd}{\hat{\thd}}
\newcommand{\thdmaj}{\theta}

\newcommand{\agf}{\alpha_{F}}
\newcommand{\agu}{\alpha_{U}}
\newcommand{\agc}{\alpha_{C}}
\newcommand{\aga}{\alpha_{A}}
\newcommand{\agi}{\alpha_{I}}
\newcommand{\tf}{\ag_{F}}
\newcommand{\tu}{\ag_{U}}
\newcommand{\tc}{\ag_{C}}
\newcommand{\ta}{\ag_{A}}
\newcommand{\ti}{\ag_{I}}
\newcommand{\lp}{\lambda}
\newcommand{\instance}{\mathcal{I}}

\newcommand{\ps}{PS} % proper scoring rule
\newcommand{\eps}{\tilde{PS}} %psr between distributions
\newcommand{\rw}{R} % Reward / Payment
\newcommand{\prQ}{\vpr} %prior Q
\newcommand{\pr}{q} %signal prob
\newcommand{\vpr}{\mathbf{q}} %signal dist

\newcommand{\Ex}{\mathbb{E}}% Expectation

\newcommand{\rp}{a} %report
\newcommand{\rpset}{\mathcal{A}} %report set
\newcommand{\sigp}{S}
\newcommand{\rpp}{A}
\newcommand{\rwd}{R} %reward /payment

\newcommand{\bpl}{\beta_{\ell}}
\newcommand{\bph}{\beta_{h}}

\newcommand{\astg}{\bar{\stg}} %average stg
\newcommand{\abpl}{\bar{\bpl}}
\newcommand{\abph}{\bar{\bph}}
\newcommand{\aut}{\bar{\ut}}
\newcommand{\astgp}{\bar{\stgp}}
\newcommand{\abpls}{\bar{\bpl^*}}
\newcommand{\abphs}{\bar{\bph^*}}

\newcommand{\nbph}{\hat{\bph}}
\newcommand{\nbpl}{\hat{\bpl}}

\newcommand{\kd}{k} %k, limit size of deviators
\newcommand{\kc}{c}

\newcommand{\qi}{interim}
\newcommand{\afrac}{\alpha}
\newcommand{\ds}{\delta} % distribution put prob 1 on s. 

\newcommand{\jtruthful}{\text{truthful}}
\newcommand{\jdeviate}{\text{deviator}}

\newcommand{\dut}{\Delta u}

\newcommand{\func}{f}
\newcommand{\gunc}{g}
\newcommand{\munc}{m}
\newcommand{\kde}{k_E}
\newcommand{\kdq}{k_B}
\newcommand{\maxdps}{\overline{\Delta \ps}}

\newcommand{\dpsh}{\Delta h}
\newcommand{\dpsl}{\Delta \ell}

\newcommand{\bphth}{\bar{b}_h}
\newcommand{\bphtl}{\bar{b}_l}

\newcommand{\inst}{\mathcal{I}}
\newcommand{\iset}{I}

\newcommand{\xrv}{X}

\newcommand\blfootnote[1]{%
  \begingroup
  \renewcommand\thefootnote{}\footnote{#1}%
  \addtocounter{footnote}{-1}%
  \endgroup
}

\newcommand{\textinst}{environment}
\newcommand{\piv}{piv}
\newcommand{\exshare}{excess expected vote share}
\newcommand{\infdif}{\Delta}
\newcommand{\maj}{{maj}}
\newcommand{\mino}{{min}}

\newcommand{\fone}{\gamma_1}
\newcommand{\fz}{\gamma_0}
\newcommand{\sign}{\text{sign}}

\newcommand{\xinl}{\xi_{NL}}
\newcommand{\anl}{\alpha_{NL}}

\newcommand{\nocontentsline}[3]{}
\let\origcontentsline\addcontentsline
\newcommand\stoptoc{\let\addcontentsline\nocontentsline}
\newcommand\resumetoc{\let\addcontentsline\origcontentsline}
% Title page for title and abstract only.
\begin{titlepage}

\maketitle

% Optionally include a table of contents
% \vspace{0.05cm}
\setcounter{tocdepth}{1} % adjust to 1 if desired
\tableofcontents
\end{titlepage}

% Paper body
\section{Introduction}

Majority voting is the most studied and applied method for making efficient, fair, and informed decisions from two alternatives. In practice, it is used in diverse scenarios such as political elections, corporate board meetings, academic or industry hiring processes, and even popular shows, just to name a few. Theoretically, extensive research in social choice supports majority voting as reliable both procedurally and in its conclusions. When voters have complete information about their preferences, majority voting ensures that casting a dishonest vote does not provide an advantage. In scenarios where agents have incomplete information, such as a jury system, the famous Condorcet Jury Theorem~\citep{Condorcet1785:Essai} shows that agents can collectively reveal the ground truth determining their preferences and reach the “correct” decision with high probability.

\begin{ex}[Condorcet Jury Theorem]
    \label{ex:CJT}
    A group of jurors vote on whether a defendant is guilty or not. Each juror receives an independent signal with $p > \frac12$ probability to be correct. The Condorcet Jury Theorem states that (1) the likelihood that a correct decision is made by the majority vote is larger than that made by any juror individually, (2) such likelihood increases as the number of voters increases, and (3) the likelihood converges to 1 and as the number of voters goes to infinity. 
\end{ex}

The Condorcet Jury Theorem, while foundational, relies on two restrictive assumptions about voter behavior and motivation. Agents are assumed to vote {\em informatively}, i.e., honestly reflect their information in the vote and to share the identical objective of selecting the alternative that matches the ground truth. These assumptions significantly limit the applicability of the theorem, as they fail to capture many real-world voting scenarios where voters may have strategic considerations or divergent objectives. The following example illustrates these limitations.

% Nevertheless, the Condorcet Jury Theorem sets strong assumptions about the agents' preferences and behaviors. These agents are assumed to vote {\em informatively}, i.e., honestly reflect their information in the vote, and the goal of every agent is to select the alternative aligned with the ground truth. These two assumptions rule out a large spectrum of real-world scenarios, as illustrated in the following example. 

\begin{ex}[Voters with conflicting preferences and incomplete information]
    \label{ex:motive}
    % \gs{This example is not very good because I don't know of anyone that argues that tariffs are deflationary.  I put some GPT suggestions in the file "GPT policy suggestions."  also, I think it is fine using  expanding or contracting the economy as two potential results where some people prefer one of these outcomes to the other.  Just chose a policy that seems more likely to lead to both. }\Qishen{Change the policy.}
    Consider a (hypothetical) House vote on whether to accept (denoted as $\bA$) or reject (denoted as $\bR$) a mixed economic policy. The consequence of the policy can be a higher inflation rate (denoted as $H$) if the expansionary component takes a primary effect, or a lower inflation rate (denoted as $L$) if the contractionary component takes a primary effect. The representatives can be categorized into two wings based on their preferences towards the consequence. The majority Doves support a higher inflation rate to stimulate the economy, while the minority Hawks are concerned about inflation and prefer a lower inflation rate. Therefore, if the policy is sure to increase the inflation rate, the Doves will vote for it, while the Hawks will vote against it, and vice versa. However, the real effect of the policy is unknown until it is implemented. Every representative has private noisy information about the real effect. Both wings of representatives will try to reach their preferred consequence via a majority vote. Under what circumstances does the vote lead to the consequences favored by the majority?
\end{ex}

% \Qishen{I use GPT to imporve this paragraph. } 

% \gs{I changed a lot at the end of this paragraph, please check it out.}
Example~\ref{ex:motive} illustrates the complexity of voting problems. There are two groups of agents with drastically different preferences, prompting each group to maximize the probability of achieving its preferred outcome. However, because their preferences hinge on an unobservable ground truth, the strategic behaviors of agents can be complex. Even under the setting where agents share identical preferences, \citet{austen1996information} show that strategic behaviors can be counterintuitive: agents reason about their best strategy under the hypothetical situation that  all other votes form a tie, because this is the only situation in which their vote is decisive.  The information provided to an agent in assuming such a situation, typically causes the agent to deviate from informative voting that would otherwise lead to a favorable result.  Moreover, although this hypothetical situation is often exponentially rare, it impacts the agents vote in every situation.  

% Example~\ref{ex:motive} reveals the complication in a voting problem. Firstly, there are two types of agents that have utterly opposed preferences. The opposition motivates both types of agents to maximize the chance that their preferred consequence is reached.  However, as the preferences of agents depend on a ground truth that is not directly observable, the strategic behaviors of agents can be complex. In fact, even under the setting where agents have identical preferences, \citet{austen1996information} show that the strategic behaviors can be counterintuitive. From their analysis, agents reason their best strategy based on the hypothetical information such that all other votes form a tie and deviate from informative voting which leads to a good outcome. 

% \gs{I would make new paragraphs for the two below issues.  I think that the point here is that we are basically making a modeling choice to use eps-strong NE, and we need to justify that.  I think that the justificaiton is multi-fold.  1) people before us have done this; 2) without eps, you get rather stylized an not broadly applicable behavior predicted in ABS.  3) without strong you are essentially assumeing that no body can coodinate, but there is cable news, etc, and these talking head provide some ability to coordinate.  Also, maybe state when these assumptions are unlikely to hold.}\Qishen{I agree with your idea on justifying the model choice early. But I prefer not mentioning 2), because some view this as criticizing ASB. I remember Biaoshuai encountered such a reviewer getting pissed off.}

More importantly, while individual agents may have limited influence on voting outcomes, they can increase their impact by forming strategic coalitions with others who share similar objectives. The widespread adoption of the Internet and social media has made such coordination significantly easier. However, the capability of group coordination varies depending on organizational structures. They may be limited to local communities and trusted networks, or extend as broadly as political parties.
These varying capabilities for group coordination call for a more refined framework for analyzing strategic behavior. In this paper, we adopt the {\em ex-ante Bayesian \kd-strong equilibrium}~\citep{han2025kstrong}. This framework explicitly considers scenarios where no group of up to $\kd$ agents has an incentive to deviate.  
To address the aforementioned issue from Austen-Smith and Banks, that agents may act only based on a hypothetical that occurs with vanishingly small probability, we focus on approximate equilibria, namely {\em $\varepsilon$-ex-ante Bayesian $\kd$-strong equilibrium}~\citep{han2025kstrong}, following the convention in previous literature~\cite{han2023wisdom,deng2024aggregation}. 
We focus on scenarios where strategic behavior leads to an {\em informed majority decision}, i.e., the alternative preferred by the majority of agents if all agents had complete knowledge of the ground truth.  This concept has been widely adopted as the benchmark for a desired voting outcome in prior literature examining agents with heterogeneous preferences ~\citep{feddersen1997voting, schoenebeck21wisdom, han2023wisdom,deng2024aggregation}. In settings with conflicting interests among agents, the informed majority decision represents the alternative preferred by the majority group when fully informed. For instance, in Example 1, the informed majority decision would be to accept the policy if it increases inflation and reject it otherwise. Our research seeks to address the following research question.

% \gs{The below RQ is not readable currently.}\Qishen{Could you specify the problem? I tried GPT but it give almost a same sentence.}

% \begin{center}
%     \textbf{When is an informed majority decision reached by opposing agents with incomplete information?}
% \end{center}

\begin{center}
    \textbf{When does majority voting achieve informed decisions under capacitated coalitional strategic behavior?}
\end{center}

Under a special case where agents share common preferences, \citet{han2023wisdom} guarantee the existence of strong equilibrium that reaches the informed majority decision even for unlimited coalitional strategic behavior. However, the question remains largely open when agents have drastically different or even conflicted preferences. \citet{deng2024aggregation} established conditions for the existence of strong equilibria that resist unlimited group deviations. Nevertheless, their analysis does not address scenarios where agents have limited capability of group coordination. This leaves open critical questions about voting outcomes and their alignment with informed majority decisions in cases where such strong equilibria do not exist.

\subsection{Our Contribution}
We completely address the question above by obtaining a closed form of the threshold $\xi^*$ as a function of the fraction of majority agents $\alpha$. This threshold represents the maximum fraction of agents that can be prevented from deviating while still maintaining an (approximated) equilibrium that achieves the informed majority decision.  

\begin{thmnb}{thm:thresholdk}[Threshold $\xi^*$ (informal)] Let $\xi^*(\alpha)$ be the threshold curve. For any $\xi < \xi^*(\alpha)$, there exists an $\varepsilon$-ex-ante Bayesian $\xi\ag$-strong equilibrium such that (1) no agents play weakly dominated strategies, (2) the informed majority decision is reached with probability converging to 1, and (3) $\varepsilon$ converges to 0 as $\ag$ goes to infinity. When $\xi = \xi^*$, no such equilibrium exists. 
\end{thmnb}

\begin{figure}[htbp]
    \subfigure[Case 1: four segments.]{
        \centering
        \includegraphics[width = 0.45\linewidth]{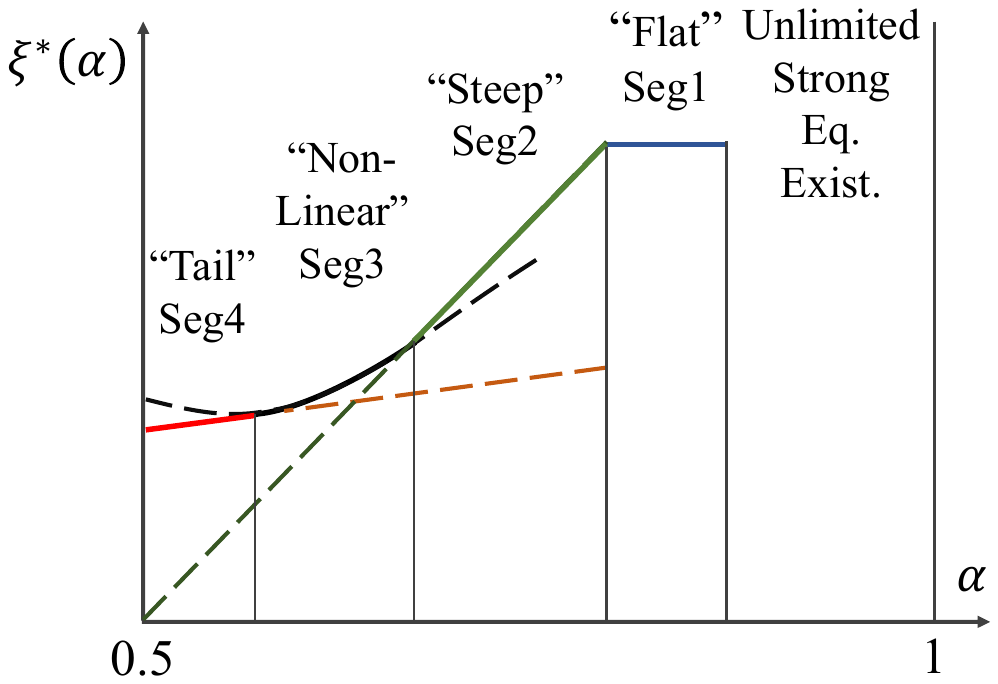}}
    \subfigure[Case 2: three segments.]{
        \centering
        \includegraphics[width = 0.45\linewidth]{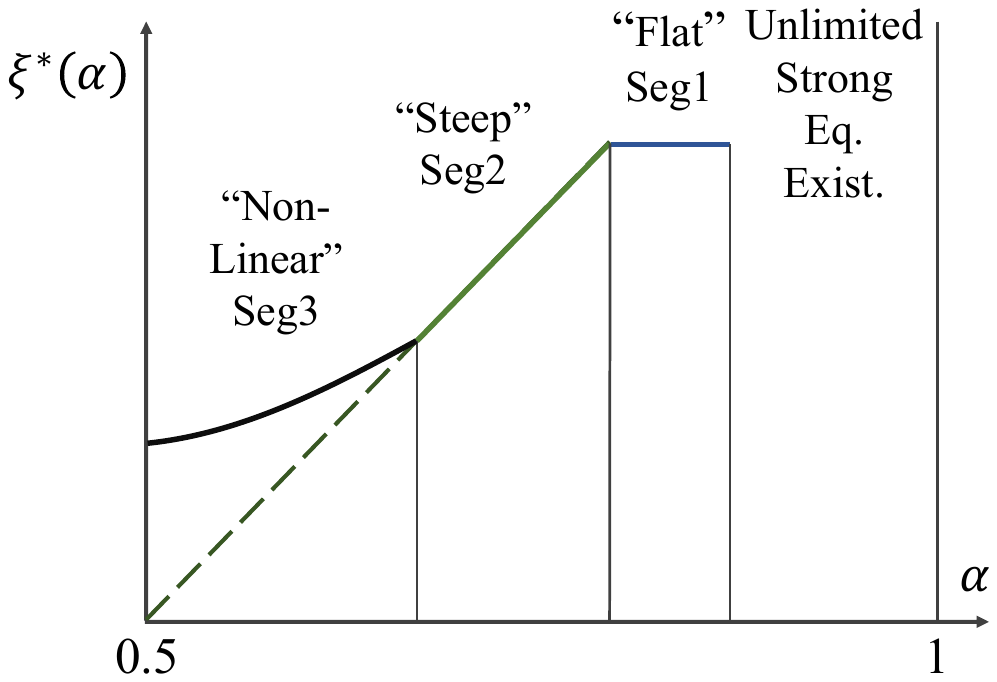}}
    \caption{The threshold curve $\xi^*(\alpha)$. The segments are proved in our Theorem~\ref{thm:thresholdk}. The existence of (unlimited) strong equilibrium is from~\cite{deng2024aggregation}. \label{fig:curve}}
\end{figure}

As illustrated in Figure~\ref{fig:curve}, 
the threshold $\xi^*(\alpha)$ has three or four segments depending on the distribution of the private information agents receive. We label them as Segment 1 to 4 in decreasing order on $\alpha$ and assign each segment a descriptive name. Firstly, there is the ``Flat'' Segment 1, where $\xi^*$ remains a constant value despite the decreasing $\alpha$. In the ``Steep'' Segment 2, $\xi^*$ declines linearly as $\alpha$ decreases, with the linear projection reaching 0 when $\alpha$ goes to $\frac12$. In the ``Non-Linear'' Segment 3, $\xi^*$ displays a surprising non-linear dependence on $\alpha$. If Segment 3 ends before $\alpha = \frac12$, there will be a linear ``Tail'' Segment 4; otherwise, $\xi^*$ ends with three segments.

% Our results reveal the complexity and variety of the strategic behaviors of conflict agents with incomplete information. Through the k-strong equilibrium framework, we provide a more nuanced characterization than the previous solution concept that considered unlimited group sizes. The threshold $\xi^*$ emerges as a critical parameter: when agents' coordination capacity is constrained below this threshold, the voting system can still achieve stable outcomes that reflect the informed majority's preference.

The threshold $\xi^*$ illuminates the interplay between coalition size, strategic behavior, and decision outcomes in voting systems. If agents cannot coordinate within at least a $\xi$ fraction of the population, a stable outcome favoring the informed majority still emerges. Notably, $\xi^*$ maintains a non-zero value as $\alpha$ approaches $\frac{1}{2}$, indicating that even with a slim numerical advantage, an uncoordinated minority still yields a stable outcome favoring the bare majority.

% This framework enhances our understanding of how coalition formation capabilities influence voting outcomes. It demonstrates that limiting the size of potential coalitions can actually promote stability and lead to more informed collective decisions, even in environments where agents have conflicting interests.
Our results reveal the complexity and variety of the strategic behaviors of conflicting agents with incomplete information, which in turn demonstrates the capability of the ex-ante Bayesian $\kd$-strong equilibrium to provide a more detailed perspective compared to its unlimited group-size counterpart. The presence of multiple segments with distinct shapes reveals drastically different agent behaviors as the power balance between types shifts. For example, when the majority fraction is large, all minority agents converge on the same alternative in an equilibrium. However, as this fraction approaches $\frac{1}{2}$, minority agents split evenly between two alternatives in the equilibrium.

% $\xi$ serves as a threshold: if agents cannot organize within at least a $\xi$ fraction of the population, a stable outcome favoring the informed majority may still emerge. This framework provides insights into the interplay between coalition size, strategic behavior, and decision outcomes in voting systems. Notably, $\xi^*$ maintains a non-zero value as $\alpha$ approaches $\frac12$. This indicates that even when the majority holds only a slim numerical advantage over the minority, an uncoordinated minority still results in a stable outcome that favors the bare majority. he presence of multiple segments with distinct shapes demonstrates that agents can behave very differently as the power balance between the two types shifts. When the majority fraction is large, all minority agents converge on the same alternative in equilibrium; however, when this fraction nears $\frac12$, minority agents split evenly between two alternatives.

At a high level, the existence of an equilibrium hinges on the fraction of votes each alternative gets in expectation. For an equilibrium to exist, the informed majority's preferred choice must secure more than half of the expected votes under both ground truth states, even when a fraction $\xi$ of agents deviates from their prescribed strategies. Consequently, the majority agents employ strategies that maximize their margin above the $\frac{1}{2}$ threshold in expected vote share. Minority agents, conversely, work to diminish this margin and push the vote share below the majority threshold. The critical deviation threshold $\xi^*$ is determined by three competing factors: a majority's ability to build its margin, its requirement to maintain the margin under both ground truth states, and a minority's ability to deviate. This interplay makes finding $\xi^{(\alpha)}$ a complex analytical challenge.

We also analyze the relationship between the shape of the threshold $\xi^*$ and the distribution of the private information of the agents. When agents possess more accurate private information, the threshold $\xi^*$ exhibits contracts in range while increases in values. This observation aligns with the intuition that agents with more precise information are more certain about their preferences, thereby increasing the likelihood of reaching a stable outcome.

\subsection{Related Works}
The study of binary voting with incomplete information traces back to the famous Condorcet Jury Theorem~\citep{Condorcet1785:Essai}, which has ``formed the basis for the development of social choice and collective decision-making as modern research fields''~\citep{Nitzan17:Collective}. The Condorcet Jury Theorem considers the scenario where agents honestly reflect their information in the vote and gives merit to the majority vote to make a correct decision. Extensive liteature generalize the Condorcet Jury Theorem into a wider range of scenarios, including correlated private information~\citep{Nitzan84:Significance,Shapley84:Optimizing}, agents with different competencies~\citep{Nitzan80:Investment,Ben11:Condorcet}, and plurality voting with more alternatives~\citep{Young88:Condorcet,Goertz14:Condorcet}. Experimental studies have also been conducted to examine the Theoerm in real-world scenarios~\citep{Ignacio2014hypo,Battaglini2010swing,Goeree2011experimental}. \citet{austen1996information} pioneers the study of the Condorcet Jury Theorem under a game theory context and show that the classical non-strategic behaviors may fail to form a Nash equilibrium. Consequently, a large literature focuses on the strategic behaviors of the voters and the existence of an equilibrium that reaches a correct outcome~\citep{Meirowitz02:Informative, wit1998rational,myerson1998extended,feddersen1997voting,Feddersen1998:innocent}.

A recent line of work brings up the coalitional strategic behavior in voting with incomplete information. \citet{schoenebeck21wisdom} design a mechanism to incentivize truthful reporting and reach the outcome that favors the majority as the ground truth is known to all. \citet{han2023wisdom} reveals a surprising equivalent between strategies reaching the informed majority decision and those forming an (approximated) strong equilibrium. Their work provides strong support for the majority vote reaching a good outcome via strategic behavior. Both studies adopt a similar model where agents have diverse preferences yet retain some common value. \citet{deng2024aggregation} study agents with drastically different preferences and characterize conditions that a strong equilibrium exists. All three papers apply the solution concept of $\varepsilon$-strong Bayes Nash equilibrium which no group (with unlimited size) of agents has incentives to deviate. 

While \citet{austen1996information} and many works~\cite{Duggan01:Bayesian, wit1998rational,coughlan2000defense} assume identical utility functions across agents, a subsequent research has explored more diverse utility structures. In the first direction, agents have different ordinal preference yet share some common values~\cite{feddersen1997voting,Gerardi07:Deliberative,han2023wisdom,schoenebeck21wisdom,Morgan2008Polls}. These models typically feature three agent types: two groups with fixed preferences for two alternatives respectively, and a third group whose preferences align with the ground truth. In the second direction which our paper follows, agents have opposite preferences contingent on the ground truth~\citep{bhattacharya2013preference,bhattacharya2023condorcet,kim2007swing,ali2018adverse}.

Our study also closely relates to research on group strategic behavior, particularly those limited capacity. There is extensive literature~\citep{desmedt2010equilibria,barbera2001voting,rabinovich2015analysis,holzman1997strong,yin2011nash,harks2012existence} examining group-strategic behavior, which originated from the concept of strong Nash equilibrium in complete information games~\citep{Aumann59:Acceptable}. In the context of Bayesian games, group strategic behavior has been investigated across multiple domains, including cooperative games~\citep{ichiishicooperative,ichiishi1996bayesian}, implementation theory in mechanism design~\citep{hahn2001coalitional,safronov2018coalition}, and voting.
% Another line of work closely related to our study is group strategic behavior , espcially those with limited capacity. There is well established literature\citep{desmedt2010equilibria, barbera2001voting,rabinovich2015analysis,holzman1997strong,yin2011nash,harks2012existence} on studying group-strategic behavior originated from the strong Nash equilibrium in the game with complete information~\citep{Aumann59:Acceptable}.In Bayesian games, group strategic behavior has been studied in cooperative games~\citep{ichiishicooperative,ichiishi1996bayesian}, implementation theory in mechanism design~\citep{hahn2001coalitional,safronov2018coalition}, and voting. 
Recent research has highlighted the importance of examining capacitated coalitional behavior and has introduced various solution concepts for this scenario.\citet{guo2022robust} propose coalitional interim equilibrium parameterized by an exogenously given set of all admissible coalitions. \citet{han2025kstrong} propose (ex-ante) Bayesian $\kd$-strong equilibrium parameterized by the largest capacity of the deviating group $\kd$. \citet{Abraham2008:lower} also propose a $k$-coalitional equilibrium where the deviators will share all their private information in the group. \citet{Abraham2006:distributed} study capacitated strategic behavior under a non-Bayesian game, where agent have complete information about their preferences. The introduction of capacitated strong equilibria enables theories to model a broader range of real-world scenarios and provides more nuanced insights, particularly in situations where unlimited strong equilibria are too restrictive to exist~\citep{Gao2019incentivizing}.  
\section{Preliminaries}

$\ag$ agents vote for two candidates $\bA$ and $\bR$. The majority rule is applied, so the candidate with more than $\frac{\ag}{2}$ votes becomes the winner. There is a world state (ground truth) $\Wosrv \in \{L, H\}$ that affects the preferences of the agents. The common prior on the world state is denoted by $P_H$ and $P_L$. The world state is not directly observable. Instead, each agent receives a private signal $\sig \in \{\ell, h\}$ whose distribution depends on the world states. The signals of different agents are i.i.d. conditioned on the world state. For a world state $\Wosrv$ and a signal $\sig$, let $P_{\sig \Wosrv}$ denote the probability that an agent receives $\sig$ signal when the world state is $\Wosrv$. We assume that an $h$ signal is more likely to occur in world state $H$ than in world state $L$, i.e., $P_{hH} > P_{hL}$, and vice versa. Let $\infdif = P_{hH} - P_{hL}$ be the difference on the signal distribution under two world states. From the fact that $P_{hH} + P_{\ell H} = P_{\ell L} +P_{hL} = 1$, we have the following observation. 

\begin{prop}
    $\infdif = P_{hH} - P_{hL} = P_{\ell L} - P_{\ell H} = P_{hH}\cdot P_{\ell L} - P_{hL}\cdot P_{\ell H} = P_{hH} + P_{\ell L} - 1$.
\end{prop}
% $\infdif = P_{hH} - P_{hL} = P_{\ell L} - P_{\ell H} = P_{hH}\cdot P_{\ell L} - P_{hL}\cdot P_{\ell H}$. \gs{why is this equation true?  It is confusing as stated.}  \gs{excited that I finally figured out where Delta is defined.}  

For each agent $i$, the utility function $\vt_i: \{H, L\}\times \{\bA, \bR\} \to \{0, 1, \cdots, B\}$ represents the preference of the agent. There are two types of agents, whose preferences depend on the world state. The majority type prefer $\bA$ in world state $H$ ($\vt(H,\bA) > \vt(H, \bR)$) and $\bR$ in world state $L$ ($\vt(L,\bA) < \vt(L, \bR)$), and the minority type prefer $\bR$ in world state $H$ ($\vt(H,\bA) < \vt(H, \bR)$) and $\bA$ in world state $L$ ($\vt(L,\bA) > \vt(L, \bR)$). Let $\alpha > 0.5$ be the approximated fraction of the majority agents.
% and $\ag_\maj = \lfloor \alpha \cdot \ag\rfloor$ be the number of the majority agents. 

The {\em informed majority decision} is the decision to be made if all the agent know the world state. In this setting, the informed majority decision follows the majority type agents' preference: the informed majority decision is $\bA$ in the world state $H$ and $\bR$ in the world state $L$.

\begin{ex}[Information Structure and Agent Preferences] 
    \label{ex:setting}
    Consider the House voting scenario in Example~\ref{ex:motive}. Suppose $\ag = 50$ agents vote on a mixed policy. The world state is the real effect of the policy, which can be a higher inflation rate ($H$) or a lower inflation rate ($L$). The prior that a high inflation rate occurs is $P_H = 0.6$ and $P_L = 0.4$. The distribution is $P_{hH} = 0.7$, $P_{\ell H} = 0.3$, $P_{hL} = 0.2$, and $P_{\ell L} = 0.8$. This means that an agent receives an $h$ signal with probability $0.7$ when the world state is $H$ and probability $0.2$ when the state is $L$. The utilities of the agents are illustrated in Table~\ref{tab:util}. In this example, we assume that agents of the same type share an identical utility function, which may not necessarily be true in general.

    \noindent\begin{minipage}{\linewidth}
\vspace{0.3cm}
\centering
\begin{tabular}{@{}ccccc@{}}
\toprule
% Winner      & \multicolumn{2}{c}{$\bA$} & \multicolumn{2}{c}{$\bR$} \\ \midrule
Type         & $\vt_i(H,\bA)$         & $\vt_i(L,\bA)$         & $\vt_i(H,\bR)$         & $\vt_i(L,\bR)$         \\ \midrule
Majority Agents & 4         & 0          & 1           & 2          \\
Minority Agents & 2          & 3          & 3           & 1          \\
\bottomrule
\end{tabular}
\captionof{table}{Utility of agents.\label{tab:util}}
\end{minipage}
\end{ex}

A strategy $\stg$ of an agent $i$ maps his/her signal to a distribution on $\{\bA,\bR\}$, denoting the action of the agent. A strategy profile $\stgp$ is a vector of the strategies of all the agents. For an agent $i$, let $\bpl^i$ and $\bph^i$ be the probability that $i$ votes for $\bA$ when his/her private signal is $\ell$ and $h$, respectively. 

% We assume that $\agf + \aga > \Thd$ and $\agu + \aga > 1 - \Thd$. This assumption implies that $\aga > \agi$, i.e., the``majority type'' agents are the majority in the contingent agents. In this case, the informed majority decision is $\bA$ in the world state $H$ and $\bR$ in the world state $L$. We believe such assumption is reasonable. If exactly one of the two inequalities hold, the informed majority decision will be the same alternative in both world states, and reaching such a decision is trivial. If both inequalities are violated, the minority type agents become the majority of the contingent agents, the informed majority decision is $\bA$ in the world state $H$ and $\bR$ in the world state $L$, which is symmetric to the scenario under the assumption. 

The {\em fidelity} of a strategy profile is the likelihood that the candidate preferred by the agents becomes the winner under this strategy profile. That is, $\bA$ becomes the winner when the world state is $H$, and $\ell$ become the winner when the world state is $\ell$. Let $\lp_{\Wosrv}^{\bA}(\stgp)$ and $\lp_{\Wosrv}^{\bR}(\stgp)$ be the probability that $\bA$ ($\bR$, respectively) becomes the winner when the world state is $W$ and the strategy profile is $\stgp$. Then the fidelity in the voting game is in the following form: 
\begin{equation*}
    \acc(\stgp) = P_H\cdot \lp_H^{\bA}(\stgp) + P_L\cdot \lp_L^{\bR}(\stgp). 
\end{equation*}

Similarly, we can define the (ex-ante) expected utility of an agent $i$ under a strategy profile $\stgp$: 
\begin{equation*}
    \ut_{i}(\stgp) =P_{L}(\lp_{L}^{\bA}(\stgp)\cdot\vt_{i}(L, \bA) + \lp_{L}^{\bR}(\stgp)\cdot \vt_{i}(L, \bR)) +  P_{H}(\lp_{H}^{\bA}(\stgp)\cdot\vt_{i}(H, \bA) + \lp_{H}^{\bR}(\stgp)\cdot \vt_{i}(H, \bR)). 
\end{equation*}

% Similarly, for an agent $i$, let $\lp_{\Wosrv}^{\bA}(\stgp \mid \sigi)$ and $\lp_{\Wosrv}^{\bR}(\stgp\mid \sigi)$ be the probability that $\bA$ ($\bR$, respectively) becomes the winner when the world state is $W$, the strategy profile is $\stgp$, and agent $i$ receives signal $\sigi$. The formula does not contain $i$ because of the symmetricity of agents. Then we can define the interim expected utility of an agent $i$ with private signal $\sigi$ under strategy profile $\stgp$. 
% \begin{equation*}
%     \ut_{i}(\stgp) =\Pr[L\mid \sigi] \cdot (\lp_{L}^{\bA}(\stgp \mid \sigi)\cdot\vt_{i}(L, \bA) + \lp_{L}^{\bR}(\stgp\mid \sigi)\cdot \vt_{i}(L, \bR)) +  \Pr[H\mid \sigi](\lp_{H}^{\bA}(\stgp \mid \sigi)\cdot\vt_{i}(H, \bA) + \lp_{H}^{\bR}(\stgp \mid \sigi)\cdot \vt_{i}(H, \bR)). 
% \end{equation*}

\paragraph{Non-dominated strategies}
Our theoretical results concern equilibria in which agents employ only non-dominated strategies. A strategy is dominated if playing a different strategy always brings the agent a higher expected utility.
\begin{dfn}
    A strategy $\stg_i$ is an (ex-ante) weakly dominated strategy for agent $i$ if there exists a strategy $\stg'_i$ such that, regardless of what other agents play $\stgp_{-i}$, $\ut_i((\stg_i, \stgp_{-i})) \le \ut_i((\stg_i, \stgp_{-i}))$, and there exists a $\stgp_{-i}$ such that $\ut_i((\stg_i, \stgp_i)) < \ut_i((\stg_i, \stgp_{-i}))$. 
\end{dfn}

The constraint of no agents playing weakly dominated strategies follows the convention in the literature~\citep{Feddersen96:Swing,feddersen1997voting,wit1998rational}, and we believe this restriction is mild and reasonable. Firstly, an agent playing a dominated strategy has incentives to deviate to a non-dominated strategy, which brings no utility loss in any circumstances and utility gain in some circumstances. This contradicts the intention of an equilibrium that predicts a stable outcome. Secondly, as shown in the following lemma, whether a strategy is (weakly) dominated can be easily determined, and playing a non-dominated strategy does not bring extra computational cost and difficulties to the agents. 

\begin{lem}
\label{lem:dominated}
    For an arbitrary agent $i$, a strategy $\stg = (\bpl, \bph)$ is NOT a weakly dominated strategy if and only if the following conditions hold. 
    \begin{itemize}
        % \item If $i$ is a friendly agent: $\bpl = \bph = 1$. 
        % \item If $i$ is an unfriendly agent: $\bpl = \bph = 0$. 
        \item If $i$ is a majority agent: $\bpl = 0$ or $\bph = 1$. 
        \item If $i$ is a minority agent: $\bpl = 1$ or $\bph = 0$. 
    \end{itemize}
\end{lem}

\paragraph{Proof Sketch of Lemma~\ref{lem:dominated}.} We give the reasoning for a majority agent, and that for a minority agent is similar. A majority agent's expected utility increases with the probability that $\bA$ wins in world state $H$ and decreases with the probability that $\bR$ wins in world state $H$. Given that a signal $h$ implies the world state $H$ is more likely than a signal $\ell$, increasing $\bph$ and decreasing $\bpl$ usually helps the agent to increase his expected utility. Therefore, a strategy where $\bpl > 0$ and $\bph < 1$ is dominated by a strategy with a larger $\bph$ and a smaller $\bpl$. On the other hand, when $\bpl = 0$ or $\bph = 1$, any change in the strategy decreases the expected utility under world state $H$ or under world state $L$. Therefore, there will exist a strategy profile of all other agents where the loss exceeds the gain in expectation, and the expected utility of the agent decreases. Therefore, in this case, we can show the strategy profile is not weakly dominated. The full proof is in Appendix~\ref{apx:dominated}. 
% The only scenarios where an agent $i$'s vote can affect the outcome (and different strategies of $i$ lead to different expected utilities) is the pivotal case~\citep{austen1996information}, where both alternatives need exactly one vote to pass their threshold. Therefore, to show a strategy is non-weakly dominated, it suffices to show its expected utility is not dominated by any other strategy conditioned on all pivotal cases. 

\paragraph{Sequence of Environments.} An \textinst{} $\instance$  contains the agent number $\ag$, the world state prior distribution $(P_L, P_H)$, the signal distributions $(P_{hH}, P_{lH})$ and $(P_{hL}, P_{lL})$, the utility functions $\{\vt_i\}_{i=1}^\ag$, and the approximated fraction of majority agents $\alpha$. Let $\{\instance_\ag\}_{\ag=1}^{\infty}$ (or $\{\instance_\ag\}$ for short) be a sequence of \textinst{}s, where each $\instance_\ag$ is an \textinst{} with $\ag$ agents.
The \textinst{}s in a sequence share the same prior distributions, signal distributions, and approximated fractions of each type. There are no additional assumptions on the utility functions. 

In an {\em $\varepsilon$-ex-ante Bayesian $k$-strong equilibrium} ($(\varepsilon, k)$-EBSE), no group of at most $\kd$ agents can increase their utilities by more than $\varepsilon$ through deviation. A strategy profile $\stgp = (\stg_1, \stg_2,\cdots, \stg_\ag)$ is an $\varepsilon$-ex-ante Bayesian $k$-strong equilibrium if there does not exist a subset of agents $D$ (the {\em deviating group}) and a strategy profile $\stgp' = (\stg_1', \stg_2',\cdots, \stg_\ag')$ (the {\em deviating strategy profile}) such that
\begin{enumerate}
    \item $\stg_\sag = \stg_\sag'$ for all $\sag\not\in D$; 
    \item $\ut_{\sag}(\stgp') \ge \ut_{\sag}(\stgp)$ for all $\sag\in D$; and 
    \item there exists $\sag\in D$ such that $\ut_{\sag}(\stgp') > \ut_{\sag}(\stgp) +\varepsilon$. 
\end{enumerate}

By definition, when $\varepsilon=0$, the equilibrium is ex-ante Bayesian $k$-strong equilibrium where no group of at most $k$ agents can strictly increase their utilities through deviation; and when $\kd = \ag$, the equilibrium is an $\varepsilon$-ex-ante Bayesian strong equilibrium (named $\varepsilon$-strong Bayes-Nash equilibrium in~\citep{deng2024aggregation}), where no group of agents can increase there their utilities with at least one more than $\varepsilon$.

\section{Group Size Robustness of Equilibria in Antagonistic Voting}

We desire a ``good'' stable outcome in the voting with conflict agents and incomplete information: the existence of an equilibria that reach the informed decision with high probability. \citet{deng2024aggregation} gives a first characterization of such problem under the solution concept of $\varepsilon$-ex-ante Bayesian strong equilibrium. They show that such a ``good'' equilibrium exists if and only if the fraction of the majority agents exceeds a fraction 
\begin{equation*}
       \theta =\begin{cases}
           \frac12 + \frac{P_{\ell H}}{2P_{\ell L}} &\ P_{\ell L} \ge P_{\ell H},\\
           \frac12 + \frac{P_{hL}}{2P_{hH}} &\ P_{\ell L}<  P_{\ell H}.
       \end{cases}
   \end{equation*}

\begin{thm}~\citep{deng2024aggregation}
    \label{thm:deng}
    When $\alpha > \theta$, there exists a strategy profile sequence $\{\stgp_\ag\}$ such that the fidelity $\acc(\stgp_\ag)$ converges to 1, and every $\stgp_\ag$ is an $\varepsilon$-ex-ante Bayesian equilibrium with $\varepsilon$ converges to 0. When $\alpha \le \theta$, no such sequence of strategy profiles exists. 
\end{thm}

When $\alpha$ is large, Theorem~\ref{thm:deng} provides a positive guarantee on a ``good'' equilibrium. Nevertheless,  as the gap between majority and minority fractions narrows, the non-existence of a strong equilibrium becomes less informative on the stable outcome of strategic agents, especially when agents have limited capability to coordinate. By introducing the $\varepsilon$-ex-ante Bayesian $\kd$-equilibrium framework, we can conduct a more granular analysis of strategic behavior, revealing how the existence of a $\kd$-strong equilibrium that achieves the informed majority decision depends on the coordination capacity $\kd$ (or its fractional version $\xi$). 
% to do a finer characterization on the existence of such equilibrium and on the largest $\kd$ an equilibrium can reach. 

% it does not tell the existence of a reasonable robust outcome --- an equilibrium that prevent deviation from at most $\kd$ deviators. 

More precisely, we give the closed-form of the largest fraction of deviators $\xi^*$ for each $\alpha$ where there exists a sequence of $\varepsilon$-ex-ante Bayesian $\xi\cdot \ag$-strong equilibria in which (1) no agents play weakly dominated strategy, (2) the fidelity (likelihood of informed majority decision) converges to 1, and (3) $\varepsilon$ converges to 0 (asymptotically no one wishes to deviate). 

The closed form of $\xi^*(\alpha)$ has three or four segments depending on the signal distribution. These segments, ordered from highest to lowest values of $\alpha$ as Segment 1 to 4, are given descriptive names on their distinct features. The ``Flat'' Segment 1 begins at $\theta$ and maintains a constant value as $\alpha$ decreases. In the ``Vanishing'' Segment 2, $\xi^*$ decreases linearly with  $\alpha$, and its extension intersects zero at $\alpha=\frac12$. Segment 3, the ``Non-Linear'' segment, is characterized by a non-linear relationship between $\xi^*$ and $\alpha$. If Segment 3 ends before $\alpha = \frac12$, there will be a ``Tail'' Segment 4; otherwise, $\xi^*$ ends with three segments. 
% \gs{Is it possible to give the segments more discriptive names, perhaps based on the constraints that are binding in that region.  Maybe sigment 1 could be the "flat" sigment, or "vote for accept" signament, because that is the minority strategy initially.  This might make the analysis a little more intuitive than Sigment 1 and Sigment 2, etc.  }
% \gs{$\anl$ seems like an odd name.}
A special case of $\xi^*$ is when the signal is symmetric, i.e, $P_{hH} = P_{\ell L}$. Under this specific setting, both the ``Flat'' Segment 1 and the ``Non-Linear'' Segment 3 vanishes, and $\xi^*$ consists of only``Steep'' Segment 2 and ``Tail'' Segment 4 as shown in Figure~\ref{fig:8080}. The explicit form of $\xi^*$ is presented later in this section.  

\begin{thm}[Threshold $\xi^*$]
\label{thm:thresholdk}
% Let $\kd = \xi\cdot \ag$. Let $\xi^*$ be defined as follows: 
% \begin{itemize}
%     \item When $\alpha > \thdmaj$, $\xi^* = +\infty$.
%     \item When $\thdmaj \ge \alpha \ge \frac{\Delta}{1 - \theta} = \frac12 + \frac{P_{hL}}{2P_{\ell L}} = \frac{1}{2P_{\ell L}}$, $\xi^* = 1 - \thdmaj =\frac{\Delta}{2P_{\ell L}}$.
%     \item When $\frac{1}{2P_{\ell L}} > \alpha \ge {\color{blue} \anl}$, $\xi^* = \frac{\Delta\cdot (\alpha - 1/2)}{P_{hL}}$. 
%     \item {\color{blue} When $\anl > \alpha > \max (\frac12, \frac{1}{1+P_{\ell L} +(P_{\ell L} - P_{hH})})$, $\xi^* = \xinl$}
%     \item {\color{blue}When $\frac12 < \frac{1}{1+P_{\ell L} +(P_{\ell L} - P_{hH})}$}, for $\frac{1}{1+P_{\ell L} +(P_{\ell L} - P_{hH})} > \alpha \ge \frac12$, $\xi^* = \frac12 \cdot \Delta \cdot \alpha$. 
% \end{itemize}

For any $\thdmaj \ge \alpha > \frac12 $ and any $\xi < \xi^*(\alpha)$, there exists a sequence of strategy profile $\{\stgp_\ag\}$ such that,
    \begin{enumerate}
        \item For all $\ag$, no agents play weakly dominated strategies in $\stgp_\ag$,
        \item the fidelity $\acc(\stgp_\ag)$ converges to 1, and
        \item for all $\ag$, $\stgp_\ag$ is an $\varepsilon$-ex-ante Bayesian $\xi\ag$-strong equilibrium, where $\varepsilon$ converges to 0.
    \end{enumerate}
    On the other hand, when $\xi = \xi^*$, no such sequence of profile exists. 
\end{thm}

The existence of multiple segments with drastically different shapes, including a surprising and highly non-trivial non-linear segment, reveals the complexity of the strategic behaviors of the agents with incomplete information in the voting game. The solution concept $\varepsilon$-ex-ante Bayesian $\kd$-strong equilibrium provides a more refined and practical perspective to the problem, where the capacity of group strategic behaviors is represented by $\xi$. As a consequence, Theorem~\ref{thm:thresholdk} provides a more informative analysis compared to Theorem~\ref{thm:deng}. This in turn demonstrates the capability of $\varepsilon$-ex-ante Bayesian $\kd$-strong equilibrium as an informative perspective. 

\subsection{Explicit Form of the Threshold Curve.}
We assume without loss of generality that $P_{hH} \le P_{\ell L}$. In this case, $\theta = \frac12 + \frac{P_{\ell H}}{2P_{\ell L}}$. 
Now we explicitly give the formula of $\xi^*$ as a function of $\alpha$. 

\vspace{0.2cm}
\begin{minipage}[t]{0.5\textwidth}
When $\frac{1}{1+P_{\ell L} +(P_{\ell L} - P_{hH})} > \frac12$,
\begin{equation*}
    \xi^*(\alpha) = 
    \begin{cases}
        \frac{\Delta}{2P_{\ell L}} &\ \thdmaj \ge \alpha \ge \frac{1}{2P_{\ell L}}\\
         \frac{\Delta\cdot (\alpha - 1/2)}{P_{hL}} &\ \frac{1}{2P_{\ell L}} > \alpha \ge \anl\\
         \xinl &\ \anl > \alpha > \frac{1}{1+P_{\ell L} +(P_{\ell L} - P_{hH})}\\
         \frac12 \cdot \Delta \cdot \alpha &\ \frac{1}{1+P_{\ell L} +(P_{\ell L} - P_{hH})} > \alpha > \frac12
    \end{cases}
    \label{eq:xi4}
\end{equation*}
\end{minipage}
\hfill\vline\hfill
\begin{minipage}[t]{0.4\textwidth}
When $\frac{1}{1+P_{\ell L} +(P_{\ell L} - P_{hH})}\le  \frac12$,
\begin{equation*}
    \xi^*(\alpha) = 
    \begin{cases}
        \frac{\Delta}{2P_{\ell L}} &\ \thdmaj \ge \alpha \ge \frac{1}{2P_{\ell L}}\\
         \frac{\Delta\cdot (\alpha - 1/2)}{P_{hL}} &\ \frac{1}{2P_{\ell L}} > \alpha \ge \anl\\
         \xinl &\ \anl > \alpha > \frac{1}{2}
    \end{cases}
    \label{eq:xi_3}
\end{equation*}
\end{minipage} 
\vspace{0.2cm}

Here, $\xinl$ is a non-linear function on $\alpha$ which makes up the ``Non-Linear'' Segment 3 in the curve. $\anl$ is the changing point between Segment 2 and 3 and comes from solving inequality $\xinl \ge \frac{\Delta(\alpha - 1/2)}{P_{hL}}$.

\begin{align}
    \anl =&\ \frac{2P_{hH}P_{\ell L}^2 +P_{hH}\cdot (P_{hH}+3P_{\ell L})-(3P_{hH} +P_{\ell L}) + 2 - 2P_{hL}\cdot \sqrt{P_{hH}P_{hL}P_{\ell H}P_{\ell L}}}{2P_{\ell H}^2 + 8P_{hH}P_{\ell L}^2}. \label{eq:anl}\\
    \xinl = &\ \frac{\Delta}{4\sqrt{2(1 - \alpha P_{\ell H})(1 - 2 \alpha P_{\ell L})P_{\ell H}P_{\ell L}} + 2P_{\ell H} + 4P_{\ell L} - 8\alpha P_{\ell H}P_{\ell L}}. \label{eq:xinl}
\end{align}

% \begin{ex}[Threshold Curve]
%     \label{ex:thresholdk} 
%     Here we adopt the information structure in Example~\ref{ex:setting} and give a concrete threshold curve $\kd$. Recall the signal distribution $P_{hH} = 0.7$, $P_{\ell H} = 0.3$, $P_{hL} = 0.2$, and $P_{\ell L} = 0.8$. In this case, $\xi^*$ has four Segments.
%     \begin{equation*}
%     \xi^*(\alpha) = 
%     \begin{cases}
%         \frac{5}{16} &\ \frac{11}{16} \ge \alpha \ge \frac{5}{8}\\
%          \frac52(\alpha - \frac12) &\ \frac{5}{8} > \alpha \ge 0.556\\
%          \frac{1}{7.6 - 3.84 \alpha + 3.2 \sqrt{3(1 - 1.6\alpha)(1 - 0.3\alpha)}} &\ 0.556 > \alpha > \frac{10}{19}\\
%          \frac14 \alpha &\ \frac{10}{19} > \alpha > \frac12
%     \end{cases}
% \end{equation*}
% \end{ex}

% \Qishen{Writing a brief discusssion on the variety and on that does not goes to 0. }
\subsection{Shape Analysis on the Threshold Curve.}
While the closed-form of $\xi^*(\alpha)$ is highly non-trivial, we can still gain valuable insights into how the underlying signal distribution influences its structural properties. 

\textit{Accuracy.} As the signals become more accurate (i.e., $P_{hH}$ and $P_{\ell L}$ increase), two effects emerge in the equilibrium threshold $\xi^*$. First, the interval of $\xi^*$ from $\theta$ to $\frac{1}{2}$ contracts because $\theta$ is negatively correlated with signal accuracy. Second, since $\Delta = P_{hH} + P_{\ell L} - 1$ increases with signal accuracy, the equilibrium threshold $\xi^*$ increases within each segment of its piecewise function. This observation aligns with intuition: as agents receive more precise signals about the true state, the majority group can more effectively leverage their information advantage to secure their preferred outcome.

% When $P_{hH}$ and $P_{\ell L}$ increases, meaning that the signals are more accurate, the interval of $\xi^*$ from $\theta$ to $\frac12$ becomes shorter. This is because $\theta$ is negatively correlated to $P_{hH}$ and $P_{\ell L}$ On the other hand, $\Delta = P_{hH} + P_{\ell L} - 1$ is positively correlated to the $\xi^*$ value in every segment. Therefore, $\xi^*$ becomes larger when the signals are more accurate. This follows the high level idea where majority agents with more accurate signal has a clearer information on the ground truth. Therefore, they are more likely to ensure the preferred outcome. 

\textit{Bias.} As the signal bias $P_{\ell L} - P_{hH}$ increases, there are corresponding changes in the intervals of Segments 1 and 4. The length of Flat Segment 1, given by $\theta - \frac{1}{2P_{\ell L}} = \frac{P_{\ell L} - P_{hH}}{2P_{\ell L}}$, exhibits a positive linear relationship with the signal bias. Conversely, ``Tail'' Segment 4 contracts as the signal bias increases, since its left boundary is a decreasing function of $P_{\ell L} - P_{hH}$ while its right boundary remains fixed at $\frac{1}{2}$.

% When $P_{\ell L} - P_{hH}$ increases, indicating more biased signals, the interval of ``Flat'' Segment 1 widens, while the interval of Tail'' Segment 4 shrinks. The length of the ``Flat'' Segment 1 $\theta - \frac{1}{2P_{\ell L}} = \frac{P_{\ell L} - P_{hH}}{2P_{\ell L}}$ is positively correlated to $P_{\ell L} - P_{hH}$. On the other hand, the left-boundary of Segment 4 is decreases with $P_{\ell L} - P_{hH}$, and the right-boundary is fixed at $\frac12$. More biased signals result in a shorter Segment 4. 

\textit{Not Converge to 0.} No matter what the signal distribution is, $\xi^*$ converges to $\frac{1}{4}\Delta > 0$ when $\alpha$ goes to $\frac12$. This indicates that, even when the majority only slightly outnumbers the minority, an unorganized minority still yields a stable outcome favoring the 51\%.

Figure~\ref{fig:shapes} illustrates the shapes of $\xi^*$ under different distributions. For example,  Figure~\ref{fig:8080} with a symmetric signal distribution contains only Segment 2 and 4, while Figure~\ref{fig:9060} with a biased signal distribution has a much larger ``Flat'' Segment 1 and shorter Segment 4. In Figure~\ref{fig:6050} where the signal is inaccurate, the interval of $\xi^*$ is long yet the value is close to 0. 

\begin{figure}[htbp]
        \centering
        \subfigure[$P_{\ell L} = 0.8, P_{hH} = 0.8$.\label{fig:8080}]{   \includegraphics[width=0.32\textwidth]{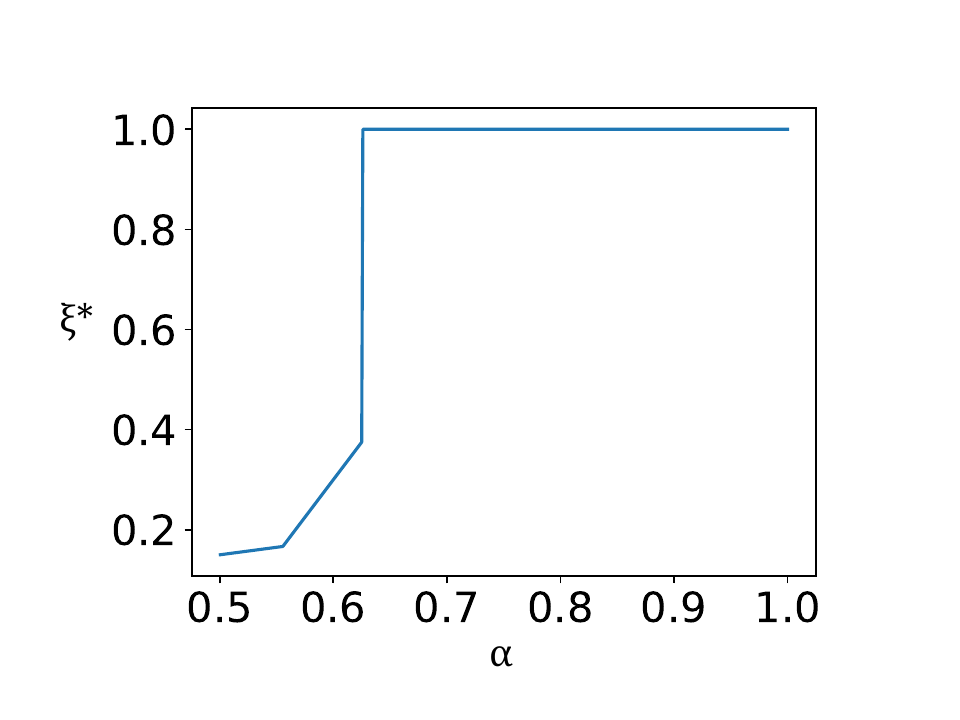}}
        % \subfigure[$P_{\ell L} = 0.8, P_{hH} = 0.7$]     {\includegraphics[width=0.475\textwidth]{figures/8070.pdf} }  
        \subfigure[$P_{\ell L} = 0.9, P_{hH} = 0.6$.\label{fig:9060}]{\includegraphics[width=0.32\textwidth]{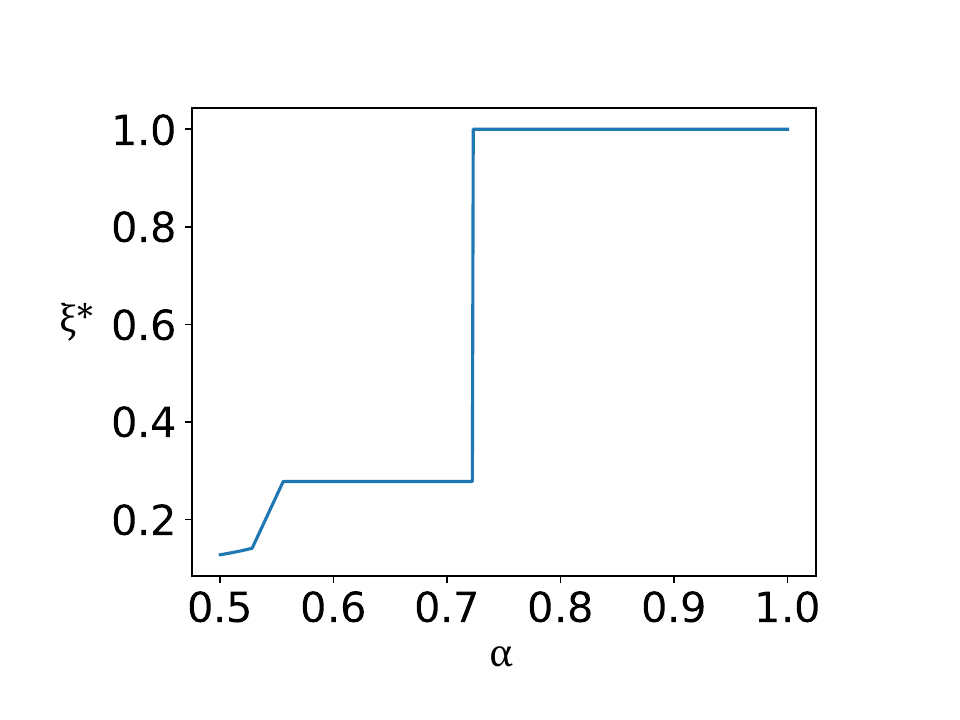} }
        \subfigure[$P_{\ell L} = 0.6, P_{hH} = 0.5$.\label{fig:6050}]{\includegraphics[width=0.32\textwidth]{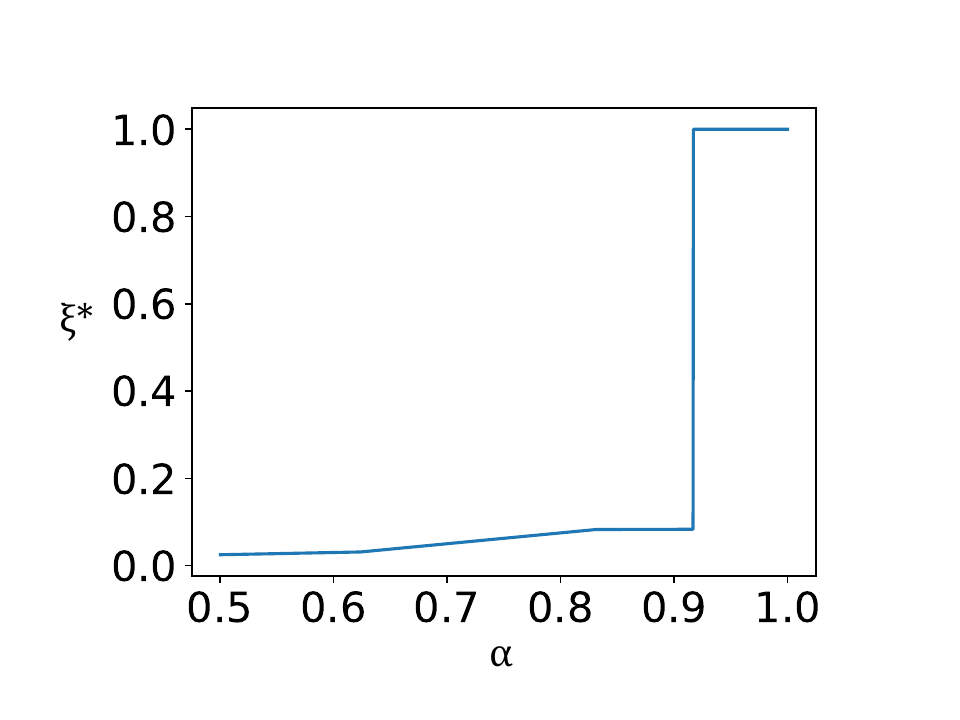}}
        \caption
        {The shape of $\xi^*$ under different signal distribution. We set $\xi^* = 1$ when $\alpha > \theta$, indicating the existence of a strong equilibrium. \label{fig:shapes}} 
\end{figure}
\vspace{-0.4cm}

\section{Proof Sketch}

In total, four segments of $\alpha$ are characterized in Theorem~\ref{thm:thresholdk}. For each threshold, we need to prove both the upper bound and the lower bound. For the upper bound, we show that any sequence of strategy profiles satisfying the first (non-dominated strategy) and the second (fidelity converges to 1) violates the third constraint (being an equilibrium).
For a lower bound, we construct a sequence of strategy profiles satisfying all three constraints. 

Before analyzing individual segments, we establish that only minority agents need to be considered in potential deviations. This assertion is justified because we focus on strategy profile sequences that satisfy constraint 2, where fidelity converges to 1. Under these conditions, any deviating strategy profile sequence must follow one of two scenarios. In the first scenario, if the deviating strategy's fidelity converges to 1, the utility change from deviation becomes negligible and approaches 0, rendering the deviation ineffective. In the second scenario, if the deviating strategy's fidelity fails to converge to 1, majority agents experience strictly lower utility and thus have no incentive to join the deviating group. Consequently, we need only examine deviations involving minority agents, whose utilities are negatively correlated with fidelity. Such a deviation succeeds (and consequently, the original strategy profile sequence fails to be a sequence of $o(1)$-equilibrium) if and only if the fidelity of the deviating strategy does not converge to 1.

Moreover, we leverage the expected vote share to determine the fidelity of the strategy profiles, especially the deviating ones. By extending the results in~\citep{han2023wisdom}, with a few exceptions, the fidelity of the strategy profile sequence converges to 1 if and only if $\thd_H > \frac12$ and $\thd_L < \frac12$ holds for all sufficiently large $\ag$. Specifically, the expected vote shares have the following form. 
\begin{align*}
    \thd_H =&\ \frac{1}{\ag} \cdot \sum_{i = 1}^{\ag} (P_{hH}\cdot \bph^i + P_{\ell H}\cdot \bpl^i). \\
    \thd_L =&\ \frac{1}{\ag} \cdot \sum_{i = 1}^{\ag} (P_{hL}\cdot \bph^i + P_{\ell L}\cdot \bpl^i). 
\end{align*}
% \gs{I changed stuff before before remembering to put on track changes.  Please check.}

The deviators aim to achieve one of two outcomes after a $\xi$ fraction of agents deviate to create a new strategy profile (with expected vote share $\thd_H'$ and $\thd_L'$): either  \begin{equation}\label{eq:H-fail}\thd_H' \le \frac12,\end{equation} leading to the expected majority being $\bR$ in state $H$, or \begin{equation}\label{eq:L-fail}\thd_L'\ge \frac12, \end{equation} leading to the expected majority being $\bA$ in state $L$.
% Thus, the goal of the deviators is, after a $\xi$ fraction of agents deviate, the expected vote share of the new strategy profile (denoted as $\thd_H'$ and $\thd_L'$) achieves either \begin{equation}\label{eq:H-fail}\thd_H' \le \frac12,\end{equation} in which case when the world state is $H$  then, in expectation, the result is $\bR$,  or \begin{equation}\label{eq:L-fail}\thd_L'\ge \frac12, \end{equation} in which case when the world state is $L$. then, in expectation, the result is $\bA$. 
Strategically, to maximize their chances,  deviators should embrace one of two extremes. In the first extreme the goal is to maximize,  $\thd_L'$, the expected fraction of $\bA$ votes in world $L$ after deviation. Here the deviators switch to always vote for $\bA$, and consist of the minority agents who are initially least likely to vote for $\bA$.  In the second extreme the goal is to minimize,  $\thd_H'$, the expected fraction of $\bA$ votes in world $H$ after deviation. Here the deviators switch to always vote for $\bR$, and consist of the minority agents who are initially least likely to vote for $\bR$.  If either deviation succeeds, they succeed; and if both fail, then no deviation will succeed; implying the original strategy profile is an equilibrium. 

In general, there are three correlated factors that may affect how large a deviating group, $\xi$, a particular strategy profile can tolerate. 
\begin{itemize}\item The first factor is the difference between $\thd_H$ and $\thd_L$, which is determined by the quantity $\bph^i - \bpl^i$ of each agent $i$. When $\thd_H - \thd_L$ is large, it is at least possible for the threshold of 1/2 to be between the two values and not too close to either of them. 
\item Secondly, the strategies of the minority agents.  Agents ability to change $\thd_H$ and $\thd_F$ depends on their initial strategy.  For example, an agent who always votes for $\bA$ cannot help deviate to increase $\thd_L$ by can deviate to $\thd_H$ by $\frac1{\ag}$. Meanwhile, an agent who votes for $\bA$ with signal $\ell$ and $\bR$ with signal $h$ can deviate to either increase $\thd_L$ by $\frac{P_{hL}}{\ag}$ or to decrease $\thd_H$ by $\frac{P_{\ell H}}{\ag}$.
\item
The last factor is balance.  Because the deviators have two possible directions to deviate (decrease $\thd_H$ and increase $\thd_L$), a deviation-proof strategy profile needs to keep both margins $\thd_H - \frac12$ and $\frac12 - \thd_L$ sufficiently large to overwhelm the aforementioned deviating power of minority agents.  Sometimes deviations can by neutralized by keeping margin balanced even at the expense of decreasing the difference $\thd_H - \thd_L$.   
\end{itemize}
The interplay of these three factors largely contributes to the complexity of the problem. In fact, as we will see in the following proof sketch, in the equilibria establishing lower bounds for Segments 1 and 2, all minority agents unanimously vote for alternative $\bA$. In contrast, the equilibrium determining the lower bound for Segment 4 exhibits an even split of minority votes between both alternatives.

The rest of the proof proceeds as follows. We will start from the first segment $\thdmaj \ge \alpha \ge \frac{1}{2P_{\ell L}}$ whose proof is somewhat simpler and independent but shares the high-level idea with other parts. Then, we present a unified analysis of Segments 2, 3, and 4. The key step is converting the upper bound and the lower bound problem into constraints on the expected vote share which is further transformed into an optimization problem where the goal is to maximize the size of the deviation $\xi$ that can be tolerated.  The parameters that yield the optimum can be leveraged to construct a sequence of equilibrium that proves the lower bound of the theorem.  By case analysis, we can determine which constraints are binding in different settings and also construct the upper bounds.  
% \gs{Is last sentence clearer than before?}  %ing which Finally, we solve the maximum problem in which only some border cases needs to be considered. 

\subsection{Segment 1: $\thdmaj \ge \alpha \ge \frac{1}{2P_{\ell L}}$. }
We first introduce the strategy profile achieving the lower bound and explain why the threshold $\xi^*$ is constant in this region. Then we present the upper bound proof which serves as a simplified version of that for other segments. 

\begin{prop}
\label{prop:seg2lower}
    For all $\alpha \ge \frac{\infdif}{1 - \thdmaj}$, for any $\xi < \frac{\Delta}{2P_{\ell L}}$, there exists a strategy profile sequence that satisfies all three constraints of the theorem. 
\end{prop}

In this strategy profile sequence, all the minority agents always vote for $\bA$. recall the strategy of majority agents need to guarantee that both Equation~\ref{eq:H-fail} and \ref{eq:L-fail} fail to hold. Note that these minority agents cannot increase $\thd_L$ because they always vote for $\bA$. On the other hand, they can decrease $\thd_H$ by $\xi$ by switching to always vote for $\bR$. Therefore, the strategies of the majority agents must satisfy $\thd_H > \frac12 + \xi$ and $\thd_L < \frac12$ to defend the possible deviations. 

Without loss of generality, we consider a strategy profile where all the majority agents play the same strategy, denoted as $(\bph^{\maj}, \bpl^{\maj})$. In this case, the expected vote share is 
\begin{align*}
    \thd_H =&\ \alpha \cdot (P_{hH}\cdot \bph^{\maj} + P_{\ell H}\cdot \bpl^{\maj}) + (1 - \alpha). \\
    \thd_L =&\ \alpha + \cdot (P_{hL}\cdot \bph^{\maj} + P_{\ell L}\cdot \bpl^{\maj}) + (1 - \alpha). 
\end{align*}

We want to find a $\xi$ as large as possible to satisfying the following two constraints. Firstly, $\bph^{\maj}$ and $\bpl^{\maj}$ should be in $[0, 1]$. Secondly, $\thd_H > \frac12 + \xi$ and $\thd_L < \frac12$ holds. 

The important observation that help us solve the problem is that $\thd_H$ is more sensitive on $\bph^{\maj}$, and $\thd_L$ is more sensitive on $\bpl^{\maj}$. This is because $P_{hH} > P_{hL}$ and $P_{\ell H} < P_{\ell L}$. Therefore, if we could find $\bph^*$ and $\bpl^*$ in $(0, 1)$ that satisfy $\{\thd_H = \frac12 + \xi, \thd_L = \frac12\}$, we can get a valid pair of $(\bph^{\maj}, \bpl^{\maj})$ by increasing $\bph^*$ and decrease $\bpl^*$ at a certain ratio so that $\thd_H$ increases and $\thd_L$ decreases. The modified $(\bph^{\maj}, \bpl^{\maj})$ will satisfy the two constraints. 

% Two solve the two constraints, we first convert the second constraint to a set of equations $\{\thd_H = \frac12 + \xi, \thd_L = \frac12\}$  and solve $(\bph^{\maj}, \bpl^{\maj})$ on the equation set. We denote the solution $\bph^*$ and $\bpl^*$. We claim that, we can find a valid strategy $(\bph^{\maj}, \bpl^{\maj})$ if and only if the solution $\bph^*$ and $\bpl^*$ are in $(0, 1)$. Suppose $\bph^*$ and $\bpl^*$ are in $(0, 1)$ holds, we construct a valid strategy 

% To construct valid strategies for the majority agents, we perform the following two steps. First, we assume all the majority agents play the same strategy and solve the equation set $\{\thd_H = \frac12 + \xi, \thd_L = \frac12\}$. We denote the solution $\bph^*$ and $\bpl^*$ and equations for the are given below.  \gs{These equations are derived by...fill in here.} \gs{I have absolutely 0 clue where these come from.  Complete mystery.}.
Let $\bph^*$ and $\bpl^*$ be the solution to the equation set $\{\thd_H = \frac12 + \xi, \thd_L = \frac12\}$ with the explicit form.
\begin{align*}
        \bpl^*  =&\ \frac1{\alpha} (\alpha - \frac12 - \frac{P_{hL}}{\infdif}\cdot \xi).\\
        \bph^*  =&\ \frac1{\alpha} (\alpha - \frac12 + \frac{P_{\ell L}}{\infdif}\cdot \xi).
\end{align*}
Examining the above equations, when $\xi < \frac{\Delta}{2P_{\ell L}}$ we can easily see that both $\bpl^*$ and $\bph^*$ are in $(0, 1)$. In the second step, we simultaneously increase $\bph^*$ at rate $P_{\ell H} + P_{\ell L} $and decrease $\bpl^*$ at rate $(P_{h H} + P_{h L})$ until the former reaches 1 or the latter reaches 0, at which point we stop. The strategy after the second step is 
$\bpl^{\maj} = \bpl^* - \delta \cdot (P_{h H} + P_{h L})$ and $\bph^{\maj} = \bph^* + \delta \cdot (P_{\ell H} + P_{\ell L})$, where $\delta > 0$ is chosen so that both are between 0 and 1 and one is at the boundary. The strategy of every majority agent is set to $(\bph^{\maj}, \bpl^{\maj})$.  

This implies that all the three constraints hold. 
Because either $\bph^{\maj} =1$ or $\bpl^{\maj} = 0$ holds, all the majority agents are playing their non-dominated strategy. 
The two-step construction guarantees that $\thd_H > \frac12 + \xi$ and $\thd_L < \frac12$.  Thus, the fidelity of the sequence converges to 1.
Moreover,  because, with high probability, a $\xi$ fraction of minority agents cannot change the outcome, it is an $\epsilon$-ex-ante Bayesian $\xi n$-strong equilibrium, where $\epsilon$ converges to 0.  

It seems a little counter-intuitive that $\xi$ is constant even while $\alpha$ varies in segment 1.  
Here we try to give some high-level explanation for this behavior. 
% \gs{I cannot follow this at all (what is below).  There might be a typo.  I never did see what $\Delta$ was.  But I cannot fix it because I don't understand it.} 
At a high level, this is a trade-off between the difference $\thd_H - \thd_L$ (the first factor that affects $\xi$) and the balance (the third factor). When $\alpha$ is large, the margin of $\thd_H$ is smaller than that of $\thd_L$, and the majority agents' votes have to lean to $\bA$, which sacrifices $\thd_H - \thd_L$ to keep the balance. As $\alpha$ decreases, there are more minority agents always voting for $\bA$, and the margins of $\thd_H$ and of $\thd_L$ become more balanced. While the fraction of the majority agents decreases, more balanced margins allow them to play a strategy that leads to a larger $\thd_H - \thd_L$. This trade-off keeps the $\xi$ at a flat constant. Eventually, when $\alpha = \frac{1}{2P_{\ell L}}$, the margins are completely balanced, and majority agents play informatively which maximizes $\thd_H - \thd_L$. When $\alpha$ further decrease, the margin on $\thd_L$ becomes larger than $\thd_L$, and majority agents cannot increases $\thd_H - \thd_L$ anymore. Therefore, the threshold $\xi^*$ begins to decrease. 

% When $\xi$ is close to the threshold $\frac{\Delta}{2P_{\ell L}}$, there will be $\bph^{\maj} = 1$ and $\bph^{\maj} > 0$. As $\alpha$ decreases, both $\bph^*$ and $\bpl^*$ decreases. so $\bpl^*$ becomes more and more close to 0. While the fraction of the majority agents decreases, the increase of $(\bph^{\maj} - \bpl^{\maj})$ enlarge everyone's contribution to $\thd_H - \thd_L$, so the threshold $\xi^*$ remains flat along this segment. 
% Eventually, when $\alpha = \frac{1}{2P_{\ell L}}$. $\bph^{\maj} = 0$, and majority agents vote informatively. For $\alpha \le \frac{1}{2P_{\ell L}}$, one majority agent cannot contribute more to $\thd_H - \thd_L$ than informative voting, so the threshold $\xi^*$ begins to decrease. 

Now we sketch the proof of the upper bound. We take an arbitrary strategy profile sequence satisfying the first two constraints and show that if $\xi \geq  \frac{\Delta}{2P_{\ell L}}$, constraint 3 is violated. 

\begin{prop}
\label{prop:seg2upper}
    For all $\alpha \le \thdmaj$, for any profile sequence $\{\stgp_\ag\}$ such that $\acc(\stgp_\ag)$ converges to 1 and no agents play weakly dominated strategies, there is a constant $\varepsilon > 0$ and infinitely many $\ag$ such that $\stgp_\ag$ is NOT an $\varepsilon$-ex-ante Bayesian $\xi\cdot \ag$-strong equilibrium, where $\xi = \frac{\Delta}{2P_{\ell L}}$. 
\end{prop}

% \gs{It is seems that the above statement is not correct.  It should be true for infinitly many $n$.  Otherwise, it could be true for 0 (which is a constant) $n$. Fixed}

For each environment with $\ag$ agents, we fix the deviating group $D$ with $\xi = \frac{\Delta}{2P_{\ell L}}$ minority agents and show that no matter what the strategy profile is, the deviator can always succeed by decreasing $\thd_H$ or increasing $\thd_L$ to $\frac12$. This is a stronger statement than the existence of a successful deviation, as it fixes the deviating group in advance. 

We first consider the strategies of the non-deviators, including all the majority agents and some minority agents. Since they are not going to deviate, their contribution to $\thd_H$ and $\thd_L$ can be represented by the average of all their strategies, denoted by $\abph$ and $\abpl$. Likewise, if the deviation fails, the expected vote share under state $H$ must be greater than $\frac12$ even when all the deviators vote for $\bR$, and the expected vote share under state $L$ must be smaller than $\frac12$ even when all the deviators vote for $\bA$. Formally, this leads to the following constraints:
\begin{align*}
    \thd_H' :=&\ (1 - \xi) \cdot (\abph \cdot P_{hH} + \abpl \cdot P_{\ell H}) > \frac12.\\
    \thd_L' := &\ (1-\xi) \cdot (\abph \cdot P_{hL} + \abpl \cdot P_{\ell L}) + \xi < \frac12. 
\end{align*}

Therefore, to show the deviation will succeed, it suffice to show that for any $(\abpl, \abph)$, at least one of the inequalities is violated. 
% As before, we solve the equation set $\{\thd_H ' = \frac12, \thd_L' = \frac12\}$. But this time, we 
As before, we utilize the observation that $\thd_H'$ is more sensitive to $\abph$ and $\thd_L'$ is a more sensitive to $\abpl$. First, we solve the equation set $\{\thd_H ' = \frac12, \thd_L' = \frac12\}$ and get the solution $\abphs$ and $\abpls$. However, when $\xi = \frac{\Delta}{2P_{\ell L}}$, the solution $\abphs \ge 1$. This difference changes the outcome into the opposite direction, as we will show that one of the inequalities is violated.  

We consider an arbitrary pair of $\abph$ and $\abpl$ in $[0, 1]$. Note that $\abph \le \abphs$. Therefore, if $\abpl$ is not large enough, then $\thd_H' > \frac12$ is violated. On the other hand, if $\abpl$ is large enough so that $\thd_H' > \frac12$ holds, it turns out that it is too large. $\thd_L'$ is more sensitive to $\abpl$, so a large $\abpl$ affects $\thd_L'$ more than it affects $\abph$. Therefore, when $\thd_H' > \frac12$, $\thd_L' < \frac 12$ will be violated. 
This implies that one of the two deviations can succeed, and $\stgp$ fail to be an equilibrium. 
% \gs{I am completely lost below this point.} \Qishen{Rewrite here.}
% Then for any $(\abph, \abpl)$, we rewrite $\thd_H'$ and $\thd_L'$ in the form of the solution $\abphs$ and $\abpls$.

%  \begin{align*}
%     \thd_H' =&\ \frac12 + (1 - \xi) \cdot ((\abph - \abphs) \cdot P_{hH} + (\abpl - \abpls) \cdot P_{\ell H}).\\
%     \thd_L' = &\ \frac12 + (1-\xi) \cdot ((\abph - \abphs) \cdot P_{hL} + (\abpl - \abpls) \cdot P_{\ell L}). 
% \end{align*}

% When $\xi = \frac{\Delta}{2P_{\ell L}}$, the solution $\abphs \ge 1$. Therefore, $(\abph - \abphs) \le 0$. If $\thd_H' > \frac12$, $(\abpl - \abpls) > 0$ must hold. Moreover, there must be $(\abph - \abphs) \cdot P_{hH} + (\abpl - \abpls) \cdot P_{\ell H} > 0$. Then we turn to $\thd_L'$. By $P_{\ell L} > P_{\ell H}$ and $P_{hH} > P_{hL}$, the positive term $(\abpl - \abpls)$ in $\thd_L'$ is amplified while the negative term $(\abph - \abphs)$ shrinks. Therefore, $\thd_L' > \frac12$ must hold. This implies that one of the two deviations succeeds, and $\stgp$ fail to be an equilibrium. 

\subsection{Segment 2, 3, and 4, $\frac{1}{2P_{\ell L}} \ge \alpha > \frac12$.}
The proof of the rest of the three segments follows the same high-level idea and shares some preliminary steps. According to Lemma~\ref{lem:dominated}, we can further categorize minority agents into to two types: Type-0 agents who always vote for $\bR$ when they receive $h$ signal, and type-1 agents who always vote for $\bA$ when they receive $\ell$ signal. While a minority agent can both always vote for $\bR$ when receiving $h$ signal and always vote for $\bA$ when receiving $\ell$ signal, in our analysis, we will always partition the agents so that each is assigned  to exactly one type. The fraction of type-0 agents will be denoted by $\gamma$, and the fraction of the type-1 agents  then $1 - \alpha - \gamma$. 

This part of proof consists of five steps. 
\begin{enumerate}
    \item We solve the case when $\gamma$ is biased (either $\gamma$  or $1- \alpha - \gamma$ is close to 0) and restrict the case when $\gamma$ is not biased to be ``double-symmetric'': all type-0 agents play the same strategy, and all type-1 agents play the same strategy.
    \item We show that a sequence of strategy profiles satisfies all three constraints if and only if $\xi< \xi_h$ and $\xi< \xi_\ell$. Here $\xi_h$ and $\xi_\ell$ are two thresholds derived by the constraints on the expected vote share ($\thd_H' > \frac12$ and $\thd_L < \frac 12$, like that in the ``Flat'' Segment 1). Therefore, the problem of finding the upper and lower bound is converted into maximizing $\min(\xi_h, \xi_\ell)$.
    \item We reduce the search space of maximizing $\min(\xi_h, \xi_\ell)$ into two corner cases. In the first case, all type-0 agents always vote for $\bR$. In the second case, all type-1 agents always vote for $\bA$. The threshold $\xi^*$ takes the maximum of these two cases and $\frac{\Delta(\alpha - \frac12)}{P_{hL}}$, which derives from case with biased $\gamma$.
    % \item  We show that $\min(\xi_h, \xi_\ell)$ is maximized in two extreme cases: all type-0 agents always vote for $\bR$, or all type-1 agents always vote for $\bA$.
    \item We tackle the two extreme cases respectively. In the case where all type-0 agents always vote for $\bR$ (the first case), the extension of ``Steep'' Segment 2 and ``Tail'' Segment 4 serves as an upper bound. On the other hand, the case where all type-1 agents always vote for $\bA$ (the second case) derives the ``Non-Linear'' Segment 3 and the ``Tail'' Segment 4. 
    \item We take the maximum of all the cases and reach the upper and the lower bounds. 
\end{enumerate}

% This part of proof consists of five steps. (1) We solve the case when $\gamma$ is biased and restrict the case when $\gamma$ is not biased to be ``double-symmetric'': all type-0 agents play the same strategy, and all type-1 agents play the same strategy. (2) We show that a sequence of strategy profiles satisfy all constraint on $\xi$ if and only if $\xi< \xi_h$ and $\xi< \xi_\ell$. Therefore, the problem of finding the upper and lower bound is converted into maximizing $\min(\xi_h, \xi_\ell)$. (3) We show that $\min(\xi_h, \xi_\ell)$ is maximized on two border cases: all type-0 agents always vote for $\bR$, or all type-1 agents always vote for $\bA$. (4) We tackle the two border cases respectively. (5) We take the maximum of all the cases and reach the upper and the lower bounds. 

% \gs{I stopped here.}
\subsubsection{Preliminaries}
We first consider the case when either $\gamma$ is close to 0 or $1- \alpha - \gamma$ is close to 0, which means that one type of minority agents dominates the minority group. In this case, the $\xi$ fraction of deviators shall exhaust the type with fewer agents and fill up the rest fraction with agents from the other type. Following a similar reasoning of that in Proposition~\ref{prop:seg2upper} (pre-fixing a group of deviators and showing the at least one of the inequality is violated), we can get the upper bound for $\xi$ of this case is $\frac{\Delta\dot (\alpha - 1/2)}{P_{hL}}$. 

\begin{prop}
\label{prop:biasedgamma_gen}
    For $\frac{1}{2P_{\ell L}} \ge \alpha >\frac12$, for any $\xi' \ge \frac{\Delta \cdot (\alpha - 1/2)}{P_{hL}}$, NO sequence with infinitely many strategy profiles where $\gamma \le \xi'$ or $1 - \alpha - \gamma \le \xi'$ satisfy all three constraints in Theorem~\ref{thm:thresholdk} with $\xi = \xi'$. 
\end{prop}

Moreover, $\frac{\Delta\dot (\alpha - 1/2)}{P_{hL}}$ also serves a lower bound for $\xi$. 

\begin{prop}
    \label{prop:seg3lower}
    When $\frac{1}{2P_{\ell L}} \ge \alpha >\frac12$, for any $\xi < \frac{\Delta(\alpha - 1/2)}{P_{hL}}$, there exists a sequence of strategy profiles $\{\stgp\}$ that satisfies the constraints in Theorem~\ref{thm:thresholdk} with $\xi  = \frac{\Delta(\alpha - 1/2)}{P_{hL}}$
\end{prop}

The strategy profiles in Proposition~\ref{prop:seg3lower} resembles that in Proposition~\ref{prop:seg2lower}. All the minority agents always vote for $\bA$, and the majority agents play the same strategy constructed by solving equation on the expected vote share and modifying. The proof also resembles Proposition~\ref{prop:seg2lower}.

From now on, we focus on the strategy profile sequences where type-0 and type-1 agents are not biased, i,e., $\gamma > \xi$ and $1 - \alpha - \gamma > \xi$ holds for all sufficiently large $\ag$. We assume this assumption holds unless otherwise stated.

Before we start analyzing the expected vote shares and deviations with type-0 and type-1 agents, we further narrow the strategy profiles to be considered to those where all type-0 agents play the same strategy and all type-1 agents play the same strategy, with the following lemma.  

\begin{lem}
    \label{lem:double}
    If a strategy profile sequence $\{\stgp_\ag\}$ satisfies all three constraints on $\xi$, then the following strategy profile $\{\stgp_\ag'\}$ all satisfies all three constraints on $\xi$. For each $\ag$, (1) each majority agent in $\stgp'_\ag$ play the same strategy as in $\stgp_\ag$; (2) all the type-1 minority agents in $\stgp'_\ag$ play the (arithmetic) average strategy of type-1 minority agents in $\stgp_\ag$; and (3) all the type-0 minority agents in $\stgp'_\ag$ play the average strategy of type-0 minority agents in $\stgp_\ag$. 
\end{lem}

The high level idea of Lemma~\ref{lem:double} is that $\stgp_\ag'$ shares the equal expected vote with $\stgp_\ag$ (so as the same margin on $\thd_H$ and $\thd_L$), but the minority agents have weaker deviating power. For example, when the deviation aims to decrease $\thd_H$, the deviators (switching to always vote for $\bR$) in $\stgp_\ag$ can be the type-1 agents with largest $\bph^i$. This brings a larger decrease to $\thd_H$ compared to a $\xi$ fraction of type-1 agents who play the average $\bph$ among all type-1 agents. Therefore, if deviations cannot succeed in $\stgp_\ag$, it cannot succeed in $\stgp_\ag'$ as well. 

\subsubsection{Transition to maximization problem.}
In this part, we present the core steps in our proof. By solving the inequalities of the deviating expected vote share, we convert the problem of finding upper and lower bounds into maximizing the minimum of two upper bounds of $\xi$, denoted as $\xi_h$ and $\xi_\ell$, on the strategies and fractions of type-0 and type-1 agents. 
% In different segments of $\alpha$, the maximum locates at different strategy profiles. 

We first introduce the notations. Consider a strategy profile $\stgp$ where no agents play weakly dominated strategies. Let $\bph^{\maj}$  and $\bpl^{\maj}$ be the average strategy of the majority agents. The strategy that type-0 agents play (recall that we can assume that they play the same strategy without loss of generality) is $(\bpl^0, \bph^0 = 0)$. The strategy that type-1 agents play is $(\bpl^1 = 1, \bph^1)$. 
In this way, the expected vote share of the strategy profile can be written as 
    \begin{align*}
        \thd_H =&\ \alpha\cdot (P_{hH}\cdot \bph^{\maj} + P_{\ell H}\cdot \bpl^{\maj}) +  (1 - \alpha - \gamma)\cdot (P_{hH}\cdot \bph^{1} + P_{\ell H}) + \gamma \cdot P_{\ell H}\cdot \bpl^0.\\
        \thd_L =&\ \alpha\cdot (P_{hL}\cdot \bph^{\maj} + P_{\ell L}\cdot \bpl^{\maj}) +  (1 - \alpha - \gamma)\cdot (P_{hL}\cdot \bph^{1} + P_{\ell L}) + \gamma \cdot P_{\ell L}\cdot \bpl^0.
    \end{align*}

$\stgp$ is an $\varepsilon$-Bayesian $\xi \ag$-strong equilibrium if and only if no deviation succeed. As said, we can consider only two extremity deviations that minimizes/maximizes the deviating expected vote share respectively. In the first deviation that minimizes $\thd_L$, a $\xi$ fraction of type-1 agents switch to always vote for $\bR$. The expected vote share under $H$ state after deviation $\thd_H'$ is required to be larger than $\frac12$ to successfully defend the deviation. In the second deviation that maximizes $\thd_H$, a $\xi$ fraction of type-0 agents switch to always vote for $\bA$. The expected vote share under $L$ state after deviation $\thd_L'$ is required to be smaller than $\frac12$. Therefore, we have the following constraints. 
\begin{align}
        \thd'_H := \alpha\cdot (P_{hH}\cdot \bph^{\maj} + P_{\ell H}\cdot \bpl^{\maj}) +  (1 - \alpha - \gamma - \xi)\cdot (P_{hH}\cdot \bph^{1} + P_{\ell H}) + \gamma \cdot P_{\ell H}\cdot \bpl^0 > &\ \frac12\label{eq:devthdH}\\
        \thd'_L := \alpha\cdot (P_{hL}\cdot \bph^{\maj} + P_{\ell L}\cdot \bpl^{\maj}) +  (1 - \alpha - \gamma)\cdot (P_{hL}\cdot \bph^{1} + P_{\ell L}) + (\gamma - \xi) \cdot P_{\ell L}\cdot \bpl^0 +\xi <&\ \frac12\label{eq:devthdL}
    \end{align}
We claim that, for any given $\bph^1$, $\bpl^0$, $\gamma$, and $\xi$, there is a pair of $(\bpl^\maj, \bph^\maj)$ satisfying this condition if and only if the solution to the equation set $\{\thd'_H = \frac12, \thd'_L = \frac12\}$ (denoted as $(\bpl^*, \bph^*)$) is in $(0, 1)$. The proof of the claim shares the same idea with the proof of Proposition~\ref{prop:seg2lower} and~\ref{prop:seg2upper} as we illustrated. 

Specifically, the explicit form of the solution is 
    \begin{align*}
        \bpl^* =&\ \frac{1}{\alpha} (\frac12 - (1-\alpha)  + \gamma \cdot ( 1 - \bpl^0)) - \frac{\xi}{\Delta\cdot \alpha}(P_{hH}\cdot P_{hL}\cdot \bph^1 + P_{hH}\cdot P_{\ell L} \cdot (1 - \bpl^0) + P_{hL}).\\
        \bph^* = &\ \frac{1}{\alpha} (\frac12 - (1 - \alpha -\gamma) \cdot \bph^1) + \frac{\xi}{\Delta\cdot \alpha}(P_{hH}\cdot P_{\ell L}\cdot \bph^1 + P_{\ell H}\cdot P_{\ell L} \cdot (1 - \bpl^0) + P_{\ell H}).
    \end{align*}
Solving the constraint $\bph^* < 1$ and $\bpl^* > 0$ on $\xi$ respectively, we get two thresholds on $\xi$, denoted as $\xi_\ell$ and $\xi_h$. 
    \begin{align}
        \xi < &\ \xi_h :=  \frac{\Delta(\alpha - 1/2 + (1 - \alpha - \gamma) \cdot \bph^1)}{P_{hH}\cdot P_{\ell L}\cdot \bph^1 + P_{\ell H}\cdot P_{\ell L} \cdot (1 - \bpl^0) + P_{\ell H}}.\label{eq:xih}\\
        \xi < &\ \xi_{\ell} := \frac{\Delta (\alpha - 1/2 + \gamma \cdot (1 - \bpl^0))}{P_{hH}\cdot P_{hL}\cdot \bph^1 + P_{hH}\cdot P_{\ell L} \cdot (1 - \bpl^0) + P_{hL}}. \label{eq:xil}
    \end{align}

Specifically, $\xi_h(0, 0) = \frac{\Delta(\alpha - 1/2)}{P_{\ell H}}$ and $\xi_{\ell}(0, 0) = \frac{\Delta(\alpha - 1/2)}{P_{hL}}$ do not depends on $\gamma$. 

The following two lemmas guarantee the relationship between bounding $\xi$ and maximizing $\min(\xi_h, \xi_\ell)$. Lemma~\ref{lem:xilower} guarantees a strategy profile sequence satisfying all constraints for every $\min(\xi_h, \xi_\ell)$ that can be reached by a set of $\bph^1, \bpl^0$, and $\gamma$. Lemma~\ref{lem:xiupper} guarantees that the upper bound of $\xi^*$ is (partially) determined by the maximum of $\min(\xi_h, \xi_\ell)$. 

\begin{lem}
    \label{lem:xilower}
    Given a set of $\bph^1, \bpl^0$, and $\gamma$ satisfying $\bph^1 \in [0,1]$, $\bpl^0 \in [0, 1]$, and let $\xi' = min(\xi_h, \xi_\ell)$ corresponding to $\bph^1, \bpl^0$, and $\gamma$. If $\gamma > \xi'$ and $1 - \alpha - \gamma >\xi'$ holds, then for all $\xi < \xi'$ there exists a sequence of profiles such that the constraints in Theorem~\ref{thm:thresholdk} holds for $\xi$. 
\end{lem}

\begin{lem}
    \label{lem:xiupper}
    Let $\xi' = \max_{\bph^1, \bpl^0, \gamma} \min(\xi_h, \xi_\ell)$. No sequence strategy profile can satisfy all three constraints in Theorem~\ref{thm:thresholdk} with $\xi = \max(\xi', \frac{\Delta(\alpha - 1/2)}{P_{hL}})$. 
\end{lem}

The proof of Lemma~\ref{lem:xilower} resembles that of Proposition~\ref{prop:seg2lower} (except for the more complicated constraints), and the proof of Lemma~\ref{lem:xiupper} resembles that of Proposition~\ref{prop:seg2upper}. Both proof shares the idea that the constraints $\{\thd_H' > \frac12, \thd_L' < \frac12\}$ holds if and only if the equation set $\{ \thd_h' = \frac12, \thd_L' = \frac12\}$ has a solution in $(0, 1)^2$. The reason that Lemma~\ref{lem:xiupper} takes a maximum on $\xi'$ and $ \frac{\Delta(\alpha - 1/2)}{P_{hL}}$ is that $\xi_h$ and $\xi_\ell$ only characterize the scenario where $\gamma > \xi'$ and $1 - \alpha - \gamma >\xi'$. When $\gamma \le \xi'$ or $1 - \alpha - \gamma \le \xi'$, Proposition~\ref{prop:biasedgamma_gen} gives the upper bound $ \frac{\Delta(\alpha - 1/2)}{P_{hL}}$. 

\subsubsection{Maximizing $\min(\xi_h, \xi_\ell)$. }
From now on, we start to maximize $\min(\xi_h, \xi_\ell)$ on $\bph^1$, $1 - \bpl^0$, and $\gamma$. We find that $\min(\xi_h, \xi_\ell)$ exceeds $\frac{\Delta (\alpha - 1/2)}{P_{hL}}$ only if one of the following condition holds. (1) all the type-1 agents always vote for $\bA$ ($\bph^1 = 1$); or (2) all the type-0 agents always vote for $\bR$ ($1 - \bpl^0 = 1$). 

\begin{prop}
    \label{prop:border_cases}
    At least one of the following holds. (1) $\max_{\bph^1, \bpl^0, \gamma} \min(\xi_h, \xi_\ell) \le \frac{\Delta(\alpha - 1/2)}{P_{hL}}$; (2) $\min(\xi_h, \xi_\ell)$ is maximized when $\bph^1 = 1$, and (3) $\min(\xi_h, \xi_\ell)$ is maximized when $1 - \bpl^0 = 1$. 
\end{prop}

Proposition~\ref{prop:border_cases} largely decrease the space we need to search for the maximum of $\min(\xi_h, \xi_\ell)$. We only need examine two corner cases $\bph^1 = 1$ and $1 - \bpl^0 = 1$ and compare the maximum of these cases with $\frac{\Delta (\alpha - 1/2)}{P_{hL}}$.

The key idea to prove Proposition~\ref{prop:border_cases} is to analyze the derivatives of $\xi_h$ and $\xi_\ell$ for a spectrum of lines that passes $(\bph^1 = 0, 1 - \bpl^0 = 0)$. That is, for every $t \ge 0$, we consider the derivatives
under a linear constraint of $1 -\bpl^0 = t\cdot \bph^1$. (Specifically, $t = +\infty$ for line $\bph^1 = 0$.) If $\xi_h$ (or $\xi_\ell$) decreases towards such a line characterized by $t$, then for every $1 -\bpl^0 = t\cdot \bph^1$ and every $\gamma$, $\xi_h$ ($\xi_\ell$, respectively) will not exceed that when $\bph = 1 - \bpl = 0$, which is no more than $\frac{\Delta (\alpha - 1/2)}{P_{hL}}$. On the other hand, if $\xi_h$ and $\xi_\ell$ both increases towards line $t$, they will keep increasing until it reach $\bph^1 = 1$ or $1 - \bpl^0 = 1$. Therefore, $\min(\xi_h, \xi_\ell)$ is maximized at $\bph^1 = 1$ or $1 - \bpl^0 = 1$ (if it exceeds $\frac{\Delta (\alpha - 1/2)}{P_{hL}}$). 

It remains to show that $\xi_h$ and $\xi_\ell$ are always increasing or always non-increasing on every line $t$. While directly solving the derivatives is difficult, we only care about whether they are positive or negative. Recall the fact that for a function $f = \frac{u}{v}$, the derivative of $f$ is $f' = \frac{u'v - uv'}{v^2}$. Therefore, we follow $\sign(f') = \sign(u'v - uv')$ and have the following observation on the derivatives. 

% We analyze the derivatives of $\xi_h$ and $\xi_\ell$ for a spectrum of lines that passes $(\bph^1 = 0, 1 - \bpl^0 = 0)$. That is, for every $t \ge 0$, we consider the derivatives
% under a linear constraint of $1 -\bpl^0 = t\cdot \bph^1$. (Specifically, $t = +\infty$ for line $\bph^1 = 0$.)
\begin{align*}
    \sign(\frac{\partial \xi_h(\bph^1, t\cdot \bph^1)}{\partial \bph^1}) =&\ \sign((1 - \alpha - \gamma) \cdot P_{\ell H} - (\alpha - \frac12) \cdot P_{hH}\cdot P_{\ell L} - (\alpha  - \frac12) \cdot P_{\ell H} \cdot P_{\ell L}\cdot t)).\\
    \sign(\frac{\partial \xi_\ell(\bph^1, t\cdot \bph^1)}{\partial \bph^1}) =&\ \sign(t \cdot (\gamma \cdot P_{hL} - (\alpha - \frac12)\cdot P_{hH} \cdot P_{\ell L}) - P_{hH} \cdot P_{\ell L} \cdot (\alpha - \frac12)). 
\end{align*}

An important observation is that both (signs of the) derivatives are linear on $t$ and independent of $\bph^1$ and $1 - \bpl^0$. Therefore, both derivatives are either always positive or always non-positive. This implies the both $\xi_h$ and $\xi_\ell$ are either always increasing or always non-increasing on each line $t$. This finished the proof of Proposition~\ref{prop:border_cases}.

In the following sections of the proof sketch,  we maximize $\min(\xi_h, \xi_\ell)$ under $\bph^1 = 1$ and under $1 - \bpl^0 = 1$ respectively. As we assumed that $P_{\ell L} \ge P_{hH}$, these two cases are not symmetric. In the $1 - \bpl^0 = 1$ case, $\min(\xi_h, \xi_\ell)$ is upper bounded by linear constraint related to ``Steep'' Segment 2 and ``Tail'' Segment 4, while in the $\bph^1 = 1$ case, the ``Non-Linear'' Segment 3 occurs.

\subsubsection{Case $1 - \bpl^0 = 1$.}
\label{sec:bpl0}
When $1 - \bpl^0 = 1$, $\min(\xi_h, \xi_\ell)$ is upper bounded by the maximum of the extensions of ``Steep'' Segment 2 and of ``Tail'' Segment 4 together. When $\frac{1}{2P_{\ell L}} \ge \alpha > \frac{1}{1+P_{\ell L}}$, $\min(\xi_h, \xi_\ell)\le \frac{\Delta(\alpha - 1/2)}{P_{hL}}$, which relates to Segment 2; and when $\frac{1}{1+P_{\ell L}} \ge \alpha > \frac12$, $\min(\xi_h, \xi_\ell)\le \frac12\Delta\alpha$, which relates to Segment 4. As $\frac{\Delta(\alpha - 1/2)}{P_{hL}} \ge \frac12\Delta\alpha$ if and only if $\alpha \ge \frac{1}{1 + P_{\ell L}}$,  this upper bound is the maximum of the line $\frac{\Delta(\alpha - 1/2)}{P_{hL}}$ and line $\frac12\Delta\alpha$. 
% The changing point of the bound is $\alpha = \frac{1}{1 + P_{\ell L}}$.

% The two lines intersects, and $\frac{\Delta(\alpha - 1/2)}{P_{hL}} \ge \frac12\Delta\alpha$ if and only if $\alpha \ge \frac{1}{1 + P_{\ell L}}$. 

Throughout the case $1 - \bpl^0 = 1$, we assume that $\xi_h$ is positively monotone on $\bph^1$, which holds if and only if an inequality between $\gamma$ and $\alpha$ holds. The case where the inequality does not hold is covered by Proposition~\ref{prop:border_cases}, and $\xi_h \le \frac{\Delta(\alpha - 1/2)}{P_{hL}}$. 

\begin{prop}
    \label{prop:seg3_h} 
    % When $(1 - \alpha - \gamma) > \frac{P_{hH}\cdot P_{\ell L}}{P_{\ell H}} ( \alpha - \frac12)$, $\gamma >\frac{P_{hH}\cdot P_{\ell L}}{P_{hL}} ( \alpha - \frac12)$,
    When $\frac{\infdif}{1 - \thdmaj} \ge \alpha \ge \frac{1}{1 + P_{\ell L}}$, if $1- \bpl^0 = 1$ holds, either $\xi_h$ or $\xi_{\ell}$ cannot exceed $\frac{\Delta(\alpha - \frac12)}{P_{hL}}$.
    % and $\min(\xi_h, \xi_\ell)$ is maximized at $\frac{\Delta(\alpha - \frac12)}{P_{hL}}$.
\end{prop}

\begin{prop}
    \label{prop:seg5_h} 
    % When $(1 - \alpha - \gamma) > \frac{P_{hH}\cdot P_{\ell L}}{P_{\ell H}} ( \alpha - \frac12)$, $\gamma >\frac{P_{hH}\cdot P_{\ell L}}{P_{hL}} ( \alpha - \frac12)$, and 
    When $ \frac{1}{1 + P_{\ell L}}\ge \alpha \ge \frac{1}2$, if $1- \bpl^0 = 1$ holds, either $\xi_h$ or $\xi_{\ell}$ cannot exceed $\frac12 \Delta \alpha$.  
\end{prop}

The proof of Proposition~\ref{prop:seg3_h} and~\ref{prop:seg5_h} resembles each other. 
The cutting point is a specific set of parameters: $\bph^1 = 1 - \bpl^0 = 1$, and $\gamma = \frac{1 - (P_{\ell L} + P_{\ell H})\cdot \alpha}{2}$. In this case, $\xi_h = \xi_{\ell} = \frac12 \Delta \alpha$. (Note that this also serves as a lower bound by Lemma~\ref{lem:xilower}, which eventually be the lower-bound proof for ``Tail'' Segment 4.) If we decrease $\bph^1$ or increase $\gamma$, $\xi_h$ will decreases and be smaller than $\frac12 \Delta \alpha$. If we decrease $\gamma$, $\xi_h \ge \frac12 \Delta \alpha$ only when $\bph^1$ is large, while $\xi_\ell \ge \frac12 \Delta \alpha$ only when $\bph^1$ is small. By solving an inequality, we shall figure out that no $\bph^1 \in [0, 1]$ can satisfies both inequalities. Therefore, $\min(\xi_h, \xi_\ell)\le \frac12 \Delta \alpha$. As $\frac12 \Delta \alpha \le \frac{\Delta(\alpha - 1/2)}{P_{hL}}$ when $\alpha \ge \frac{1}{1 +P_{\ell L}}$, this also implies Proposition~\ref{prop:seg3_h}. 

\subsubsection{Case $\bph^1 = 1$.}
\label{sec:bph1}
Unlike the $1 - \bpl^0 = 1$ case, $\min(\xi_h, \xi_\ell)$ contains the ``Non-Linear'' Segment 3, and the method for Proposition~\ref{prop:seg3_h} and~\ref{prop:seg5_h} does not apply.  Therefore, we take a more straightforward method by maximizing $\min(\xi_h, \xi_\ell)$ parameter by parameter. Firstly, for every fixed $1 - \bpl^0$, we maximizes $\min(\xi_h, \xi_\ell)$ on $\gamma$. We denote $\gamma^*$ as the $\gamma$ that maximizes $\min(\xi_h, \xi_\ell)$. Then we maximize $\min(\xi_h, \xi_\ell)$ where $\gamma = \gamma^*$ on $1 - \bpl^0$. The maximum naturally becomes an upper bound. Finally, the parameters leading to this maximum also serve as proof for a tight lower bound. 

Given that $\xi_h$ decreases and $\xi_\ell$ increases when $\gamma$ increases, the $\gamma$ that maximizes $\min(\xi_h, \xi_\ell)$ satisfies $\xi_h = \xi_\ell$. Therefore, solving this equation, we get $\gamma^*$ and the corresponding $\xi$ (denoted as $\hat{\xi}$) being an upper bound of $\min(\xi_h, \xi_\ell)$

\begin{lem}
    \label{lem:seg4_upper_fix} When $ \frac{1}{2 P_{\ell L}}\ge \alpha \ge \frac{1}2$, for $\bph^1 = 1$ and any fixed $1 - \bpl^0$, let 
    \begin{equation*}
        \hat{\xi} = \frac{\Delta \cdot (\alpha - 1/2 + 1/2 \cdot (1 - \bpl^0))}{(P_{\ell L}\cdot (1 - \bpl^0) + P_{hL})\cdot (P_{\ell H}\cdot (1 - \bpl^0) + P_{hH} +1)}. 
    \end{equation*}
    Then, $\min(\xi_h, \xi_{\ell}) \le \hat{\xi}$ holds for all feasible $\gamma$. 
\end{lem}

The next step is the maximize $\hat{\xi}$ on $1 - \bpl^0$. Here we again apply the trick $\sign(f') = \sign(u'v - uv')$ and examine the sign of the derivative $\frac{\partial \hat{\xi}}{\partial (1 - \bpl^0)}$. It turns out that the $1 - \bpl^0$ that maximizes $\hat{\xi}$ increases as $\alpha$ decreases, and 
that three cases are occurring in order as $\alpha$ decreases. (1) Firstly, when $\alpha$ is large, $\hat{\xi}$ is maximized at $1 - \bpl^0 = 0$, in which $\hat{\xi}\le \frac{\Delta(\alpha - 1/2)}{P_{hL}}$. (2) Secondly, when $\alpha$ becomes smaller, $\hat{\xi}$ is maximized at $1 - \bpl^0  \in (0, 1)$, in which $\hat{\xi} = \xinl$. (3) Finally, when $\alpha$ further decreases, $\hat{\xi}$ is maximized at $1 - \bpl^0 = 1$, and  $\hat{\xi} = \frac12\Delta\alpha$. Note this is also considered as a case $1 - \bpl^0 = 1$, for which Section~\ref{sec:bpl0} gives an upper and lower bound of $\frac12 \Delta \alpha$. 

% In this case, $\gamma^* = \frac{1 - (P_{\ell L} + P_{\ell H})\cdot \alpha}{2}$. This set of parameters is exactly the one mentioned in the case $1 - \bpl^0 = 1$. 

% Instead of directly examine the derivative of $\hat{\xi}$, let $\hat{\func}$ be the numerator of $\frac{\partial \hat{\xi}}{\partial (1 - \bpl^0)}$ (i.e., $u'v - uv'$). Then the sign of $\hat{\func}$ equals that of $\frac{\partial \hat{\xi}}{\partial (1 - \bpl^0)}$. 

% \begin{align*}
%         \hat{\func}(1 - \bpl^0) =&\ -\frac12 P_{\ell L} P_{\ell H} \cdot (1- \bpl^0)^2 -(2\alpha - 1) P_{\ell L}P_{\ell H}\cdot (1-\bpl^0) \\
%         &\ - (\alpha -\frac12) \cdot (P_{hH}P_{\ell L} + P_{hL} P_{\ell H} +P_{\ell L}) + \frac12(P_{hL}P_{hH} + P_{hL})
% \end{align*}

% It is not hard to verify that $\hat{\func}$ is negatively correlated to both $1 - \bpl^0$ and $\alpha$. Therefore, as $\alpha$ decreases, the following situations occurs in turns. (1) $\hat{\func}(0) \le 0$. In this case, $\hat{\xi}$ is maximized at $1 - \bpl^0 = 0$ and $\hat{\xi} = \frac{\Delta(\alpha - 1/2)}{P_{hL} \cdot (P_{hH} + 1)} \le \frac{\Delta(\alpha - 1/2)}{P_{hL}}$. (2) $\hat{\func}(0) > 0 > \hat{\func}(1)$. In this case, $\hat{\xi}$ is maximized at $1 - \bpl^0 \in (0, 1)$, and $\hat{\xi} = \xinl$. (3) $\hat{\func}(1) \ge 0$. In this case, $\hat{\xi}$ is maximized at $1 - \bpl^0 = 1$, and $\hat{\xi} = \frac12\Delta\alpha$. 

By Lemma~\ref{lem:seg4_upper_fix},  case (2) is exactly the upper bound for the ``Non-Linear'' Segment 3, and case (3) is the upper bound for the ``Tail'' Segment 4. Additionally, the changing point between case (2) and (3) is $\alpha = \frac{1}{1+P_{\ell L} + (P_{\ell L} - P_{hH})}$, which is also the cut-off point of Segment 3 and 4. Specifically, if this cut-off point is smaller than $\frac12$, Segment 4 disappears. On the other hand, $\xinl \ge \frac{\Delta(\alpha - 1/2)}{P_{hL}}$ if and only if $\alpha \le \anl$. This becomes the cut-off point of Segment 2 and 3. 

As a conclusion, the upper bound given in the case $\bph^1 = 1$ has three Segments. 

\begin{enumerate}
    \item For $\alpha \ge \anl$, the upper-bound is at most $\frac{\Delta(\alpha - 1/2)}{P_{hL}}$. 
    \item For $\anl > \alpha > \frac{1}{1+P_{\ell L} + (P_{\ell L} - P_{hH})}$, the upper bound is $\xinl$. 
    \item For $\alpha > \frac{1}{1+P_{\ell L} + (P_{\ell L} - P_{hH})}$, the upper bound is $\frac12\Delta\alpha$. 
\end{enumerate}
This contributes to the upper bound for Segment 3 and Segment 4 in $\xi^*(\alpha)$. 

% \Qishen{A Lower Bound should be here.} 

Finally, we wish to derive a (tight) lower bound from the parameters $\gamma^*$, $\bph^1 = 1$, and $1 - \bpl^0$ that maximizes $\min(\xi_h, \xi_\ell)$ (or, equivalently, $\hat{\xi}$) for ``Non-Linear'' Segment 3 and ``Tail'' Segment 4, aligning with the upper bounds, by applying Lemma~\ref{lem:xilower}. We need to show that the parameters satisfy the conditions for Lemma~\ref{lem:xilower}. $\bph^1$ and $\bpl^0$ are guaranteed to be $[0, 1]$. Therefore, it remains to show that $\gamma^* \ge \hat{\xi}$ and $1 - \gamma^* \ge \hat{\xi}$. We don't do Segment 2 because it cannot do better than the case where $\gamma$ is biased, where a lower bound $\frac{\Delta(\alpha - 1/2)}{P_{hL}}$ is given. 

While the proof for this constraint is mostly calculation and solving inequalities, we would like to give a brief high-level idea here. Recall that as $\alpha$ decreases, the $1 - \bpl^0$ that maximizes $\hat{\xi}$ increases. Given that $\bph^1 = 1$ always holds, this implies that minority agents are playing a more and more balanced strategy. Therefore, $\gamma^*$ which guarantees that $\xi_h = \xi_\ell$, will also be close to half of the minority agents, which is $\frac{1 - \alpha}{2}$. As a consequence, $\alpha$ decreases, both $\gamma^*$ and $1 - \gamma^*$ are increasing, while $\hat{\xi}$ is decreasing. Therefore, it suffices to show that at the beginning of the ``Non-Linear'' Segment 3, where $\alpha = \anl$, both $\gamma^* \ge \hat{\xi}$ and $1 - \gamma^* \ge \gamma^*$ (which is easier than directly proving that $1 - \gamma^* \ge \hat{\xi}$) hold. This is achieved by directly solving the inequalities. 

% The condition that $\hat{\func}(1) \ge 0$ is exactly, $\alpha \le  \frac{1}{1+P_{\ell L} + (P_{\ell L} - P_{hH})}$, which serves as the cut-off point of Segment 3 and 4. 

\subsubsection{Wrap up the bounds.}
So far, we have shown that the upper and the lower bound of $\xi$ only derive from one of the four cases. (1) When $\gamma$ is biased (either type-0 agents or type-1 agents dominates the minority group). (2) When $1 - \bpl^0 = 1$ (all type-0 agents always vote for $\bR$). (3) When $\bph^1 = 1$ (all type-1 agents always vote for $\bA$). (4) All other cases, upper-bounded with $\frac{\Delta(\alpha-1/2)}{P_{hL}}$ by Proposition~\ref{prop:border_cases}. Each case has its own upper bound, and cases (1), (2), and (3) also have lower bounds. 

Therefore, the last step is to integrate all the cases and reach a global upper and lower bound for each Segment. For the upper bound, we take the maximum of the upper bound for each case. In each case, there does not exist a strategy profile satisfying all three constraints with $\xi$ equal to their corresponding upper bound. Therefore, overall, no strategy profiles can satisfy all three constraints with $\xi$ equal to the maximum of all upper bounds. For the lower bound, we show that the upper bound in each segment binds a tight lower bound. As a consequence, the bounds consist of the threshold curve $\xi^*(\alpha)$, which finishes the proof.

% The lower bound under each cse 

% as the goal is to find a strategy profile that satisfies all the constraints, it is natural to take the maximum among all the cases. For the upper bound, we also take the maximum of each case. As each case

% For the upper bound, we will see that every lower bound we take binds with a tight upper bound. As a consequence, the bounds consist of the threshold curve $\xi^*(\alpha)$, which finishes the proof. 

Now we start to integrate the upper bound of each case. 
\begin{enumerate}
    \item When $\gamma$ is biased, the upper bound is $\frac{\Delta(\alpha - 1/2)}{P_{hL}}$ by Proposition~\ref{prop:biasedgamma_gen} along Segment 2, 3, and 4.
    \item When $1 - \bpl^0 = 1$, the upper bound is given by the maximum of $\frac{\Delta(\alpha - 1/2)}{P_{hL}}$ and $\frac12\Delta\alpha$ by Proposition~\ref{prop:seg3_h} and~\ref{prop:seg5_h}. 
    \item When $\bph^1 = 0$, the upper  bound is no more than $\frac{\Delta(\alpha - 1/2)}{P_{hL}}$ in Segment 2, $\xinl$ in Segment 3, and $\frac12\Delta\alpha$ in Segment 4, as analyzed in Section~\ref{sec:bph1}. 
    \item In all other cases, the upper bound is at most $\frac{\Delta(\alpha - 1/2)}{P_{hL}}$ by Proposition~\ref{prop:border_cases}.
\end{enumerate}
The order of value of these upper bounds varies in different segments, of which the full proof is in Appendix~\ref{apx:lemmas}. In Segment 2, $\frac{\Delta(\alpha - 1/2)}{P_{hL}}$ is the largest. In Segment 3, $\xinl$ is the largest. Finally, in Segment 4, $\frac12 \Delta \alpha \ge \frac{\Delta(\alpha - 1/2)}{P_{hL}}$. Therefore, taking the maximum on each segment, the upper bound for $\xi$ is $\frac{\Delta(\alpha - 1/2)}{P_{hL}}$ in Segment 2, $\xinl$ in Segment 3, and $\frac12 \Delta \alpha$ in Segment 4 (if exists). This is exactly the upper bound given in $\xi^*(\alpha)$. 

Finally, we show the tight lower bound corresponding to the upper bound in each segment. 

\begin{enumerate}
    \item In ``Steep'' Segment 2, Proposition~\ref{prop:seg3lower} in the biased case gives the lower bound $\frac{\Delta(\alpha - 1/2)}{P_{hL}}$, where all minority agents always vote for $\bA$. 
    \item In ``Non-Linear'' Segment 3, we show that the corresponding $\xi^*$ and $1 - \bpl^0$ that maximizes $\hat{\xi}$ serves as a valid set of parameters in Section~\ref{sec:bph1}. By Lemma~\ref{lem:xilower}, $\xinl$ is also a lower bound. 
    \item For ``Tail'' Segment 4, the parameters $\bph^1 = 1 - \bpl^0 = 1$, and $\gamma = \frac{1 - (P_{\ell L} + P_{\ell H})\cdot \alpha}{2}$ (as discussed in the Section~\ref{sec:bpl0} and~\ref{sec:bph1}) gives the lower bound $\frac12\Delta\alpha$ by Lemma~\ref{lem:xilower}. 
\end{enumerate}

Therefore, the lower bound on $\xi$ is exactly given in $\xi^*(\alpha)$. We finish the proof that $\xi^*(\alpha)$ is exactly the threshold curve for which a strategy profile satisfying all three constraints exists.

\section{Conclusion and Future Works}
We study a voting problem with two opposed groups whose preferences depend on a ground truth that cannot be directly observed. conflicting preferences, incomplete information, and group strategic behavior bring challenge to the problem. \citet{deng2024aggregation} provide conditions for a strong Bayes-Nash equilibrium to exist yet is less informative in the scenario where no equilibrium exists. We adopt the solution concept of $\varepsilon$-approximate ex-ante Bayesian $\kd$-strong equilibrium, in which no group of at most $\kd$ agents have incentives to deviate, to study scenarios where agents are less organized and have limited power to coordinate on strategic behavior. We give a closed form of the threshold $\xi^*(\alpha)$, which fully characterizes the region where an $\varepsilon$-approximate ex-ante Bayesian $\kd$-strong equilibrium exists and reaches the informed majority decision. The curve includes multiple drastically different segments, including a highly non-trivial non-linear segment. This implies the complexity of the strategic behavior in the voting problem, which in turn demonstrates the capability of the $\kd$-strong equilibrium to provide an informative perspective. 

Several natural extensions serve as interesting future directions for our work. Although we fully characterize the existence of $\kd$-strong equilibrium in favor of the majorities, the question remains open regarding the equilibrium in favor of the minorities, which will be an extension of the current $\xi^*$ to the range $\frac12 > \alpha > 0$. The second direction is to consider a full spectrum of preferences among agents, including those who always prefer one alternative and those whose preferences depend on the world state. Finally, it is also interesting to study the problem of a more structured agent relationship. For example, the relationship of agents are represented as edges on a graph, and a group will deviate only if they form a clique or have enough edges. 
\bibliographystyle{ACM-Reference-Format}
\bibliography{references,newref}
% \bibliography{}

\clearpage
\appendix
% \setcounter{tocdepth}{1}

% \section{Useful Lemmas}
\section{Proof of Lemma~\ref{lem:dominated}}
\label{apx:dominated}
\begin{proof}
    Suppose the strategies other agents play is $\stgp_{-i}$. The only scenario where $i$'s vote can change the outcome is the pivotal case, where both alternatives need exactly one vote to pass their threshold. Therefore, different strategies $i$ plays only affect his/her expected utility conditioned on $i$ being pivotal. The ex-ante expected utility of $i$ with strategy profile $\stgp = (\stg, \stgp_{-i})$ conditioned on him/her being pivotal is as follows. Note that the posterior belief on the world state conditioned on $i$ being pivotal does not depend on $i$'s strategy. 
    \begin{align*}
        \ut_{i}(\stgp \mid \piv) =&\ \Pr[L\mid \piv, \stgp_{-i}] \cdot ((P_{hL}\cdot \bph + P_{\ell L}\cdot \bpl)\cdot\vt_{i}(L, \bA)\\ &\ \ + (P_{hL}\cdot (1-\bph) + P_{\ell L}\cdot (1-\bpl))\cdot \vt_{i}(L, \bR)) \\
        &\ +  \Pr[H\mid \piv, \stgp_{-i}]\cdot ((P_{hH}\cdot \bph + P_{\ell H}\cdot \bpl)\cdot\vt_{i}(H, \bA) \\&\ \ + (P_{hH}\cdot (1-\bph) + P_{\ell H}\cdot (1-\bpl))\cdot \vt_{i}(H, \bR)). 
    \end{align*}

    When agent $i$ plays a different strategy $\stg' = (\bpl', \bph')$, the expected utility conditioned on pivotal can be written in the similar form. Let $\stgp' = (\stg', \stgp_{-i})$. Then, the difference between the expected utility when $i$ plays $\stg$ and $\stg'$ is exactly the expected utility difference conditioned on $i$ being pivotal.

    \begin{align*}
        \ut_i(\stgp') - \ut_i(\stgp) = &\ \ut_i(\stgp' \mid \piv) - \ut_i(\stgp \mid \piv) \\
        =&\ \Pr[L \mid \piv, \stgp_{-i}] \cdot ((P_{hL} \cdot (\bph' - \bph) + P_{\ell L} \cdot (\bpl' - \bpl))\cdot (\vt_i(L, \bA) - \vt_i(L, \bR))\\
        &\ + \Pr[H \mid \piv, \stgp_{-i}] \cdot ((P_{hH} \cdot (\bph' - \bph) + P_{\ell H} \cdot (\bpl' - \bpl))\cdot (\vt_i(H, \bA) - \vt_i(H, \bR)). 
    \end{align*}

% For friendly agents, $\vt_i(L, \bA) - \vt_i(L, \bR) > 0$ and $\vt_i(H, \bA) - \vt_i(H, \bR) > 0$. When $\bph < 1$ or $\bpl < 1$, let $\bph' = \bpl' = 1$. Then $\ut_i(\stgp') - \ut_i(\stgp) > 0$, and $\stg$ is weakly dominated by $\stg'$. On the other hand, when $\bpl = \bph = 1$, for any $\stg'$, $\ut_i(\stgp') - \ut_i(\stgp) \le 0$, and $\stg$ is not a weakly dominated strategy. The reasoning for unfriendly agents resembles that for friendly agents. 

For majority agents, $\vt_i(L, \bA) - \vt_i(L, \bR) < 0$ and $\vt_i(H, \bA) - \vt_i(H, \bR) > 0$. When $\bph < 1$ and $\bpl > 0$, let $\bph' = \bph + \varepsilon$ and $\bpl' = \bpl - \varepsilon\cdot \frac{P_{hL} + P_{hH}}{P_{\ell L} + P_{\ell H}}$, where $\varepsilon > 0$ is a small constant that guarantees $\bph \le 1$ and $\bpl \ge 0$. Then we have $P_{hL} \cdot (\bph' - \bph) + P_{\ell L} \cdot (\bpl' - \bpl) < 0$ and $P_{hH} \cdot (\bph' - \bph) + P_{\ell H} \cdot (\bpl' - \bpl) > 0$, which implies that $\ut_i(\stgp') - \ut_i(\stgp) > 0$ and that $\stg$ is dominated by $\stg'$. Now, without loss of genearlity, suppose $\bph = 1$. In this case, for all $\stg' \neq \stg$, either $P_{hL} \cdot (\bph' - \bph) + P_{\ell L} \cdot (\bpl' - \bpl) > 0$ or $P_{hH} \cdot (\bph' - \bph) + P_{\ell H} \cdot (\bpl' - \bpl) < 0$ hold. (If $P_{hH} \cdot (\bph' - \bph) + P_{\ell H} \cdot (\bpl' - \bpl) < 0$ does not hold and $\stg' \neq \stg$, then $\bpl' > \bpl$ must hold. Combined with the fact that $P_{hH} > P_{hL}$ and $P_{\ell H} < P_{\ell L}$, $P_{hL} \cdot (\bph' - \bph) + P_{\ell L} \cdot (\bpl' - \bpl) > 0$ holds.) 
Therefore, by constructing $\stgp_{-i}$ so that $\Pr[L \mid \piv, \stgp_{-i}]$ is sufficiently close to 0 or 1, we can find a scenario where $\ut_i(\stgp') - \ut_i(\stgp) < 0$, which implies that $\stgp$ is not weakly dominated by $\stgp'$. 
The reasoning for $\bpl = 0$ and that for the minority agents resembles the reasoning above. 
\end{proof}

\section{Technical Lemma on reaching informed majority decision. }

Here we give some useful results that will be applied in the full proof. 

\citep{han2023wisdom}, provides useful tool to determine the fidelity of a profile sequence by the {\em expected vote share} of alternative $\bA$ under different world states. 

\begin{dfn}[\bf{Expected vote share}]
   Given an instance of $\ag$ agents, and a strategy profile $\stgp$, let random variable $\xrv_{i}^{\ag}$ be "agent $i$ votes for $\bA$":  $\xrv_{i}^{\ag} = 1$ if agent $i$ votes for $\bA$, and $\xrv_{i}^{\ag} = 0$ if $i$ votes for $\bR$. Then the expected vote share (of $\bA$) is defined as follows: 
\begin{align}
    \thd_{H} =&\ \frac{1}{\ag}\sum_{i=1}^{\ag}E[\xrv_{i}^{\ag}\mid H].\label{eq:hthdh}
\end{align}
\begin{align}
    \thd_{L} =& \ \frac{1}{\ag}\sum_{i=1}^{\ag}E[\xrv_{i}^{\ag}\mid L]\label{eq:hthdl}. 
\end{align}
Specifically, $\thd_H$ ($\thd_L$, respectively) is the \exshare\ of $\bA$ conditioned on world state $H$ (world state $L$, respectively). 
\end{dfn}

\begin{thm}
\label{thm:arbitrary}
Given an arbitrary sequence of instances and arbitrary sequence of strategy profiles $\{\stgp_{\ag}\}_{\ag=1}^\infty$, let $\thd_H^\ag$ and $\thd_L^{\ag}$ be the expected vote share for each $\stgp_{\ag}$. Let 
\begin{equation*}
    F = \liminf_{\ag\to\infty} \sqrt{\ag}\cdot \min \left( \thd^{\ag}_H - \frac12, \frac12 - \thd^{\ag}_L\right)
\end{equation*}
\begin{itemize}
     \item If $F= +\infty$, the fidelity of $\stgp_{\ag}$ converges to 1 , i.e., $\lim_{\ag\to\infty} \acc(\stgp_{\ag}) = 1$. 
     \item If $F < 0$ (including $-\infty$), $\acc(\stgp_{\ag})$ does NOT converge to 1. 
     \item If $F \ge 0$ (not including $+\infty$), and the variance of $\sum_{i=1}^{\ag} \xrv_{i}^{\ag}$ is  $\Theta(\ag)$ under every world state,  $\acc(\stgp_{\ag})$ does NOT converge to 1. 
\end{itemize}
\end{thm}

\begin{remark}
    The original theorem in ~\citep{han2023wisdom} is presented on {\em \exshare}, which is the expected vote share of the informed majority decision minus the majority rule's winning threshold. For the convenience of our presentation, we translate the theorem to expected vote share of $\bA$ while keeping it equivalent to the original theorem. 
\end{remark}

Below are the lemmas that we apply to deal with the case where the variance of a strategy profile becomes $o(\ag)$. 

\begin{lem}
    \label{lem:var_sublinear} For any instance of $\ag$ agents and an arbitrary strategy profile $\stgp$, let $\delta_H = \frac{1}{\ag}Var(\sum_{i=1}^{\ag} \xrv_{i}^{\ag} \mid H)$ and $\delta_L = \frac{1}{\ag}Var(\sum_{i=1}^{\ag} \xrv_{i}^{\ag} \mid L)$, then $\thd_H - \thd_L \le \frac{2\cdot \Delta\cdot \delta_H}{\min(P_{hH}, P_{\ell H})}$ and $\thd_H - \thd_L \le \frac{2\cdot \Delta\cdot \delta_L}{\min(P_{hL}, P_{\ell L})}$.
\end{lem}

\begin{proof}
    For each $i$, we denote agent $i$' strategy as $(\bpl^i, \bph^i)$. In world state $H$, $\xrv_{i}^{\ag} = 1$ with probability $P_{hH} \cdot \bph^i + P_{\ell H} \cdot \bpl^i$. Therefore, 
    \begin{align*}
        \delta_H =&\ \frac{1}{\ag}Var(\sum_{i=1}^{\ag} \xrv_{i}^{\ag} \mid H)\\
        =&\ \frac1{\ag} \cdot \sum_{i=1}^{\ag}Var( \xrv_{i}^{\ag} \mid H)\\
        =&\ \frac1{\ag} \cdot \sum_{i=1}^{\ag}(P_{hH} \cdot \bph^i + P_{\ell H} \cdot \bpl^i) \cdot (1 - P_{hH} \cdot \bph^i - P_{\ell H} \cdot \bpl^i). 
    \end{align*}
    Similarly, when the world state is $L$, 
    \begin{align*}
        \delta_L =&\ \frac1{\ag} \cdot \sum_{i=1}^{\ag}(P_{hL} \cdot \bph^i + P_{\ell L} \cdot \bpl^i) \cdot (1 - P_{hL} \cdot \bph^i - P_{\ell L} \cdot \bpl^i). 
    \end{align*}

    On the other hand,
    \begin{align*}
        \thd_H - \thd_L = &\ \frac{1}{\ag}\sum_{i=1}^{\ag}(E[\xrv_{i}^{\ag}\mid H] - E[\xrv_{i}^{\ag}\mid L])\\
        =&\ \frac{1}{\ag}\sum_{i=1}^{\ag}\cdot \Delta\cdot (\bph^i - \bpl^i). 
    \end{align*}

    Let $\delta^i_H = (P_{hH} \cdot \bph^i + P_{\ell H} \cdot \bpl^i) \cdot (1 - P_{hH} \cdot \bph^i - P_{\ell H} \cdot \bpl^i)$. If $P_{hH} \cdot \bph^i + P_{\ell H} \cdot \bpl^i \le \frac12$, $\delta^i_H\ge \frac12 (P_{hH} \cdot \bph^i + P_{\ell H} \cdot \bpl^i)$. Therefore, $\bph^i - \bpl^i \le \bph^i \le \frac{2\delta_H^i}{P_{hH}}$. If $P_{hH} \cdot \bph^i + P_{\ell H} \cdot \bpl^i > \frac12$, $\delta^i_H>\frac12 (1 - P_{hH} \cdot \bph^i + P_{\ell H} \cdot \bpl^i)$, then 
    \begin{align*}
        \bph^i - \bpl^i = & (1 - \bpl^i) - (1 - \bph^i) \\
        \le &\ 1 - \bpl^i \\
        \le & \frac{2\delta_H^i}{P_{\ell H}}. 
    \end{align*}

    Similarly, for world sate $L$, we also have $\bph^i - \bpl^i \le \frac{2\delta_L}{P_{\ell L}}$ or $\bph^i - \bpl^i \le \frac{2\delta_L}{P_{h L}}$. Note that Therefore, by summing up all the $i$, we have 
    $\thd_H - \thd_L \le \frac{2\cdot \Delta\cdot \delta_H}{\min(P_{hH}, P_{\ell H})}$ and $\thd_H - \thd_L \le \frac{2\cdot \Delta\cdot \delta_L}{\min(P_{hL}, P_{\ell L})}$.
\end{proof}

\begin{lem}
    \label{lem:var_linear}
    For any instance of $\ag$ agents and an arbitrary strategy profile $\stgp$ where all agents play non-dominated strategies, let $\maj$ be the set of all majority agents, $\delta_H^{\maj} = \frac{1}{\ag}Var(\sum_{i \in \maj} \xrv_{i}^{\ag} \mid H)$ and $\delta_L^{\maj} = \frac{1}{\ag}Var(\sum_{i\in\maj} \xrv_{i}^{\ag} \mid L)$. Then  $\delta_H^\maj \ge \frac{\min(P_{\ell H}, P_{hH})}{2 \cdot \Delta} \cdot (\thd_H - \thd_L)$, and $\delta_L^\maj \ge \frac{\min(P_{\ell L}, P_{hL})}{2 \cdot \Delta} \cdot (\thd_H - \thd_L)$. 
\end{lem}

\begin{proof}
    The proof resembles that of Lemma~\ref{lem:var_sublinear}. For each $i$, we denote agent $i$' strategy as $(\bpl^i, \bph^i)$. In world state $H$, $\xrv_{i}^{\ag} = 1$ with probability $P_{hH} \cdot \bph^i + P_{\ell H} \cdot \bpl^i$. Therefore, 
    \begin{equation*}
        Var(\xrv_{i}^{\ag} \mid H) = (P_{hH} \cdot \bph^i + P_{\ell H} \cdot \bpl^i). 
    \end{equation*}
    Likewise, 
    \begin{equation*}
        Var(\xrv_{i}^{\ag} \mid L) = (P_{hL} \cdot \bph^i + P_{\ell L} \cdot \bpl^i). 
    \end{equation*}
    By the proof of Lemma~\ref{lem:var_sublinear}, we know that $\bph^i - \bpl^i \le \frac{2Var(\xrv_{i}^{\ag} \mid H)}{\min(P_{\ell H}, P_{hH})}$ and $\bph^i - \bpl^i \le \frac{2Var(\xrv_{i}^{\ag} \mid L)}{\min(P_{\ell L}, P_{hL})}$. Now we sum up the inequalities for all $i \in \maj$. This will be 
    \begin{align*}
         \frac{2\delta_H^\maj \cdot \ag}{\min(P_{\ell H}, P_{hH})} \ge &\ \sum_{i \in \maj} \bph^i - \bpl^i,\\
         \frac{2\delta_L^\maj \cdot\ag}{\min(P_{\ell L}, P_{hL})} \ge &\ \sum_{i \in \maj} \bph^i - \bpl^i.
    \end{align*}
    On the other hand, $\thd_H - \thd_L = \frac{1}{\ag} \cdot \sum_{i= 1}^{\ag}\Delta \cdot (\bph^i - \bpl^i)$. For all the minority playing non-dominated strategies, $\bph^i = 0$ or $\bpl^i = 1$ holds, which implies to $\bph - \bpl \le 0$. Therefore, $\thd_H - \thd_L \le \frac{1}{\ag} \cdot \sum_{i \in \maj} \Delta \cdot (\bph^i - \bpl^i)$. Therefore, we have 
    \begin{align*}
        \delta_H^\maj \ge &\ \frac{\min(P_{\ell H}, P_{hH})}{2 \cdot \ag} \cdot \sum_{i \in \maj} \bph^i - \bpl^i\\
        \ge &\ \frac{\min(P_{\ell H}, P_{hH})}{2 \cdot \ag} \cdot \frac{\ag}{\Delta}\cdot (\thd_H - \thd_L)
    \end{align*}
    Likewise, $\delta_L^\maj \ge \frac{\min(P_{\ell L}, P_{hL})}{2 \cdot \Delta} \cdot  (\thd_H - \thd_L)$.
\end{proof}

Lemma~\ref{lem:var_linear} implies that, if the majority agents' strategies of  strategy profile sequence has a sub-linear variance, the difference between the expected vote share is also sub-linear. 

\begin{coro}
    \label{coro:var_sublinear}
    Let $\maj$ be the set of all majority agents. For an arbitrary sequence of instances and arbitrary sequence of strategy profiles $\{\stgp_{\ag}\}_{\ag=1}^\infty$ where all majority agents play non-dominated strategies, if for any $\delta > 0$ there exists infinitely many $\ag$ such that in some world state $\wos$, $Var(\sum_{i\in \maj} \xrv_{i}^{\ag} \mid \wos) < \delta \cdot \ag$, then for any $\varepsilon > 0$, there exists infinitely many $\ag$ such that in some world state $\wos$, $\thd_H - \thd_L < \varepsilon$. 
\end{coro}
\section{Full Proof}
\subsection{Lemmas}
\label{apx:lemmas}
In this section, we give the technical lemmas to prove Theorem~\ref{thm:thresholdk}. They either form crucial prerequisites for some intermediate results or guarantee some parameters are in the correct range. 

The following Lemma~\ref{lem:acctoeq} restrict the deviation to be consider into those with only minority agents. 
\begin{lem}
\label{lem:acctoeq}
    Given a strategy profile $\stgp$, if there exists $e > 0$ such that $\lp_H^{\bA}(\stgp) \ge 1 - e$ and $\lp_H^{\bA}(\stgp) \ge 1 - e$ holds, then for any deviating strategy profile $\stgp'$, at least one of the following holds. (1) No agent has a utility increase more than $\frac{2B(B+1)e}{P_H \cdot P_L}$. (2) Every majority agents get strictly negative utility decrease. 
\end{lem}

\begin{proof}
    The difference on the ex-ante expected utility of agent $i$ under two strategy profiles is 
    \begin{align*}
        \ut_{i}(\stgp'_\ag) - \ut_{i}(\stgp_\ag) = &\ P_H\cdot (\lp_{H}^{\bA}(\stgp'_\ag) - \lp_{H}^{\bA}(\stgp_\ag))(\vt_{i}(H, \bA) - \vt_{i}(H,\bR))\\
        &\ + P_L\cdot (\lp_{L}^{\bR}(\stgp'_\ag) - \lp_{L}^{\bR}(\stgp_\ag))(\vt_{i}(L, \bR) - \vt_{i}(L,\bA)).
    \end{align*}
    
    We discuss on different cases on the fidelity of $\stgp'_\ag$. 

    Suppose $\acc(\stgp'_\ag) \ge 1 -  \frac{(B+1)e}{P_H\cdot P_L}$. For a majority agent $i$, $\ut_{i}(\stgp'_\ag) - \ut_{i}(\stgp_\ag) \le P_H\cdot e\cdot B + P_L\cdot e\cdot B = Be.$ For a minority agent $i$, $\ut_{i}(\stgp'_\ag) - \ut_{i}(\stgp_\ag) \le P_H\cdot \frac{(B+1)e}{P_H^2\cdot P_L}\cdot B + P_L\cdot \frac{(B+1)e}{P_H\cdot P_L^2}\cdot B = \frac{2B(B+1)e}{P_H\cdot P_L}.$
    Therefore, no agent can get expected utility gain more than $\frac{2B(B+1)e}{P_H\cdot P_L}$. 

    Suppose $\acc(\stgp'_\ag) < 1 - \frac{(B+1)e}{P_H\cdot P_L}$. We first show that majority agents will not join the deviating group. Then we show that at most $\xi\cdot \ag$ minority agents cannot successfully deviate. $\acc(\stgp'_\ag) < 1 - \frac{(B+1)e}{P_H\cdot P_L}$, at least one of $\lp_H^{\bA}(\stgp'_\ag) < 1 - \frac{(B+1)e}{P_H\cdot P_L}$ and $\lp_L^{\bR}(\stgp'_\ag) < 1 - \frac{(B+1)e}{P_H\cdot P_L}$ holds. When $\lp_H^{\bA}(\stgp'_\ag) < 1 - \frac{(B+1)e}{P_H\cdot P_L}$, the utility difference of a majority agent $i$ will be $\ut_{i}(\stgp'_\ag) - \ut_{i}(\stgp_\ag) < P_H\cdot (e - \frac{(B+1)e}{P_H\cdot P_L} )\cdot 1 + P_L\cdot e\cdot B < (P_H - 1)\cdot e + (P_L -1)\cdot B \cdot e < 0.$ When $\lp_L^{\bR}(\stgp'_\ag) < 1 - \frac{(B+1)e}{P_H\cdot P_L}$ holds, the utility difference of a majority agent $i$ is also strictly negative according to similar reasoning. Therefore, a majority agent has no incentives to join the deviation group. 
\end{proof}
The following Lemma shows that the deviation succeeds if and only if the fidelity of the deviating strategy does not converges to 1. 
\begin{lem}
    \label{lem:minority_deviate}
    Let $\{\stgp_\ag\}$ be an arbitrary sequence of strategy profile such that $\acc(\stgp_\ag)$ converges to 1 and $\{\stgp_\ag'\}$ be a sequence of deviating strategy profile from $\{\stgp_\ag\}$ where for each $\ag$, the deviating group $D$ contains only minority agents.
    \begin{enumerate}
        \item If $\lim_{\ag \to \infty} \acc(\stgp_\ag') =1$, the expected utility increase of the minority agents converges to 0. 
        \item If $\lim_{\ag \to \infty} \acc(\stgp_\ag') =1$ does not hold, there exists a constant $\varepsilon$ and infinitely many $\ag$ such that every minority agent gains at least $\varepsilon$ in the expected utility. 
    \end{enumerate}
    % Then the deviation succeed if and only if $\lim_{\ag \to intfy} \acc(\stgp_\ag) =1$ does not hold. 
\end{lem}
\begin{proof}
    Consider the utility difference of an arbitrary minority agent between $\stgp_\ag$ and $\stgp_\ag'$.  $\acc(\stgp_\ag)$ converges to 1 implies that $\lp_H^\bA(\stgp_\ag)$ and $\lp_L^\bR(\stgp_\ag)$ converges to 1.

    When $\lim_{\ag \to \infty} \acc(\stgp_\ag') =1$, $\lp_H^\bA(\stgp_\ag')$ and $\lp_L^\bR(\stgp_\ag')$ converges to 1. Therefore, for any constant $\varepsilon > 0$, we can find a $N > 0$ such that for all $\ag > N$, all $\lp_H^\bA(\stgp_\ag)$, $\lp_L^\bR(\stgp_\ag)$, $\lp_H^\bA(\stgp_\ag')$, and $\lp_L^\bR(\stgp_\ag')$ is at least $1 - \varepsilon$
    \begin{equation*}
        \ut_{i}(\stgp'_\ag) - \ut_{i}(\stgp_\ag) \le P_H \cdot \varepsilon \cdot B + P_L \cdot \varepsilon \cdot B = \varepsilon\cdot B. 
    \end{equation*}
    Therefore, the expected utility gain of any minority agents converge to 0. 

    On the other hand, when $\lim_{\ag \to \infty} \acc(\stgp_\ag') =1$ does not hold, there are a constant $\delta > 0$ and infinitely many $\ag$ (denote the set as $N'$) in which either $\lp_H^\bA(\stgp_\ag') \le 1 - \delta$ or $\lp_L^\bR(\stgp_\ag') \le 1 - \delta$ holds. Without loss of generality, suppose $\lp_H^\bA(\stgp_\ag') \le 1 - \delta$ hold. On the other hand, for all sufficiently large $\ag$, both $\lp_H^\bA(\stgp_\ag)$ and $\lp_L^\bR(\stgp_\ag)$ are at least $1 - P_HP_L \cdot \delta/(2B)$. Therefore, 
    \begin{align*}
        \ut_{i}(\stgp'_\ag) - \ut_{i}(\stgp_\ag) \ge &\ P_H \cdot (\delta - P_HP_L\cdot\delta / (2B))\cdot 1 - P_L \cdot P_HP_L\delta/(2B) \cdot B\\
        \ge &\ P_H\cdot \delta \cdot (\frac{2B-1}{2B} - \frac{P_L^2}{2})\\
        >&\ 0.  
    \end{align*}
    Therefore, there exists a constant $\varepsilon = P_H\cdot \delta \cdot (\frac{2B-1}{2B} - \frac{P_L^2}{2})$ and infinitely many $\ag$ such that every minority agent gain at least $\varepsilon$ in the expected utility. 
\end{proof}

The following two lemmas guarantees that $\anl$ is between $\frac{1}{2P_{\ell L}}$ and $\frac{1}{1 +P_{\ell L}}$. 

\begin{lem}
    \label{lem:a4upper}
    $\anl < \frac{1}{2P_{\ell L}}$. 
\end{lem}
\begin{proof}
    By manipulating the inequality, $\anl < \frac{1}{2P_{\ell L}}$ is equivalent to 
    \begin{equation*}
        (1-P_{hH})^2 + P_{hH}P_{\ell L}^2 + 3 P_{hH}P_{\ell L} + P_{\ell L}^2 + 2P_{hL}P_{\ell L}\sqrt{P_{hH}P_{\ell L}P_{hL}P_{\ell H}} - 2P_{hH}P_{\ell L}^3 - P_{hH}^2 P_{\ell L} - 2P_{\ell L} > 0. 
    \end{equation*}

    Given that $P_{hL} < P_{hH} \le P_{\ell L}$, $\sqrt{P_{hH}P_{\ell L}P_{hL}P_{\ell H}} \ge P_{hL} P_{\ell L}$. Therefore, for fixed $P_{hL}$ and $P_{\ell L}$, we consider the following function of $P_{hH}$

    \begin{align*}
        \func(P_{hH}) =&\ (1-P_{hH})^2 + P_{hH}P_{\ell L}^2 + 3 P_{hH}P_{\ell L} + P_{\ell L}^2 + 2P_{hL}^2P_{\ell L}^2 - 2P_{hH}P_{\ell L}^3 - P_{hH}^2 P_{\ell L} - 2P_{\ell L} \\
        =&\ P_{hL} \cdot P_{hH}^2 - (2P_{\ell L}^3 - P_{\ell L}^2 -3P_{\ell L}+2) \cdot P_{hH} + P_{\ell L}^2 + 2P_{hL}^2 P_{\ell L}^2 - 2P_{\ell L} + 1. 
    \end{align*}
    It is not hard to verify that $\func(P_{hH}) \le LHS$. Therefore, it suffice to show that $\func(P_{hH}) > 0$ for all $P_{hL} \le P_{hH} \le P_{\ell L}$. We prove this by showing that $\func(P_{hL}) > 0$ and that $\func$ increases on $P_{hH}$. 

    First, we have
    \begin{align*}
        \func(P_{hL}) =&\ P_{\ell L}^2 +P_{hL}P_{\ell L}^2 + 3 P_{hL}P_{\ell L} + P_{\ell L}^2 + 2P_{hL}^2P_{\ell L}^2 - 2P_{hL}P_{\ell L}^3 - P_{hL}^2 P_{\ell L} - 2P_{\ell L}\\
        =&\ 4P_{\ell L}^2 P_{hL}^2\\
        > &\ 0.
    \end{align*}
    
    Secondly, 
    \begin{align*}
        \frac{d\func}{d P_{hH}} = &\ 2P_{hL}\cdot P_{hH} - (2P_{\ell L}^3 - P_{\ell L}^2 -3P_{\ell L}+2)\\
        > &\ 2(1 - P_{\ell L}^2)- (2P_{\ell L}^3 - P_{\ell L}^2 -3P_{\ell L}+2)\\
        =& P_{\ell L}\cdot (1 - P_{\ell L})\cdot (2P_{\ell L} - 1)\\
        >&\ 0.
    \end{align*}
    Therefore, $LHS \ge \func(P_{hH}) > \func(P_{hL}) > 0$. This implies $\anl < \frac{1}{2P_{\ell L}}$. 
\end{proof}

\begin{lem}
    \label{lem:a4lower} $\anl \ge \frac{1}{1 +P_{\ell L}}$.
\end{lem}

\begin{proof}
    By manipulating the inequality, $\anl \ge \frac{1}{1 +P_{\ell L}}$ is equivalent to 
    \begin{equation*}
        P_{hH}^2 - P_{hH} + 3P_{hH}P_{\ell L}^2  + P_{\ell L}^2 + 2P_{hL}(1+P_{\ell L})\sqrt{P_{hH}P_{\ell L}P_{hL}P_{\ell H}} - 2P_{hH}P_{\ell L}^3 - P_{hH}^2 P_{\ell L} - P_{\ell L} \le 0. 
    \end{equation*}

    Likewise, we view $P_{\ell L}$ as fixed and define $\func(P_{hH}) = LHS$ be a function of $P_{hH}$. We show that $\func(P_{\ell L}) = 0$ and $\func$ is increasing on $P_{hH}$. 

    Firstly, 
    \begin{align*}
        \func(P_{\ell L}) = &\ P_{\ell L}^2 - P_{\ell L} + 3P_{\ell L}^3 +P_{\ell L}^2 +2(1-P_{\ell L}^2)\cdot P_{\ell L}\cdot (1 - P_{\ell L}) - 2P_{\ell L}^4 - P_{\ell L}^3 -P_{\ell L}\\
        =&\ 0. 
    \end{align*}
    Then we consider the derivative of $\func$.
    \begin{equation*}
        \frac{d\func}{dP_{hH}} = 2(1 - P_{\ell L})\cdot P_{hH} + 3P_{\ell L}^2-2P_{\ell L}^3 - 1 + \frac{P_{hL}(1 +P_{\ell L})\cdot \sqrt{P_{hL}P_{\ell L}}}{\sqrt{P_{hH} P_{\ell H}}}\cdot (1 - 2P_{hH}). 
    \end{equation*}
    We discuss the derivative on $P_{hH} < \frac12$ and $P_{hH} \ge \frac12$ separately.

    When $P_{hL} < P_{hH}< \frac12$, $(1 - 2P_{hH}) \ge0$. Therefore, given that $\sqrt{P_{hH} P_{\ell H}} \le \frac12$, we have 
    \begin{equation*}
        \frac{d\func}{dP_{hH}} \ge  2(1 - P_{\ell L})\cdot P_{hH} + 3P_{\ell L}^2-2P_{\ell L}^3 - 1 + 2P_{hL}(1 +P_{\ell L})\cdot \sqrt{P_{hL}P_{\ell L}}\cdot (1 - 2P_{hH}). 
    \end{equation*}
    Note that RHS (denoted by $\func'_1(P_{hH})$) is linear to $P_{hH}$. As we aim to show that  it suffice to show that $\func'_1(P_{hH}) \ge 0$ for all $P_{hL}\le P_{hH} \le \frac12$, it suffice to show that $\func'_1(P_{hH}) \ge 0$ when $P_{hH} = P_{hL}$ and when $P_{hH} = \frac12$. 
    \begin{align*}
        \func'_1(P_{hL}) =&\  2(1-P_{\ell L})^2 - (1-P_{\ell L})^2 \cdot (2P_{\ell L} + 1) + 2(1-P_{\ell L})(1+P_{\ell L}) \cdot (2P_{\ell L} - 1) \cdot \sqrt{P_{\ell L}(1-P_{\ell L})}\\
        =&\ (2P_{\ell L} - 1)\cdot (1 - P_{\ell L}) \cdot \sqrt{1-P_{\ell L}} \cdot (2\cdot (1+P_{\ell L})\cdot \sqrt{P_{\ell L}} - \sqrt{1-P_{\ell L}})\\
        >&\ 0. \\
        \func'_1(\frac12) = &\ 1 - P_{\ell L} + 3P_{\ell L}^2-2P_{\ell L}^3 - 1\\
        =&\ P_{\ell L} \cdot (1 - P_{\ell})(2P_{\ell L} - 1)\\
        >&\ 0. 
    \end{align*}

    When $\frac12 \le  P_{hH}\le  P_{\ell L}$, $(1 - 2P_{hH}) \le0$. Therefore, by $P_{hL} < P_{hH} \le P_{\ell L}$, $\sqrt{P_{hH} P_{\ell H}}\ge \sqrt{P_{hL} P_{\ell L}} $, we have 
    \begin{equation*}
        \frac{d\func}{dP_{hH}} \ge  2(1 - P_{\ell L})\cdot P_{hH} + 3P_{\ell L}^2-2P_{\ell L}^3 - 1 + P_{hL}(1 +P_{\ell L})\cdot (1 - 2P_{hH}). 
    \end{equation*}

    Similarly, the RHS (denoted by $\func'2$) is linear to $P_{hH}$. Therefore, it suffice to show that $\func'_2(\frac12) \ge0$ and $\func'_2(P_{\ell L}) \ge 0$. 
    \begin{align*}
        \func'_2(\frac12) = &\ \func'_1(\frac12) > 0. \\
        \func'_2(P_{\ell L}) = &\ 2(1 - P_{\ell L})\cdot P_{\ell L} + 3P_{\ell L}^2-2P_{\ell L}^3 - 1 +(1 - P_{\ell L})(1+P_{\ell L})(1 - 2P_{\ell L})\\
        =&\ 0. 
    \end{align*}
    Therefore, for all $P_{hH} \in (P_{hL}, P_{\ell L}]$, $\func$ is increasing on $P_{hH}$ Therefore, $\anl \ge \frac{1}{1 + P_{\ell L}}$. 
\end{proof}

The following lemma guarantees the switch between the ``Steep'' Segment 2 and ``Non-Linear'' Segment 3 is numerically correct. That is, $\xinl \ge \frac{\Delta(\alpha - 1/2)}{P_{hL}}$ if and only if $\alpha \le \anl. $

\begin{lem}
    \label{lem:a4threshold}
    For $\frac{1}{2P_{\ell L}} > \alpha > \anl$, $\xinl < \frac{\Delta\cdot (\alpha -1/2)}{P_{hL}}$. For $\anl \ge \alpha > \frac12$, $\xinl \ge \frac{\Delta\cdot (\alpha -1/2)}{P_{hL}}$ 
\end{lem}

\begin{proof}
    Rewrite  $\xinl \ge \frac{\Delta\cdot (\alpha -1/2)}{P_{hL}}$ as the following formula. 

    \begin{equation*}
        (2\alpha - 1)\cdot 2\sqrt{2(1 - \alpha P_{\ell H})(1 - 2 \alpha P_{\ell L})P_{\ell H}P_{\ell L}} \le P_{hL} + (2\alpha-1)\cdot ( 4\alpha P_{\ell H}P_{\ell L}-P_{\ell H} - 2P_{\ell L})
    \end{equation*}

    The proof proceeds as follows. First, we show that when $ P_{hL} + (2\alpha-1)\cdot ( 4\alpha P_{\ell H}P_{\ell L}-P_{\ell H} - 2P_{\ell L}) \ge 0$ and $\frac{1}{2P_{\ell L}} \ge \alpha > \frac12$, the inequality holds at $\anl \ge\alpha > \frac12$. Secondly, we show that for all $\anl \ge \alpha > \frac12$, $ P_{hL} + (2\alpha-1)\cdot ( 4\alpha P_{\ell H}P_{\ell L}-P_{\ell H} - 2P_{\ell L}) \ge 0$ holds. Lastly, we show that when $\frac{1}{2P_{\ell L}} \ge \alpha > \anl$, the inequality does not hold. 

    When $ P_{hL} + (2\alpha-1)\cdot ( 4\alpha P_{\ell H}P_{\ell L}-P_{\ell H} - 2P_{\ell L}) \ge 0$, it is equivalent to solve the inequality by squaring both side, i.e., 
    \begin{equation*}
        8(2\alpha - 1)^2\cdot ((1 - \alpha P_{\ell H})(1 - 2 \alpha P_{\ell L})P_{\ell H}P_{\ell L}) \le (P_{hL} + (2\alpha-1)\cdot ( 4\alpha P_{\ell H}P_{\ell L}-P_{\ell H} - 2P_{\ell L}))^2.
    \end{equation*}

    Refactoring on $\alpha$, we get the following quadratic inequality. 
    \begin{align*}
        &\ \alpha^2 \cdot 4(P_{\ell H}^2 + 4P_{hH}P_{\ell L^2})\\
        - &\ \alpha \cdot 4(2P_{hH}P_{\ell L}^2 + 3P_{hH}P_{\ell L}  +P_{hH}^2  - 2P_{hH} - P_{\ell L} + 2) \\
        +&\ (P_{hL}^2 - 4P_{\ell H}P_{\ell L} + 4P_{\ell L}^2 + P_{\ell H}^2 + 4P_{hL}P_{\ell L} + 2P_{\ell H}P_{hL})\\
        \ge &\ 0. 
    \end{align*}

    Using the root formula of the quadratic equation, the solution of this quadratic inequality is 
    \begin{align*}
        \alpha \le&\  \frac{2P_{hH}P_{\ell L}^2 +P_{hH}\cdot (P_{hH}+3P_{\ell L})-(3P_{hH} +P_{\ell L}) + 2 - 2P_{hL}\cdot \sqrt{P_{hH}P_{hL}P_{\ell H}P_{\ell L}}}{2P_{\ell H}^2 + 8P_{hH}P_{\ell L}^2} = \anl\\
        \alpha \ge &\ \frac{2P_{hH}P_{\ell L}^2 +P_{hH}\cdot (P_{hH}+3P_{\ell L})-(3P_{hH} +P_{\ell L}) + 2 + 2P_{hL}\cdot \sqrt{P_{hH}P_{hL}P_{\ell H}P_{\ell L}}}{2P_{\ell H}^2 + 8P_{hH}P_{\ell L}^2}
    \end{align*}

    By the similar reasoning of Lemma~\ref{lem:a4upper}, we can show that 
    \begin{align*}
        &\ \frac{2P_{hH}P_{\ell L}^2 +P_{hH}\cdot (P_{hH}+3P_{\ell L})-(3P_{hH} +P_{\ell L}) + 2 + 2P_{hL}\cdot \sqrt{P_{hH}P_{hL}P_{\ell H}P_{\ell L}}}{2P_{\ell H}^2 + 8P_{hH}P_{\ell L}^2}\\
        \ge &\ \frac{2P_{hH}P_{\ell L}^2 +P_{hH}\cdot (P_{hH}+3P_{\ell L})-(3P_{hH} +P_{\ell L}) + 2 + 2P_{hL}^2\cdot P_{\ell L}}{2P_{\ell H}^2 + 8P_{hH}P_{\ell L}^2}\\
        \ge &\ \frac{1}{2P_{\ell L}}. 
    \end{align*}

    Therefore, when $ P_{hL} + (2\alpha-1)\cdot ( 4\alpha P_{\ell H}P_{\ell L}-P_{\ell H} - 2P_{\ell L}) \ge 0$ and $\frac{1}{2P_{\ell L}} \ge \alpha > \frac12$, the inequality holds at $\anl \ge \alpha > \frac12$. 

    Now we are going to show that when $\anl \ge \alpha > \frac12$, $ P_{hL} + (2\alpha-1)\cdot ( 4\alpha P_{\ell H}P_{\ell L}-P_{\ell H} - 2P_{\ell L}) \ge 0$ holds. Note that this constraint is a quadratic inequality on $\alpha$. Let $\gunc(\alpha) = P_{hL} + (2\alpha-1)\cdot ( 4\alpha P_{\ell H}P_{\ell L}-P_{\ell H} - 2P_{\ell L})$. If we can show that $\gunc(\alpha)$ is monotonically decreasing on $\alpha$ for $\frac{1}{2P_{\ell L}} \ge \alpha > \frac12$, and find a $\frac{1}{2P_{\ell L}} \ge \alpha' \ge \anl$ such that $\gunc(\alpha') \ge 0$, we show that $\gunc(\alpha) \ge 0$ for all $\anl \ge \alpha > \frac12$. Refactoring $\gunc$, we have 
    \begin{equation*}
        \gunc(\alpha) = \alpha^2 \cdot 8P_{\ell L}P_{\ell L} - \alpha \cdot 2(2P_{\ell L}P_{\ell H} + 2P_{\ell L} + P_{\ell H}) + 2P_{\ell L} + P_{\ell H} + P_{hL}. 
    \end{equation*}
    As a quadratic function, $\gunc$ is decreasing on $\alpha \le \frac14 + \frac1{4 P_{\ell H}} + \frac1{8P_{\ell L}}$. Note that 
    \begin{align*}
        \frac14 + \frac1{4 P_{\ell H}} + \frac1{8P_{\ell L}} > &\ \frac14 + \frac1{4 P_{\ell L}} + \frac1{8P_{\ell L}}\\
        \ge &\ \frac1{8P_{\ell L}} + \frac1{4 P_{\ell L}} + \frac1{8P_{\ell L}}\\
        =&\ \frac{1}{2P_{\ell L}}. 
    \end{align*}

    Therefore, $\gunc(\alpha)$ is monotonically decreasing on $\alpha$ for $\frac{1}{2P_{\ell L}} \ge \alpha > \frac12$. 

    Now we set $$\alpha' =  \frac{2P_{hH}P_{\ell L}^2 +P_{hH}\cdot (P_{hH}+3P_{\ell L})-(3P_{hH} +P_{\ell L}) + 2 - 2P_{hL}^2\cdot P_{\ell L}}{2P_{\ell H}^2 + 8P_{hH}P_{\ell L}^2}.$$

    Note that $\alpha \ge \anl$. And note that Lemma~\ref{lem:a4upper} implies $\alpha' \le \frac{1}{2P_{\ell L}}$ (by showing that $f(P_{hH})$ in Lemma~\ref{lem:a4upper} is positive).
    \begin{align*}
        \gunc(\alpha') = &\ \frac{4P_{\ell L} P_{hL}^2}{(P_{\ell H}^2 + 4P_{hH}P_{\ell L}^2)^2} \cdot ((1 - P_{hH})^4 + (1 - P_{hH})^2 \cdot (3P_{hH} - 2))\cdot P_{\ell L}\\
        &\ + 2(P_{hH}^3 +P_{hH}^2 - 4P_{hH}+2)\cdot P_{\ell L}^2 + 4(2P_{hH} - 1)\cdot P_{\ell L}^3 + 2 (1- P_{hH})\cdot P_{\ell L}^4).
    \end{align*}
    Therefore it suffices to show that for all $\frac12 < P_{\ell L} <1$ and $P_{hL} < P_{hH} \le P_{\ell L}$, $\func (P_{hH}, P_{\ell L}) =$
    $$(1 - P_{hH})^4 + (1 - P_{hH})^2 \cdot (3P_{hH} - 2))\cdot P_{\ell L}
         + 2(P_{hH}^3 +P_{hH}^2 - 4P_{hH}+2)\cdot P_{\ell L}^2 + 4(2P_{hH} - 1)\cdot P_{\ell L}^3 + 2 (1- P_{hH})\cdot P_{\ell L}^4 \ge 0. $$ 

    To solve this problem, we calculate $\func$'s higher derivatives on $P_{\ell L}$ and show that all of them are strictly positive from higher ones to lower ones. Specifically, we have 
    \begin{align*}
        \frac{\partial\func}{\partial P_{\ell L}} =&\ (1 - P_{hH})^2 \cdot (3P_{hH} - 2))
         + 4(P_{hH}^3 +P_{hH}^2 - 4P_{hH}+2)\cdot P_{\ell L} + 12(2P_{hH} - 1)\cdot P_{\ell L}^2 + 8 (1- P_{hH})\cdot P_{\ell L}^3.\\
         \frac{\partial^2\func}{\partial P_{\ell L}^2} =&\ 4(P_{hH}^3 +P_{hH}^2 - 4P_{hH}+2) + 24(2P_{hH} - 1)\cdot P_{\ell L} + 24(1- P_{hH})\cdot P_{\ell L}^2.\\
         \frac{\partial^3\func}{\partial P_{\ell L}^3} =&\ 24(2P_{hH} - 1)+ 48(1- P_{hH})\cdot P_{\ell L}.
    \end{align*}
    
    Firstly, note that $\frac{\partial^3\func}{\partial P_{\ell L}^3} = 24(1 - P_{hH}) \cdot (2P_{\ell L} - 1) \ge 0$. 
    If we want to show that $\frac{\partial^2\func}{\partial P_{\ell L}^2}$ is positive, under constraint $P_{\ell L} > P_{hH}$ and $P_{hH} > (1 - P_{\ell L})$, it suffices to show that $\frac{\partial^2\func}{\partial P_{\ell L}^2}$ is positive when $P_{\ell L} = P_{hH}$ and when $P_{\ell L} = 1 - P_{hH}$.
    
    $\frac{\partial^2\func}{\partial P_{\ell L}^2}(P_{hH}, P_{\ell L} = 1 - P_{hH}) = 4(1 - P_{hH}) \cdot (5(1 - P_{hH})- 8(1 - P_{hH}) + 5) = 4(1 - P_{hH}) \cdot (5P_{hH}^2 +2(1 - P_{hH})) > 0$.
    \begin{align*}
        \frac{\partial^2\func}{\partial P_{\ell L}^2}(P_{hH}, P_{\ell L} = P_{hH}) =&\  4\cdot (-5P_{hH}^3 + 19P_{hH}^2 - 10 P_{hH} + 2)\\
        =&\ 4\cdot (2(2P_{hH} - 1)^2 -2P_{hH} \cdot (2P_{hH} - 1)^2 + 3P_{hH}^2 + 3P_{hH}^3) \\
        >&\ 0. 
    \end{align*}

    Therefore, $\frac{\partial^2\func}{\partial P_{\ell L}^2} \ge 0$ for all feasible $P_{hH}$ and $P_{\ell L}$. Now we show that $\frac{\partial\func}{\partial P_{\ell L}} \ge 0$ for all feasible $P_{hH}$ and $P_{\ell L}$ by showing that $\frac{\partial\func}{\partial P_{\ell L}}(P_{hH}, P_{\ell L} = 1 - P_{hH}) \ge 0$ and  $\frac{\partial\func}{\partial P_{\ell L}}(P_{hH}, P_{\ell L} = P_{hH}) \ge 0$. 
    \begin{align*}
        \frac{\partial\func}{\partial P_{\ell L}}(P_{hH}, P_{\ell L} = 1 - P_{hH}) =&\  (1 - P_{hH})^2 \cdot (4(1 - P_{hH})^2 - 11 (1 - P_{hH}) + 9)\\
        =&\ P_{\ell L}^2 \cdot((2(1 - P_{hH}) -3)^2 + 1 - P_{hH})\\
        >&\ 0.
    \end{align*}
    \begin{align*}
        \frac{\partial\func}{\partial P_{\ell L}}(P_{hH}, P_{\ell L} = P_{hH}) =&\ (9P_{hH} - P_{hH}^2 - 2)\cdot (4P_{hH}^2 - 3P_{hH} +1)\\
        \ge &\ (9/2 - 3) \cdot ((2P_{hH}-1)^2 + P_{hH})\\
 \ge &\ 0. 
    \end{align*}

Therefore, $\frac{\partial\func}{\partial P_{\ell L}} \ge 0$ for all feasible $P_{hH}$ and $P_{\ell L}$. Now we show that $\func(P_{hH}, P_{\ell L}) \ge 0$ for all feasible $P_{hH}$ and $P_{\ell L}$ by showing that $\func(P_{hH}, P_{\ell L} = 1 - P_{hH}) \ge 0$ and  $\func(P_{hH}, P_{\ell L} =  P_{hH}) \ge 0$. 

\begin{align*}
    \func(P_{hH}, P_{\ell L} = 1 - P_{hH}) =&\ (2 + P_{hH})\cdot 1 - P_{hH}^3 >0\\
    \func(P_{hH}, P_{\ell L} =  P_{hH})=&\ (4P_{hH}^2 - 3P_{hH} + 1)^2 > 0. 
\end{align*}

Therefore, for all $\anl \ge \alpha > \frac12$, $ P_{hL} + (2\alpha-1)\cdot ( 4\alpha P_{\ell H}P_{\ell L}-P_{\ell H} - 2P_{\ell L}) \ge 0$ holds. This implies that $\xinl \ge \frac{\Delta\cdot (\alpha -1/2)}{P_{hL}}$ holds for all $\anl \ge \alpha > \frac12$. 

Finally, we show that for all $\frac{1}{2P_{\ell L}}\ge \alpha > \anl $, $\xinl < \frac{\Delta\cdot (\alpha -1/2)}{P_{hL}}$. Note that $\gunc(\alpha)$ is monotone on $\frac{1}{2P_{\ell L}}\ge \alpha > \frac12$, $\gunc(\anl ) > 0$, and 
\begin{align*}
    \gunc(\frac{1}{2P_{\ell L}}) = -\frac{(1 - P_{\ell L})\cdot (P_{\ell L} +P_{hH} -1)}{P_{\ell L}} < 0. 
\end{align*}
By the intermediate value theorem, there exists a unique solution $\anl < \alpha'' \le \frac{1}{2P_{\ell L}}$ such that $\gunc(\alpha'') = 0$. For all $\alpha'' < \alpha \le \frac{1}{2P_{\ell L}}$, $\gunc(\alpha) < 0$. This directly implies that $\xinl < \frac{\Delta\cdot (\alpha -1/2)}{P_{hL}}$. For all $\anl < \alpha < \alpha''$, The first part of this proof implies that $\xinl < \frac{\Delta\cdot (\alpha -1/2)}{P_{hL}}$. This finishes the proof. 
\end{proof}

% \subsection{Case 1: $\alpha > \thdmaj$.}
\citet{deng2024aggregation} shows that when the fraction of the majority type agents is above the threshold, a good equilibrium always exists. In the equilibrium, the majority agents play the following ``optimal strategy''. A majority agent with $h$ signal votes for $\bA$ with probability $\bph = \frac{1}{2(P_{\ell L} + P_{\ell H})}$, and a majority agent with $\ell$ signal votes for $\bA$ with probability $\bpl = 0$. 

\begin{coro}\citep{deng2024aggregation}
\label{coro:optimal}
    When the fraction of the majority $\alpha > \thdmaj$, for any $1 \le \kd \le \ag$, any sequence of strategy profiles where all the majority type agents play the optimal strategy profile forms a $\varepsilon$-ex-ante Bayesian $\kd$-strong equilibrium with $\varepsilon$ converges to 1. Such sequence of equilibria leads to the informed majority decision with probability converging to 1. 
\end{coro}

By Lemma~\ref{lem:dominated}, the optimal strategy majority agents play is not weakly dominated. Therefore, when minority agents play any non-weakly-dominated strategy, such equilibrium satisfies the constraint where no agents play weakly dominated strategy. 

\subsection{Proof for ``Flat'' Segment 1.} In this section, it is sufficient to show the following two statements. Proposition~\ref{prop:seg2upper} shows that when $\alpha \le \thdmaj$ and $\kd = (1 - \thdmaj)\cdot \ag$, even constraint (2) and (3) are not compatible. Proposition~\ref{prop:seg2lower} shows that when $\alpha \ge \frac{\infdif}{1 - \thdmaj}$ and $\kd < (1 - \thdmaj)\cdot \ag$, a strategy profile satisfying all three constraints exists. (Note that given that $\theta = \frac12 + \frac{P_{\ell H}}{2P_{\ell L}}$, $1 - \theta = \frac{\Delta}{2P_{\ell L}}$ and $\frac{1}{2P_{\ell L}} = \frac{\Delta}{1 - \theta}$.)

\begin{propnb}{prop:seg2lower}
    For all $\alpha \ge \frac{\infdif}{1 - \thdmaj}$, for any $\xi < 1 - \thdmaj$ and $\kd = \xi\cdot \ag$, there exists a strategy profile sequence that satisfies all three constraints. 
\end{propnb}

\begin{proof}
    For each $\ag$, let $\stgp_\ag$ be the following. All the minority agents always vote for $\bA$. The majority agents play the following strategy.
    \begin{align*}
        \bpl^{\maj}  =&\ \frac1{\alpha} (\alpha - \frac12 - \frac{P_{hL}}{\infdif}\cdot \xi) - \delta \cdot (P_{h H} + P_{h L}).\\
        \bph^{\maj}  =&\ \frac1{\alpha} (\alpha - \frac12 + \frac{P_{\ell L}}{\infdif}\cdot \xi) +  \delta \cdot (P_{\ell H} + P_{\ell L}).
    \end{align*}
    Here $\delta > 0$ is a constant such that $1 \ge \bpl^{\maj} \ge 0$ and $1 \ge \bph^{\maj} \ge 0$, and at least one of $\bpl^{\maj} = 0$ and $\bph^{\maj} = 1$ holds. We first show that such $\delta$ exists. 

    For $\bpl^{\maj}$, $\frac1{\alpha} (\alpha - \frac12 - \frac{P_{hL}}{\infdif}\cdot \xi) <\frac1{\alpha} \cdot \alpha < 1$. 
    Since $\alpha \ge \frac{\infdif}{1 - \thdmaj} = \frac{1}{2} + \frac{P_{hL}}{2P_{\ell L}}$,  there is $\frac{\alpha - 1/2}{P_{hL}} \ge \frac{1}{2P_{\ell L}}$. 
    Moreover, given that $\xi < 1- \thdmaj = \frac12 - \frac{P_{\ell H}}{2P_{\ell L}} = \frac{\infdif}{2P_{\ell L}}$, 
    \begin{align*}
        \alpha - \frac12 - \frac{P_{hL}}{\infdif}\cdot \xi >&\ \alpha - \frac12 - \frac{P_{hL}}{\infdif}\cdot  \frac{\Delta}{2P_{\ell L}}\\
        =&\ \alpha - \frac12-  \frac{P_{hL}}{2P_{\ell L}}\\
        \ge&\ 0. 
    \end{align*}

    For $\bph^{\maj}$, $\alpha - \frac12 + \frac{P_{\ell L}}{\infdif}\cdot \xi \ge \alpha - \frac12 > 0$. Given that $\alpha \ge \frac{1}{2} + \frac{P_{hL}}{2P_{\ell L}}$ and that $\xi < \frac{\infdif}{2P_{\ell L}}$, 
    \begin{align*}
        \frac1{\alpha} (\alpha - \frac12 + \frac{P_{\ell L}}{\infdif}\cdot \xi) < \frac{1}{\alpha}(\alpha - \frac12 +  \frac{P_{\ell L}}{2P_{\ell L}}) = \frac{1}{\alpha}\cdot \alpha = 1.
    \end{align*}
    Therefore, such $\delta  > 0$ exists. 

    Now we are ready to show that the sequence $\{\stgp_\ag\}$ satisfies all three constraints. According to Lemma~\ref{lem:dominated}, it is not hard to verify that all agents are playing a non-weakly dominated strategy. For the fidelity, we compute the expected vote share for $\bA$ of the profile, denoted by $\thd_H$ and $\thd_L$. 
    \begin{align*}
        \thd_H = &\ \alpha \cdot \delta \cdot \infdif + \xi + \frac12 > \frac12 + \xi.\\
        \thd_L = &\ \frac12 - \alpha \cdot \delta \cdot \infdif < \frac12. 
    \end{align*}
    Moreover, the difference $\alpha \cdot \delta$ does not depend of $\ag$. Therefore, $\liminf \sqrt{\ag} \cdot (\thd_H - \frac12) = + \infty$ and $\liminf \sqrt{\ag} \cdot (\frac12 - \thd_L) = + \infty$. 
    By applying the first case of Theorem~\ref{thm:arbitrary}, the fidelity of $\stgp_\ag$ converges to 1. 
    
    % Hoeffding Inequality, both $\lp_H^{\bA}(\stgp_\ag) \ge 1 - 2\exp(-(\alpha \cdot \delta \cdot \infdif)^2 \cdot \ag)$ and $\lp_L^{\bR}(\stgp_\ag) \ge 1 - 2\exp(-(\alpha \cdot \delta \cdot \infdif)^2 \cdot \ag)$. This implies that the fidelity of $\stgp_\ag$ converges to 1. 

    Now we show that $\stgp_\ag$ is an equilibrium. Let $e = 2\exp(-(\alpha \cdot \delta \cdot \infdif)^2 \cdot \ag)$ and $\stgp'_\ag$  be a different strategy profile such that at most $\xi\cdot \ag$ agents deviate. Applying Lemma~\ref{lem:acctoeq}, we know that either no agents receive more than $\frac{2B(B+1)e}{P_HP_L}$ utility increase or minority agents will not join the deviating group. 

    Now suppose the deviation group $D$ contains only minority types of agents. Let $\thd_L'$ and $\thd_H'$ be the \exshare{} of $\stgp'_\ag$. Since in $\stgp_\ag$ all minorities always vote for $\bA$, any deviation will increase the vote share of $\bR$. Therefore, $\thd_L' \ge \thd_L^\ag = \alpha \cdot \delta \cdot \infdif$. Moreover, since there are at most $\xi\cdot \ag$ deviators, $\thd_H' \ge \thd_H^\ag - \xi = \alpha \cdot \delta \cdot \infdif$. By the Hoeffding Inequality, $\lp_H^{\bA}(\stgp'_\ag) \ge 1 - e$ and $\lp_L^{\bR}(\stgp'_\ag) \ge 1 - e$. Therefore, $\acc(\stgp'_\ag) \ge 1 - e$, which contradicts the assumption that $\acc(\stgp'_\ag) < 1 - \frac{(B+1)e}{P_H\cdot P_L}$. Therefore, such deviation cannot succeed, and $\stgp_\ag$ is an $\frac{2B(B+1)e}{P_HP_L}$-ex-ante Bayesian $\kd$-strong equilibrium. 
\end{proof}

\begin{propnb}{prop:seg2upper}
    For all $\alpha \le \thdmaj$, for any profile sequence $\{\stgp_\ag\}$ such that $\acc(\stgp_\ag)$ converges to 1 and no agents play weakly dominated strategies, there is a constant $\varepsilon > 0$ and constantly many $\ag$ such that $\stgp_\ag$ is NOT an $\varepsilon$-ex-ante Bayesian $\kd$-strong equilibrium, where $\kd = (1 - \thdmaj)\cdot \ag$. 
\end{propnb}

% \begin{prop}
% \label{prop:seg2upper}
%     For all $\alpha \le \thdmaj$, for any profile sequence $\{\stgp_\ag\}$ such that $\acc(\stgp_\ag)$ converges to 1 and no agents play weakly dominated strategies, there is a constant $\varepsilon > 0$ and constantly many $\ag$ such that $\stgp_\ag$ is NOT an $\varepsilon$-ex-ante Bayesian $\kd$-strong equilibrium, where $\kd = (1 - \thdmaj)\cdot \ag$. 
% \end{prop}

\begin{proof}
    In this case, we show that a group of $(1 - \thdmaj)\cdot \ag$ of minority type agents can successfully deviate by always voting for $\bA$ or always voting for $\bR$ simultaneously. Suppose $\{\stgp\}$ is a strategy profile sequence such that $\acc(\stgp)$ converges to 1.
    
    For every $\ag$, we construct two deviating strategy profile sequences $\{\stgp^1\}$ and $\{\stgp^2\}$. 
    Let $D$ be the deviating group containing $(1 - \thdmaj)\cdot \ag$ minority type agents. Let $\astg = \sum_{i\not\in D} \stg_i = (\abpl, \abph)$ be the average of strategies of non-deviators, which includes all the majority agents and the $(1 - \alpha - \thdmaj)\cdot \ag$ minority agents. Note that $(\abpl, \abph)$ may or may not be identical in different $\ag$ in the sequence. 
    
    In $\{\stgp^1\}$, all the deviators vote for $\bR$. Then the expected vote share for $\bA$ of $\stgp^1$ is 
    \begin{align*}
        \thd_H^1 =&\ \thdmaj\cdot (P_{hH}\cdot \abph + P_{\ell H}\cdot \abpl),\\
        \thd_L^1 =&\ \thdmaj\cdot (P_{hL}\cdot \abph + P_{\ell L}\cdot \abpl).
    \end{align*}

    In $\{\stgp^2\}$, all the deviators vote for $\bA$. Then the expected vote share for $\bA$ of $\stgp^2$ is 
    \begin{align*}
        \thd_H^2 =&\ \thdmaj\cdot (P_{hH}\cdot \abph + P_{\ell H}\cdot \abpl)+ (1 - \thdmaj),\\
        \thd_L^2 =&\ \thdmaj\cdot (P_{hL}\cdot \abph + P_{\ell L}\cdot \abpl)+ (1 - \thdmaj).
    \end{align*}
    
    % Let $\thd_H'$ be the expected vote share (for $\bA$) in the world state $H$ in the first deviating strategy, i.e., $$\thd_H' =\thdmaj\cdot (P_{hH}\cdot \abph + P_{\ell H}\cdot \abpl).$$ In the second deviation, all the deviators vote for $\bA$. Let $\thd_L'$ be the expected vote share (for $\bA$) in the world state $L$ in the second deviating strategy. i.e., $$\thd_L' = \thdmaj\cdot (P_{hL}\cdot \abph + P_{\ell L}\cdot \abpl) + (1 - \thdmaj).$$

    We first consider the case where $\{\stgp\}$ has infinite $\ag$ such that the majority agent strategies in $\stgp$ have sub-linear variance. Suppose for any $\delta > 0$ there exists infinitely many $\ag$ such that in some world state $\wos$, $Var(\sum_{i\in\maj} \xrv_{i}^{\ag} \mid \wos) < \delta \cdot \ag$. Firstly, $\{\stgp^1\}$ and $\{\stgp^2\}$ also have infinite $\stgp$ where majority agent strategies has sub-linear variance. This is because, the majority agents do no change their strategies in these two deviations. Then by Corollary~\ref{coro:var_sublinear}, 
    then for any $\varepsilon > 0$, there exists infinitely many $\ag$ such that in some world state $\wos$, $\thd_H^1 - \thd_L^1 < \varepsilon$ and $\thd_H^2 - \thd_L^2 < \varepsilon$. (Note that for every such $\ag$, $\thd_H^1 - \thd_L^1 < \varepsilon$ and $\thd_H^2 - \thd_L^2 < \varepsilon$ holds simultaneously.)

    Then we show that in this case, at least one of $\{\stgp^1\}$ and $\{\stgp^2\}$ falls into the second case in Theorem~\ref{thm:arbitrary}, which indicates that its fidelity does not converges to 1. If $\liminf_{\ag} \sqrt{\ag}\cdot (\thd_H^1 - \frac12) < 0$, $\{\stgp^1\}$ falls into the second case in Theorem~\ref{thm:arbitrary}. Suppose $\liminf_{\ag} \sqrt{\ag}\cdot (\thd_H^1 - \frac12) \ge 0$, we show that $\liminf_{\ag} \sqrt{\ag}\cdot (\frac12 - \thd_L^2) <0$. $\liminf_{\ag} \sqrt{\ag}\cdot (\thd_H^1 - \frac12) \ge 0$ implies that, for any constant $\delta' > 0$ for all sufficiently large $\ag$, $\thd_H^1 - \frac12 > -\delta'$. Then by Corollary~\ref{coro:var_sublinear}, there exists infinitely many $\ag$ such that in some world state $\wos$, $\thd_H^1 - \thd_L^1 < \varepsilon$ and $\thd_H^2 - \thd_L^2 < \varepsilon$. Note that $\thd_H^1 = \thd_H^2 - (1 - \thdmaj)$. Combining there three fact, and setting the constant such that $\delta' + \varepsilon < \frac{1 -\thdmaj}{2}$.  we know that there exists infinitely many $\ag$ such that 
    \begin{equation*}
        \thd_L^2 > \thd_H^2 - \varepsilon = \thd_H^1 - \varepsilon + (1 - \theta) > \frac12 - \delta' - \varepsilon + (1 - \theta) > \frac12 +\frac{1 -\thdmaj}{2}.
    \end{equation*}
    This directly implies that $\liminf_{\ag} \sqrt{\ag}\cdot (\frac12 - \thd_L^2) <0$. Therefore, the fidelity of either $\{\stgp^1\}$ or $\{\stgp^2\}$ does not converges to 1. Precisely, there exists an $\delta > 0$ and infinitely many $\ag$ such that $\acc(\stgp^1) \le 1 - \delta$ or $\acc(\stgp^2) \le 1 - \delta$ holds. As the fidelity decreases, the utility of the minority agents increases, and the deviation is succcessful. This implies that $\{\stgp\}$ is not a sequence of $\varepsilon$-ex-ante $\kd$-strong BNE with $\varepsilon$ converges to 0. 
    % Let $\thd_H$ and $\thd_{L}$ be expected vote share for $\bA$ from all the non-deviators under world state $H$ and $L$, respectively. Then,
    % \begin{align*}
    %     \thd_H =&\ \thdmaj\cdot (P_{hH}\cdot \abph + P_{\ell H}\cdot \abpl).\\
    %     \thd_L = &\ \thdmaj\cdot (P_{hL}\cdot \abph + P_{\ell L}\cdot \abpl).
    % \end{align*}
    
    For now we suppose that the majority agent strategies $\stgp$ have linear variance for all sufficiently large $\ag$. Likewise, the majority agent strategies in $\stgp^1$ and $\stgp^2$ also have linear variances. In this case, if we can show that for all $\ag$, either $\thd_H^1 \le \frac12$ or $\thd_L^2 \ge \frac12$ holds, then at least one of them holds for infinitely many $\ag$. By the second or the third case (as the variance is linear) of Theorem~\ref{thm:arbitrary}, either the fidelity of $\stgp^1$ or the fidelity of $\stgp^2$ does not converge to 1. This implies that the deviation is successful, and $\{\stgp\}$ is not a sequence of $\varepsilon$-ex-ante $\kd$-strong BNE with $\varepsilon$ converges to 0. 

    It remains to show that $\thd_H^1 \le \frac12$ or $\thd_L^2 \ge \frac12$ holds for an arbitrary $\ag$. Let $\thd_H' = \thd_H^1$ and $\thd_L' = \thd_L^2$. We fix an arbitrary $\ag$. 
    % If $\thd_H' \le \frac{1}{2}$, when all the deviators always vote for $\bR$, and the expected votes for $\bA$ in the deviating strategy $\stgp'$ will be no more than $\frac{1}{2}$. Then by applying the Hoeffding Inequality or Berry Esseen Theorem, We can show that $\lp_{H}^{\bA}(\stgp') \le 1 - \delta$ from some constant $\delta$, and the minority type deviator can gain a constant utility from deviation. 
    Suppose $\thd_H' > \frac{1}{2}$, we show that $\thd_L'> \frac12$ so that the deviation can succeed by all the deviators always voting for $\bA$. Given that $\thdmaj = \frac{P_{\ell L} + P_{\ell H}}{2P_{\ell L}}$ this is equivalent to show that 
    \begin{equation*}
        P_{hL} \cdot \abph + P_{\ell L}\cdot \abpl > \frac{P_{\ell H}}{P_{\ell L} + P_{\ell H}}
    \end{equation*}
    under the condition that
    \begin{equation*}
        P_{hH} \cdot \abph + P_{\ell H}\cdot \abpl > \frac{P_{\ell L}}{P_{\ell L} + P_{\ell H}}.
    \end{equation*}
    Denote $\hthd_H' =  P_{hH} \cdot \abph + P_{\ell H}\cdot \abpl$ and $\hthd_L' =  P_{hL} \cdot \abph + P_{\ell L}\cdot \abpl$. 
    Now, let function $\gunc(\bph) = P_{hL} \cdot \bph + P_{\ell L}\cdot \frac{\hthd'_H - P_{hH}\cdot \bph}{P_{\ell H}} - \frac{P_{\ell H}}{P_{\ell L} + P_{\ell H}}$. Then it is sufficient to show that $\gunc(\abph) > 0. $

    Firstly, $\frac{d \gunc}{d \bph} = \frac{P_{hL}\cdot P_{\ell H} - P_{hH}\cdot P_{\ell L}}{P_{\ell H}} < 0$. Therefore, $\gunc(\bph)$ is minimized when $\bph$ is maximized. 
    Therefore, 
    \begin{align*}
        \gunc(\abph) \ge &\ \gunc(1)\\
        =&\ P_{hL} + P_{\ell L}\cdot \frac{\hthd'_H - P_{hH}\cdot 1}{P_{\ell H}} - \frac{P_{\ell H}}{P_{\ell L} + P_{\ell H}}.
    \end{align*}
    Recall that $\hthd'_H >  \frac{P_{\ell L}}{P_{\ell L} + P_{\ell H}}$. Therefore,  
    \begin{align*}
        \hthd'_H - P_{hH}> &\ \frac{P_{\ell L}}{P_{\ell L} + P_{\ell H}} - P_{hH}\\
        =&\ \frac{(P_{\ell L} -P_{hH} \cdot P_{\ell L}) - P_{hH}\cdot P_{\ell H}}{P_{\ell L} + P_{\ell H}}\\
        =&\ \frac{P_{\ell L}\cdot P_{\ell H} - P_{hH}\cdot P_{\ell H}}{P_{\ell L} + P_{\ell H}}\\
        =&\ \frac{(P_{\ell L} - P_{hH})\cdot P_{\ell H}}{P_{\ell L} + P_{\ell H}}.
    \end{align*}
    Therefore, 
    \begin{align}
        \gunc(\abph) > &\ \frac{P_{\ell L}}{P_{\ell H}}\cdot \frac{(P_{\ell L} - P_{hH})\cdot P_{\ell H}}{P_{\ell L} + P_{\ell H}} + (P_{hL} - \frac{P_{\ell H}}{P_{\ell L} + P_{\ell H}})\\
        =&\  \frac{(P_{\ell L} - P_{hH})\cdot P_{\ell L}}{P_{\ell L} + P_{\ell H}} + \frac{(P_{h L} - P_{\ell H})\cdot P_{\ell L}}{P_{\ell L} + P_{\ell H}}\\
        =&\ 0. 
    \end{align}
    Therefore, $\hthd'_L > \frac{P_{\ell H}}{P_{\ell L} + P_{\ell H}}$.
    Therefore, by applying the second or the third case of Theorem~\ref{thm:arbitrary}, either the fidelity of $\stgp_\ag^1$ or that of $\stgp_\ag^2$ does not converge to 1. By Lemma~\ref{lem:minority_deviate}, minority can gain a constant utility increase from at least one of such deviation under infinitely many $\ag$. Therefore, $\{\stgp_\ag\}$ fail to be an equilibrium sequence with $\varepsilon$ converges to 0. 
\end{proof}

% \begin{prop}
% \label{prop:seg2lower}
%     For all $\alpha \ge \frac{\infdif}{1 - \thdmaj}$, for any $\xi < 1 - \thdmaj$ and $\kd = \xi\cdot \ag$, there exists a strategy profile sequence that satisfies all three constraints. 
% \end{prop}
\subsection{Segment 2 to 4: Overview}
The rest of the proof will proceed in the order of the proof sketch. Here is the outline. 
\begin{enumerate}
    \item We solve the case when $\gamma$ is biased (either $\gamma$  or $1- \alpha - \gamma$ is close to 0) and restrict the case when $\gamma$ is not biased to be ``double-symmetric'': all type-0 agents play the same strategy, and all type-1 agents play the same strategy.
    \item We show that a sequence of strategy profiles satisfy all constraint on $\xi$ if and only if $\xi< \xi_h$ and $\xi< \xi_\ell$. Here $\xi_h$ and $\xi_\ell$ are two constraints derived by the constraints on the expected vote share ($\thd_H' > \frac12$ and $\thd_L < \frac 12$, like that in the ``Flat'' Segment 1). Therefore, the problem of finding the upper and lower bound is converted into maximizing $\min(\xi_h, \xi_\ell)$.
    \item We reduce the search space of maximizing $\min(\xi_h, \xi_\ell)$ into two corner cases. In the first case, all type-0 agents always vote for $\bR$. In the second case, all type-1 agents always vote for $\bA$. The threshold $\xi^*$ take the maximum of these two cases and $\frac{\Delta(\alpha - \frac12)}{P_{hL}}$, which derives from the biased case. 
    % \item  We show that $\min(\xi_h, \xi_\ell)$ is maximized in two extreme cases: all type-0 agents always vote for $\bR$, or all type-1 agents always vote for $\bA$.
    \item We tackle the two extreme cases respectively. In the case where all type-0 agents always vote for $\bR$ (the first case), the extension of ``Steep'' Segment 2 and ``Tail'' Segment 4 serves as an upper bound. On the other hand, the case where all type-1 agents always vote for $\bA$ (the second case) derives the ``Non-Linear'' Segment 3 and the ``Tail'' Segment 4. 
    \item We take the maximum of all the cases and reach the upper and the lower bounds. 
\end{enumerate}

We first introduce notations. Consider a strategy profile $\stgp$ where no agents play weakly dominated strategies. For the majority type agents, we only care about their contribution to the expected vote share, so it suffices to know the average of their strategies. Let $\bph^{\maj}$  and $\bpl^{\maj}$ be the average strategy of the majority agents. 
We partition minority agents into two subgroups. The type-0 minority agents always vote for $\bR$ when they receive $h$ signals, i.e. $\bph^0 = 0$. The type-1 minority agents always vote for $\bA$ when they receive $\ell$ signals, i.e. $\bpl^1 = 1$. Each minority agent will only be categorized into exactly one type. Specifically, those who vote anti-informatively will either be a type-1 or type-0 agent, but will not be both. Let $\bpl^0$ and $\bph^1$ be the average probability of voting for $\bA$ for a type-0 agent receiving signal $\ell$  and a type-1 agent receiving signal $h$, respectively. Let $0\le \gamma\le 1 - \alpha$ be the fraction of type-0 agents. 

In this way, the expected vote share of the strategy profile can be written as 
    \begin{align*}
        \thd_H =&\ \alpha\cdot (P_{hH}\cdot \bph^{\maj} + P_{\ell H}\cdot \bpl^{\maj}) +  (1 - \alpha - \gamma)\cdot (P_{hH}\cdot \bph^{1} + P_{\ell H}) + \gamma \cdot P_{\ell H}\cdot \bpl^0.\\
        \thd_L =&\ \alpha\cdot (P_{hL}\cdot \bph^{\maj} + P_{\ell L}\cdot \bpl^{\maj}) +  (1 - \alpha - \gamma)\cdot (P_{hL}\cdot \bph^{1} + P_{\ell L}) + \gamma \cdot P_{\ell L}\cdot \bpl^0.
    \end{align*}

\subsection{Segment 2 to 4: Preliminaries}
We first show that, when $\gamma$ is very biased, no strategy profile can tolerate a deviation group of $\xi \ge \frac{\Delta \cdot (\alpha - 1/2)}{P_{hL}}$.

\begin{propnb}{prop:biasedgamma_gen}
    For $\alpha \le \frac{\infdif}{1 - \thdmaj}$, for any $\xi' \ge \frac{\Delta \cdot (\alpha - 1/2)}{P_{hL}}$, no sequence with infinitely many strategy profiles where $\gamma \le \xi'$ or $1 - \alpha - \gamma \le \xi'$ can satisfy all three constraint in Theorem~\ref{thm:thresholdk} with $\xi = \xi'$. 
\end{propnb}

% \begin{prop}
% \label{prop:biasedgamma_gen}
%     For $\alpha \le \frac{\infdif}{1 - \thdmaj}$, for any $\xi' \ge \frac{\Delta \cdot (\alpha - 1/2)}{P_{hL}}$, no sequence with infinitely many strategy profiles where $\gamma \le \xi'$ or $1 - \alpha - \gamma \le \xi'$ can satisfy all three constraint in Theorem~\ref{thm:thresholdk} with $\xi = \xi'$. 
% \end{prop}

\begin{proof}
     The high-level idea is the deviator group exhausts one type of minority agent. Suppose $\stgp$ is a strategy profile with $\gamma \le \xi'$. (The $1 - \alpha - \gamma \le \xi'$ side follows similar reasoning given $\frac{\Delta \cdot (\alpha - 1/2)}{P_{hL}} \ge \frac{\Delta \cdot (\alpha - 1/2)}{P_{\ell H}}$).  Let the deviation group contain all the type-0 agents and $(\xi - \gamma)\cdot n$ of type-1 agents. We consider two possible deviations. 

    In the first deviation, denoted as $\stgp^1$, the type-1 agents in the deviation group are those with minimum probability to vote for $\bA$ when receiving $h$ signal. All the deviators switch to always vote for $\bA$. We denote the expected share for $\stgp^1$ as $\thd_H^1$ and $\thd_L^1.$

    In the second deviation, denoted as $\stgp^2$, the type-1 agents in the deviation group are those with maximum probability to vote for $\bA$ when receiving $h$ signal. All the deviators switch to always vote for $\bR$.We denote the expected share for $\stgp^2$ as $\thd_H^2$ and $\thd_L^2.$ 
    
    Note that $\thd_H^1 \ge \thd_H^2 + \xi'$ and $\thd_L^1 \ge \thd_L^2 + \xi'$. Firstly, all the deviators (with a fraction of $\xi'$) in $\stgp^1$ vote always for $\bA$, while all the deviators in $\stgp^2$ always vote for $\bR$. Secondly, the majority agents play the same strategy in two deviations. Finally, the non-deviator minority agents in $\stgp^1$ are more likely to vote for $\bA$ than the non-deviator minority agents. 

    Similar to the proof of Proposition~\ref{prop:seg2upper},
    we first consider the case where $\{\stgp\}$ has infinite $\stgp$ where majority agent strategies have sub-linear variance. Suppose for any $\delta > 0$ there exists infinitely many $\ag$ such that in some world state $\wos$, $Var(\sum_{i \in \maj} \xrv_{i}^{\ag} \mid \wos) < \delta \cdot \ag$. Firstly, $\{\stgp^1\}$ and $\{\stgp^2\}$ also have infinite $\stgp$ where majority agent strategies have sub-linear variance. Then by Corollary~\ref{coro:var_sublinear}, 
    then for any $\varepsilon > 0$, there exists infinitely many $\ag$ such that in some world state $\wos$, $\thd_H^1 - \thd_L^1 < \varepsilon$ and $\thd_H^2 - \thd_L^2 < \varepsilon$. (Note that for every such $\ag$, $\thd_H^1 - \thd_L^1 < \varepsilon$ and $\thd_H^2 - \thd_L^2 < \varepsilon$ holds simultaneously.)

    Then we show that in this case, at least one of $\{\stgp^1\}$ and $\{\stgp^2\}$ falls into the second case in Theorem~\ref{thm:arbitrary}, which indicates that its fidelity does not converges to 1. If $\liminf_{\ag} \sqrt{\ag}\cdot (\thd_H^2 - \frac12) < 0$, $\{\stgp^2\}$ falls into the second case in Theorem~\ref{thm:arbitrary}. Suppose $\liminf_{\ag} \sqrt{\ag}\cdot (\thd_H^2 - \frac12) \ge 0$, we show that $\liminf_{\ag} \sqrt{\ag}\cdot (\frac12 - \thd_L^1) <0$. $\liminf_{\ag} \sqrt{\ag}\cdot (\thd_H^2 - \frac12) \ge 0$ implies that, for any constant $\delta' > 0$ for all sufficiently large $\ag$, $\thd_H^2 - \frac12 > -\delta'$. Then by Corollary~\ref{coro:var_sublinear}, there exists infinitely many $\ag$ such that in some world state $\wos$, $\thd_H^1 - \thd_L^1 < \varepsilon$ and $\thd_H^2 - \thd_L^2 < \varepsilon$. Note that $\thd_H^1 \ge \thd_H^2 - \xi'$. Combining there three fact, and setting the constant such that $\delta' + \varepsilon < \xi'$.  we know that there exists infinitely many $\ag$ such that 
    \begin{equation*}
        \thd_L^1 > \thd_H^1 - \varepsilon = \thd_H^2 - \varepsilon + \xi'> \frac12 - \delta' - \varepsilon + \xi' > \frac12 +\frac{\xi'}{2}.
    \end{equation*}
    This directly implies that $\liminf_{\ag} \sqrt{\ag}\cdot (\frac12 - \thd_L^2) <0$. Therefore, the fidelity of either $\{\stgp^1\}$ or $\{\stgp^2\}$ does not converges to 1. Precisely, there exists an $\delta > 0$ and infinitely many $\ag$ such that $\acc(\stgp^1) \le 1 - \delta$ or $\acc(\stgp^2) \le 1 - \delta$ holds. As the fidelity decreases, the utility of the minority agents increases, and the deviation is succcessful. This implies that $\{\stgp\}$ is not a sequence of $\varepsilon$-ex-ante $\kd$-strong BNE with $\varepsilon$ converges to 0. 

    For now we suppose that the majority agent strategies in $\stgp$ (and consequently $\stgp^1$ and $\stgp^2$) have linear variance for all sufficiently large $\ag$. In this case, if we can show that for all $\ag$, either $\thd_H^2 \le \frac12$ or $\thd_L^1 \ge \frac12$ holds, then at least one of them holds for infinitely many $\ag$. By the second or the third case of Theorem~\ref{thm:arbitrary}, either the fidelity of $\stgp^1$ or the fidelity of $\stgp^2$ does not converge to 1. By Lemma~\ref{lem:minority_deviate}, this implies that the deviation is successful, and $\{\stgp\}$ is not a sequence of $\varepsilon$-ex-ante $\kd$-strong BNE with $\varepsilon$ converges to 0. 

    By the definition of $\stgp^1$ and $\stgp^2$ ,we have 
    \begin{equation*}
        \thd_L^1 \ge \alpha\cdot (P_{hL}\cdot \bph^{\maj} + P_{\ell L}\cdot \bpl^{\maj}) +  (1 - \alpha - \xi)\cdot (P_{h L}\cdot \bph^{1} + P_{\ell L}) + \xi.
    \end{equation*}
    \begin{align*}
        \thd_H^2 \le&\ \alpha\cdot (P_{hH}\cdot \bph^{\maj} + P_{\ell H}\cdot \bpl^{\maj}) +  (1 - \alpha - \xi)\cdot (P_{hH}\cdot \bph^{1} + P_{\ell H}).
    \end{align*}

    We show that either $\thd_H^2 \le \frac12$ or $\thd_L^1 \ge \frac12$, which implies that the deviating succeed, and the strategy profile is not an equilibrium. 

    Note that the solution to equation
    \begin{align*}
       \alpha\cdot (P_{hL}\cdot \bph^{\maj} + P_{\ell L}\cdot \bpl^{\maj}) +  (1 - \alpha - \xi)\cdot (P_{h L}\cdot \bph^{1} + P_{\ell L}) + \xi =&\ \frac12 \\
        \alpha\cdot (P_{hH}\cdot \bph^{\maj} + P_{\ell H}\cdot \bpl^{\maj}) +  (1 - \alpha - \xi)\cdot (P_{hH}\cdot \bph^{1} + P_{\ell H}) =&\ \frac12
    \end{align*}
    is 
    \begin{align*}
    \bph^{maj} = &\ \frac{1}{\alpha} \left(\frac12 + \frac{P_{\ell H}\cdot \xi}{\Delta} - (1 - \alpha - \xi)\cdot \bph^1\right). \\
    \bpl^{maj} = &\ \frac{1}{\alpha} \left(\frac12 - \frac{P_{hH}\cdot \xi}{\Delta} - (1 - \alpha - \xi)\right).
\end{align*}
    Denote this solution as $\bph^{*}$ and $\bpl^*$. Note that when $\xi \ge \frac{\Delta \cdot (\alpha - 1/2)}{P_{hL}}$, $\bpl^* \le 0$
    
    Then we rewrite the inequality of $\thd_L^1$ and $\thd_H^2$ as 
    \begin{align*}
        \thd_H^2 \le&\ \frac12 + \alpha\cdot (P_{hH}\cdot (\bph^{\maj} - \bph^*) + P_{\ell H}\cdot (\bpl^{\maj} - \bpl^{*})).\\
        \thd_L^1 \ge &\ \frac12 + \alpha\cdot (P_{hL}\cdot (\bph^{\maj} - \bph^*) + P_{\ell L}\cdot (\bpl^{\maj} - \bpl^{*})).
    \end{align*}
    If $\thd_H^2 > \frac12$, there must be $P_{hH}\cdot (\bph^{\maj} - \bph^*) + P_{\ell H}\cdot (\bpl^{\maj} - \bpl^{*}) > 0$. By $P_{hH} > P_{hL}$, $P_{\ell L} > P_{\ell H}$, and $(\bpl^{\maj} - \bpl^{*}) \ge 0$, there must be $(P_{hL}\cdot (\bph^{\maj} - \bph^*) + P_{\ell}\cdot (\bpl^{\maj} - \bpl^{*}) > 0$. Therefore, $\thd_L^1 \ge \frac12$.

    Therefore, there exists a constant $\varepsilon$ such that for every strategy $\stgp$ in the sequence $\{\stgp\}$ such that $\gamma \le \xi'$ or $1 - \alpha - \gamma \le \xi'$, $\stgp$ is NOT a $\varepsilon$-ex-ante Bayesian $\xi'\cdot \ag$ strong equilibrium. 
\end{proof}

At the same time, the biased case also contributes a tight lower bound, which eventually becomes the lower bound for Segment 2. 

\begin{propnb}{prop:seg3lower}
    When $\frac{\infdif}{1 - \thdmaj} \ge \alpha$, for any $\xi < \frac{\Delta(\alpha - 1/2)}{P_{hL}}$, there exists a sequence of strategy profiles $\{\stgp\}$ that satisfies the constraints in Theorem~\ref{thm:thresholdk} with $\xi  = \frac{\Delta(\alpha - 1/2)}{P_{hL}}$
\end{propnb}

% \begin{prop}
%     \label{prop:seg3lower}
%     When $\frac{\infdif}{1 - \thdmaj} \ge \alpha \ge \frac{1}{1 + P_{\ell L}}$, for any $\xi < \frac{\Delta(\alpha - 1/2)}{P_{hL}}$, there exists a sequence of strategy profiles $\{\stgp\}$ that satisfies the constraints in Theorem~\ref{thm:thresholdk} with $\xi  = \frac{\Delta(\alpha - 1/2)}{P_{hL}}$
% \end{prop}
\begin{proof}
    For a fixed $\xi < \frac{\Delta(\alpha - 1/2)}{P_{hL}}$ and for every $\ag$, the strategy profile is as follows. All the minority agents always vote for $\bA$ regardless of their signals. All the majority agents play $\bph^{\maj} = 1 + \frac1{\alpha} \cdot (-\frac12 + \frac{\xi \cdot P_{\ell L}}{\Delta}) + \delta$ and $\bpl^{\maj} = 0$, where $\delta  \in (\frac{\xi + (\alpha -1/2)\cdot P_{\ell H}}{\alpha \cdot P_{hH}} - \frac{\xi \cdot P_{\ell L}}{\alpha \cdot \Delta}, \frac{(\alpha -1/2)\cdot P_{\ell L}}{\alpha \cdot P_{hL}} - \frac{\xi \cdot P_{\ell L}}{\alpha \cdot \Delta})$. Given that $\xi < \frac{\Delta(\alpha - 1/2)}{P_{hL}}$, we can verify that the feasible set of $\delta$ is not empty.

    \begin{align*}
        \frac{(\alpha -1/2)\cdot P_{\ell L}}{\alpha \cdot P_{hL}} -\frac{\xi + (\alpha -1/2)\cdot P_{\ell H}}{\alpha \cdot P_{hH}} > &\ \frac{(\alpha -1/2)\cdot P_{\ell L}}{\alpha \cdot P_{hL}} -\frac{(\alpha -1/2)\cdot P_{\ell H}}{\alpha \cdot P_{hH}} - \frac{\Delta(\alpha - 1/2) \cdot P_{\ell L}}{\alpha\cdot P_{hL} \cdot P_{hH}}\\
        =&\ \frac{(\alpha - 1/2)}{\alpha}\cdot \frac{P_{\ell L} \cdot P_{hH} - \Delta - P_{\ell H} \cdot P_{hL}}{P_{hL} \cdot P_{hH}}\\
        =&\ 0. 
    \end{align*}

    It is not hard to verify that $\bph^{\maj} \ge 0$.     
    Given that $\alpha \le \frac{\infdif}{1 - \thdmaj} = \frac{1}{2P_{\ell L}}$ and $\xi < \frac{\Delta(\alpha - 1/2)}{P_{hL}}$, we can verify that $ \bph^{\maj} \le 1$. 

    \begin{equation*}
        \bph^{\maj} < 1 + \frac1{\alpha} \cdot (-\frac12 + \frac{(\alpha -1/2)\cdot P_{\ell L}}{P_{hL}})
        = 1 + \frac1{\alpha} \cdot \frac{2\alpha \cdot P_{\ell L} - 1}{2P_{hL}} \le 1.
    \end{equation*}

    By Lemma 1, no agents play weakly dominated strategies. For constraint 2, we calculated the expected vote share for $\bA$ under two world states respectively. 
    \begin{align*}
        \thd_H =&\ \alpha \cdot P_{hH} \cdot (1 + \frac1{\alpha} \cdot (-\frac12 + \frac{\xi \cdot P_{\ell L}}{\Delta}) + \delta) + (1 - \alpha). \\
        \thd_L =&\ \alpha\cdot P_{hL} \cdot (1 + \frac1{\alpha} \cdot (-\frac12 + \frac{\xi \cdot P_{\ell L}}{\Delta}) + \delta) + (1 - \alpha).
    \end{align*}

    For $\thd_H$ we have 
    \begin{align*}
        \thd_H > &\ \alpha \cdot P_{hH} \cdot (1 + \frac1{\alpha} \cdot (-\frac12 + \frac{\xi \cdot P_{\ell L}}{\Delta}) + \frac{\xi + (\alpha -1/2)\cdot P_{\ell H}}{\alpha \cdot P_{hH}} - \frac{\xi \cdot P_{\ell L}}{\alpha \cdot \Delta}) + (1 - \alpha)\\
        =&\ P_{hH} \cdot (\alpha -\frac12 + \frac{\xi + (\alpha -1/2)\cdot P_{\ell H}}{P_{hH}}) + (1 - \alpha)\\
        =&\ P_{hH} \cdot (\alpha -\frac12) + (\alpha -1/2)\cdot P_{\ell H} + \xi + (1 - \alpha)\\
        =&\ \frac12 + \xi. 
    \end{align*}
    And for $\thd_L$, we have 
    \begin{align*}
        \thd_H < &\ \alpha \cdot P_{hL} \cdot (1 + \frac1{\alpha} \cdot (-\frac12 + \frac{\xi \cdot P_{\ell L}}{\Delta}) + \frac{(\alpha -1/2)\cdot P_{\ell L}}{\alpha \cdot P_{hL}} - \frac{\xi \cdot P_{\ell L}}{\alpha \cdot \Delta}) + (1 - \alpha)\\
        =&\ P_{hL} \cdot (\alpha - 1/2 + \frac{(\alpha -1/2)\cdot P_{\ell L}}{P_{hL}}) + (1 - \alpha) \\
        =&\ \frac12.
    \end{align*}
    Moreover, $\thd_H$ and $\thd_L$ does not depend on $\ag$. Therefore, $\liminf \sqrt{\ag} \cdot (\thd_H - \frac12) = + \infty$ and $\liminf \sqrt{\ag} \cdot (\frac12 - \thd_L) = + \infty$. 
    By applying the first case of Theorem~\ref{thm:arbitrary}, the fidelity of $\stgp_\ag$ converges to 1. 

    Now we show that $\stgp$ is an equilibrium with $\xi^* = \frac{\Delta\cdot (\alpha - 1/2)}{P_{hL}}$. Note that $\xi^* \le 1 -\alpha$ when $\alpha \le \thdmaj$. By Lemma~\ref{lem:acctoeq}, it suffices to show that a fraction $\xi < \xi^*$ of minority agents cannot successfully deviate. Suppose a $\xi < \xi^*$ fraction of minority agents change their strategies. Consider the expected vote share after deviation, denoted as $\thd_H'$ and $\thd_L'$. 
    Recall that all minority agents always vote for $\bA$ in $\stgp$. Therefore, $\thd_L' \le \thd_L < \frac12$. On the other hand, the deviator can decrease $\thd_H$ by no more than its fraction $\xi$. Therefore, $\thd_H' \ge \thd_H - \xi > \frac12$. 
    Therefore, applying the first case of Theorem~\ref{thm:arbitrary} and Lemma~\ref{lem:minority_deviate}, the fidelity still converges to 1 after any deviation, and the expected utility gain of the minority agents converges to 0. Therefore, the strategy profile proposed is an $\varepsilon$ equilibrium with $\varepsilon$ converges to 0. 
\end{proof}

Now we come to solve the case where $\gamma > \xi$ and $1 - \alpha - \gamma > \xi$, i.e.,  the deviation group cannot exhuast one type of minority agents. The following lemma shows that in the strategy profile maximizing $\xi$, each type of minority agents should play the same strategy.

\begin{lemnb}{lem:double}
    Let $\{\stgp_\ag\}$ be a sequence of strategy profiles that satisfies the constraints in Theorem~\ref{thm:thresholdk} with $\xi$. Consider the following profiles $\{\stgp'_\ag\}$, where each $\ag$, $\stgp_\ag$ and $\stgp'_\ag$ has the following relationship. 
    \begin{enumerate}
        \item Each majority agent in $\stgp'_\ag$ play the same strategy as in $\stgp_\ag$.
        \item All the type-1 minority agents in $\stgp'_\ag$ play the average strategy of type-1 minority agents in $\stgp_\ag$, denoted by $\bpl^1 = 1$ and $\bph^1$. 
        \item All the type-0 minority agents in $\stgp'_\ag$ play the average strategy of type-0 minority agents in $\stgp_\ag$, denoted by $\bph^0 = 0$ and $\bpl^0$. 
    \end{enumerate}
    Then $\{\stgp'_\ag\}$ also satisfies the constraints on $\xi$. 
\end{lemnb}

% \begin{lem}
%     \label{lem:double}
%     Let $\{\stgp_\ag\}$ be a sequence of strategy profiles that satisfies the constraints in Theorem~\ref{thm:thresholdk} with $\xi$. Consider the following profiles $\{\stgp'_\ag\}$, where each $\ag$, $\stgp_\ag$ and $\stgp'_\ag$ has the following relationship. 
%     \begin{enumerate}
%         \item Each majority agent in $\stgp'_\ag$ play the same strategy as in $\stgp_\ag$.
%         \item All the type-1 minority agents in $\stgp'_\ag$ play the average strategy of type-1 minority agents in $\stgp_\ag$, denoted by $\bpl^1 = 1$ and $\bph^1$. 
%         \item All the type-0 minority agents in $\stgp'_\ag$ play the average strategy of type-0 minority agents in $\stgp_\ag$, denoted by $\bph^0 = 0$ and $\bpl^0$. 
%     \end{enumerate}
%     Then $\{\stgp'_\ag\}$ also satisfies the constraints on $\xi$. 
% \end{lem}
\begin{proof}
    The first constraint is straightforward. 
    
    % The second constraint comes from the fact that two strategy profiles share the same expected vote share ($\thd_H = \thd'_H, \thd_L = \thd'_L$). Now we show the third constraint holds. We fix an arbitrary $\ag$ and use $\stgp$ and $\stgp'$ to denote $\stgp_\ag$ and $\stgp'_\ag$. 
For the second and the third constraint, we consider the expected vote share. Let $\thd_H$ and $\thd_L$ be the expected vote share of $\stgp_\ag$ and $\thd_H'$ and $\thd_L'$ be the expected vote share of $\stgp'_\ag$. For simplicity, we use $\stgp$ and $\stgp'$ to denote $\stgp_\ag$ and $\stgp'_\ag$.

For an arbitrary $\ag$, the extremities of expected vote share in $\stgp$. For the $H$ side, $\thd_H$ is minimized by selecting $\xi\cdot \ag$ agents in type-1 agents with maximum probability to vote for $\bA$ with signal $h$ and switching them to always vote for $\bR$. The change bringing to $\thd_H$ is $- \frac1\ag \sum_{i\in D} (P_{hH}\cdot \bph^i + P_{\ell H}) \le -\xi\cdot (P_{hH}\cdot \bph^1 + P_{\ell H})$. For the $L$ side, $\thd_L$ is maximized by selecting $\xi\cdot \ag$ agents in type-0 agents with minimum probability to vote for $\bA$ with signal $\ell$ and switching them to always vote for $\bA$. The change bringing to $\thd_L$ is $\xi - \frac1\ag \sum_{i\in D} (P_{hL}+ P_{\ell L} \cdot \bpl^i) \ge \xi - \xi\cdot (P_{hL} + P_{\ell L}\cdot \bpl^0)$. 

Given that $\stgp$ satisfies constraint 3, any deviation from $\stgp$ still have fidelity converging to 1 (otherwise the deviation succeed). By Theorem~\ref{thm:arbitrary}, for any constant $\varepsilon > 0$ and all sufficiently large $\ag$,  there must be 
    \begin{align*}
        \thd_H - \frac1\ag \sum_{i\in D} (P_{hH} \cdot \bph^i + P_{\ell H}) \ge&\ \frac12 - \varepsilon\\
        \thd_L + \xi - \frac1\ag \sum_{i\in D} (P_{hL} + P_{\ell L} \cdot \bpl^i) \le &\ \frac 12 + \varepsilon. 
    \end{align*}
By taking $\varepsilon < \frac{\xi}{4}$, this implies that $\thd_H - \thd_L \ge \frac{\xi}{2}$ for all sufficiently large $\ag$. Then Lemma~\ref{lem:var_linear}, the variance of the majority agent's strategy is already $\Theta(\ag)$ for all sufficiently large $\ag$. Therefore, for any deviation from $\stgp_\ag$, the variance is $\Theta(\ag)$. Given that $\stgp_\ag$ satisfies the third constraint (equilibrium), these deviation fails, and their fidelities converges to 1. This rules out the third case in Theorem~\ref{thm:arbitrary} and further show that 
\begin{align*}
    \liminf \sqrt{\ag} \cdot ( \thd_H - \frac1\ag \sum_{i\in D} (P_{hH} \cdot \bph^i + P_{\ell H}) - \frac12) = &\ +\infty,\\
    \liminf \sqrt{\ag} \cdot (\frac12 - (\thd_L + \xi - \frac1\ag \sum_{i\in D} (P_{hL} + P_{\ell L} \cdot \bpl^i))) =&\  +\infty.
\end{align*}
This also implies that $\liminf \sqrt{\ag} \cdot ( \thd_H -\frac12) = +\infty$ and $\liminf \sqrt{\ag} \cdot  (\frac12 - \thd_L) = +\infty$.
    
% This implies that $\thd_H \ge \frac12$, $\thd_L < \frac12$, and $\thd_H - \thd_L \ge \xi$ for all sufficiently large $\ag$. By Lemma~\ref{lem:var_linear}, the variance of the majority agent's strategy is already $\Theta(\ag)$ for all sufficiently large $\ag$. 

% By Corollary~\ref{coro:var_sublinear}, there exists a constant $\delta > 0$ such that the variance of $\stgp$ is at lease $\delta \cdot \ag$ for all sufficiently large $\ag$. Given that the fidelity of $\stgp_\ag$ converges to 1, by Theorem~\ref{thm:arbitrary}, this implies that 

Then we consider the $\stgp_\ag'$. We start from the second constraint (fidelity converging to 1). Note that $\thd_H' = \thd_H$ and $\thd_L' = \thd_L$ for all $\ag$. Therefore, $\liminf \sqrt{\ag} \cdot (\thd_H' -\frac12) = +\infty$ and $\liminf \sqrt{\ag} \cdot  (\frac12 - \thd_L') = +\infty$ also holds.
Therefore, by the first case in Theorem~\ref{thm:arbitrary}, the fidelity of $\stgp_\ag$ converges to 1.    

Now we prove the third constraint holds by showing that no deviation in $\stgp'_\ag$ can succeed.
Consider the extremities for $\stgp'_\ag$. For the $H$ side, $\thd'_H$ is minimized by selecting $\xi\cdot \ag$ agents in type-1 (where they play the same strategy). Therefore, the change bringing to $\thd'_H$ is at most $-\xi\cdot (P_{hH}\cdot \bph^1 + P_{\ell H})$. Likewise, the change bringing to $\thd'_L$ is at most $ \xi - \xi\cdot (P_{hL} + P_{\ell L}\cdot \bpl^0)$. Given that $\thd_H = \thd'_H$ and $\thd_L = \thd'_L$, there must be 
    \begin{align*}
        \thd'_H -\xi\cdot (P_{hH}\cdot \bph^1 + P_{\ell H}) \ge&\ \thd_H - \frac1\ag \sum_{i\in D} (P_{hH} \cdot \bph^i + P_{\ell H}) \\
        \thd'_L + \xi - \xi\cdot (P_{hL} + P_{\ell L}\cdot \bpl^0) \le &\ \thd_L + \xi - \frac1\ag \sum_{i\in D} (P_{hL} + P_{\ell L} \cdot \bpl^i) . 
    \end{align*}
    % This implies that $\{\stgp'_\ag\}$ also implies constraint 3. 
    Therefore, it also holds that
    \begin{align*}
    \liminf \sqrt{\ag} \cdot (  \thd'_H -\xi\cdot (P_{hH}\cdot \bph^1 + P_{\ell H}) - \frac12) = &\ +\infty,\\
    \liminf \sqrt{\ag} \cdot (\frac12 - (\thd'_L + \xi - \xi\cdot (P_{hL} + P_{\ell L}\cdot \bpl^0))) =&\  +\infty.
\end{align*}
Therefore, for any deviating strategy profile $\stgp''_\ag$ from $\stgp'_\ag$, whose expected vote share is $\thd_H''$ and $\thd_L''$, there must be $\liminf \sqrt{\ag} \cdot \min(\thd_H'' - \frac12, \frac12 - \thd_L'') = +\infty$. By the first case of Theorem~\ref{thm:arbitrary} and Lemma~\ref{lem:minority_deviate}, the fidelity of $\stgp_\ag''$ converges to 1, and the deviation cannot succeed. Therefore, $\{\stgp'_\ag\}$ satisfies the third constraint. 
% This implies that $\thd_H' \ge \frac12$, $\thd_L' < \frac12$, and $\thd_H' - \thd_L' \ge \xi$ for all sufficiently large $\ag$. By Corollary~\ref{coro:var_sublinear}, there exists a constant $\delta > 0$ such that the variance of $\stgp'$ is at lease $\delta \cdot \ag$ for all sufficiently large $\ag$. Therefore, applying the first case of Theorem~\ref{thm:arbitrary}, we shall conclude that the fidelity of $\{\stgp'_\ag\}$ converge to 1 (constraint 2 satisfied). 
\end{proof}

Therefore, if we show that any strategy profile where each type of minority agents plays the same strategy does not satisfy the constraints, then no strategy profile can satisfy the constraints. 
In the rest of the proof, we focus on strategy profiles where each type of minority agents plays the same strategy.

\subsection{Segment 2 to 4: Transition to Maximization Problem. }
Consider a strategy profile $\stgp$ where the average strategy of majority agents is $(\bpl^{\maj}, \bph^{\maj})$, every type-1 minority agents play $(1, \bph^1)$, every type-0 minority agents play $(\bpl^0, 0)$, and the fraction of type-0 minority agents is $\gamma$. 

Consider the condition for $\stgp$ to be an equilibrium under $\xi$. First, after a $\xi$ fraction of type-0 agents switch to always vote for $\bA$, the deviating expected vote share under world state $L$ (denoted as $\thd'_L$) should be strictly less than $\frac12$. Second, after a $\xi$ fraction of type-1 agents switch to always vote for $\bR$, the deviating expected vote share under world state $H$ (denoted as $\thd'_H$) should be strictly greater than $\frac12$. Therefore, we have the following conditions. 
\begin{align*}
        \thd'_H = \alpha\cdot (P_{hH}\cdot \bph^{\maj} + P_{\ell H}\cdot \bpl^{\maj}) +  (1 - \alpha - \gamma - \xi)\cdot (P_{hH}\cdot \bph^{1} + P_{\ell H}) + \gamma \cdot P_{\ell H}\cdot \bpl^0 > &\ \frac12\\
        \thd'_L = \alpha\cdot (P_{hL}\cdot \bph^{\maj} + P_{\ell L}\cdot \bpl^{\maj}) +  (1 - \alpha - \gamma)\cdot (P_{hL}\cdot \bph^{1} + P_{\ell L}) + (\gamma - \xi) \cdot P_{\ell L}\cdot \bpl^0 +\xi <&\ \frac12
    \end{align*}
We claim that, for any given $\bph^1$, $\bpl^0$, $\gamma$, and $\xi$, there is a pair of $(\bpl^\maj, \bph^\maj)$ satisfying this condition if and only if the solution of $\thd'_H = \frac12$ and $\thd'_L = \frac12$ (denoted as $(\bpl^*, \bph^*)$) is in $(0, 1)^2$. Specifically, the explicit form of the solution is 
    \begin{align*}
        \bpl^* =&\ \frac{1}{\alpha} (\frac12 - (1-\alpha)  + \gamma \cdot ( 1 - \bpl^0)) - \frac{\xi}{\Delta\cdot \alpha}(P_{hH}\cdot P_{hL}\cdot \bph^1 + P_{hH}\cdot P_{\ell L} \cdot (1 - \bpl^0) + P_{hL}).\\
        \bph^* = &\ \frac{1}{\alpha} (\frac12 - (1 - \alpha -\gamma) \cdot \bph^1) + \frac{\xi}{\Delta\cdot \alpha}(P_{hH}\cdot P_{\ell L}\cdot \bph^1 + P_{\ell H}\cdot P_{\ell L} \cdot (1 - \bpl^0) + P_{\ell H}).
    \end{align*}
Combining with the constraint that $(\bpl^*, \bph^*)\in (0,1)^2$, we get two constraints on $\xi$, denoted as $\xi_\ell$ and $\xi_h$. 
    \begin{align*}
        \xi < &\ \xi_h =  \frac{\Delta(\alpha - 1/2 + (1 - \alpha - \gamma) \cdot \bph^1)}{P_{hH}\cdot P_{\ell L}\cdot \bph^1 + P_{\ell H}\cdot P_{\ell L} \cdot (1 - \bpl^0) + P_{\ell H}}.\\
        \xi < &\ \xi_{\ell} = \frac{\Delta (\alpha - 1/2 + \gamma \cdot (1 - \bpl^0))}{P_{hH}\cdot P_{hL}\cdot \bph^1 + P_{hH}\cdot P_{\ell L} \cdot (1 - \bpl^0) + P_{hL}}.
    \end{align*}

Specifically, when $\bph^1 = (1 - \bpl^0) = 0$, $\xi_h = \frac{\Delta(\alpha - \frac12)}{P_{\ell H}}$ and $ \xi_\ell = \frac{\Delta(\alpha - \frac12)}{P_{hL}}$. When $\gamma = \frac{1 - (P_{\ell L} + P_{\ell H})\cdot\alpha}{2}$, and $\bph^1 = (1 - \bpl^0) = 1$, $\xi_h = \xi_{\ell} = \frac{1}{2}\Delta \alpha$. 

\begin{lemnb}{lem:xilower}
    Given a set of $\bph^1, \bpl^0$, and $\gamma$ satisfying $\bph^1 \in [0,1]$, $\bpl^0 \in [0, 1]$, and let $\xi' = min(\xi_h, \xi_\ell)$ corresponding to $\bph^1, \bpl^0$, and $\gamma$. If $\gamma \ge \xi'$ and $1 - \alpha - \gamma \ge \xi'$ holds, then for all $\xi < \xi'$ there exists a sequence of profiles such that the constraints in Theorem~\ref{thm:thresholdk} holds for $\xi$. 
\end{lemnb}

% \begin{lem}
%     \label{lem:xilower}
%     Given a set of $\bph^1, \bpl^0$, and $\gamma$ satisfying $\bph^1 \in [0,1]$, $\bpl^0 \in [0, 1]$, and let $\xi' = min(\xi_h, \xi_\ell)$ corresponding to $\bph^1, \bpl^0$, and $\gamma$. If $\gamma \ge \xi'$ and $1 - \alpha - \gamma \ge \xi'$ holds, then for all $\xi < \xi'$ there exists a sequence of profiles such that the constraints in Theorem~\ref{thm:thresholdk} holds for $\xi$. 
% \end{lem}

\begin{proof}
    We fix an arbitrary $\xi < \xi'$, and let $\bph^*$ and $\bpl^*$ defined as above. By the definition of $\xi_h$ and $\xi_\ell$, $\bph^*$ and $\bpl^*$ are in $(0, 1)$. Now we set $\bph^\maj = \bph^* +\delta \cdot (P_{\ell H}+P_{\ell L})$ and $\bpl^\maj = \bpl^* - \delta (P_{hH} +P_{hL})$, where $\delta > 0$ guarantees that $\bpl^\maj$ and $\bph^\maj$ are in $[0, 1]$, and either $\bph^\maj = 1$ or $\bpl^\maj = 0$ holds. 

    Now we consider a sequence of strategy profile ${\stgp}$. For each $\ag$, $\stgp$ is a follows. All majority agents play $\bph^\maj$ and $\bpl^\maj$. A fraction $\gamma$ of minority agents play $\bph^0 = 0$ and $\bpl^0$. The rest of the minority agents play $\bpl^1 = 1$ and $\bph^1$. We show such strategy satisfies all three constraints. 

    First, it is not hard to verify that all agents play non-weakly dominated strategy. Secondly, we consider the fidelity of this strategy. In this proof, $\thd_H'$ and $\thd_L'$ refers two deviating expected vote share defined as above. Note that 
    \begin{align*}
        \thd_H =&\ \thd_H' + \xi \cdot (P_{hH}\cdot \bph^{1} + P_{\ell H}) > \thd_H'.\\
        \thd_L =&\  \thd_L' - \xi \cdot ( 1 - P_{\ell L}\cdot \bpl^0) < \thd_L'.
    \end{align*}
    Therefore, it is sufficient to show that $\thd_H' > \frac12$ and $\thd_L' < \frac12$. 
    Recall that $\bph^*$ and $\bpl^*$ is the solution of the set of equation ($\thd_H' = \frac12$, $\thd_L' = \frac12$). Therefore, assigning $\bph^\maj$ and $\bpl^\maj$, 
    \begin{align*}
        \thd_H' =&\ \frac12+ \alpha\cdot \delta \cdot (P_{hH}\cdot (P_{\ell H}+P_{\ell L}) - P_{\ell H} \cdot (P_{hH} +P_{hL})) =\frac12+ \alpha \cdot \delta \cdot \Delta > \frac12\\
        \thd_L' =&\ \frac12+\alpha\cdot \delta \cdot (P_{hL}\cdot (P_{\ell H}+P_{\ell L}) - P_{\ell L} \cdot (P_{hH} +P_{hL})) =\frac12- \alpha \cdot \delta \cdot \Delta < \frac12
    \end{align*}

    Therefore, for the second constraint, we show that $\thd_H > \frac12$ and $\thd_L < \frac12$, and that they do not depend on $\ag$. By the first case of Theorem~\ref{thm:arbitrary}, the fidelity of $\stgp$ converges to 1.

    For the third constraint, by Lemma~\ref{lem:acctoeq}, it suffices to consider a deviating group with only minority agents. for any deviating strategy with deviating expected vote share $\thd_H''$ and $\thd_L''$, there must be $\thd_H'' \ge \thd_H'$ and $\thd_L'' \le \thd_L'$. Given that $\thd_H' > \frac12$ and $\thd_L' < \frac12$ and do not depend on $\ag$, this directly implies $\liminf \sqrt{\ag} \cdot \min(\thd_H'' - \frac12, \frac12 - \thd_L'') = +\infty$. By the first case of Theorem~\ref{thm:arbitrary} and Lemma~\ref{lem:minority_deviate}, the fidelity of $\stgp_\ag''$ converges to 1, and the deviation cannot succeed. Therefore, $\{\stgp'_\ag\}$ satisfies the third constraint. 
\end{proof}

\begin{lemnb}{lem:xiupper}
    Let $\xi' = \max_{\bph^1, \bpl^0, \gamma} \min(\xi_h, \xi_\ell)$. No sequence strategy profile can satisfy all three constraints in Theorem~\ref{thm:thresholdk} with $\xi = \max(\xi', \frac{\Delta(\alpha - 1/2)}{P_{hL}})$. 
\end{lemnb}

% \begin{lem}
%     \label{lem:xiupper}
%     Let $\xi' = \max_{\bph^1, \bpl^0, \gamma} \min(\xi_h, \xi_\ell)$. No sequence strategy profile can satisfy all three constraints in Theorem~\ref{thm:thresholdk} with $\xi = \max(\xi', \frac{\Delta(\alpha - 1/2)}{P_{hL}})$. 
% \end{lem}
\begin{proof}
    Let $\{\stgp\}$ be a sequence strategy profile satisfying constraint 1 and 2.  We show that constraint 3 does not hold, i.e., there exists a constant $\varepsilon$ such that infinitely many strategy profiles in $\{\stgp\}$ are not $\varepsilon$-ex-ante Bayesian $\xi\cdot \ag$ strong equilibrium. 

    When $\{\stgp\}$ contains infinitely many strategy profiles where $\gamma \le \xi'$ or $1 - \alpha - \gamma \le \xi'$, we apply Proposition~\ref{prop:biasedgamma_gen}. If $\xi' < \frac{\Delta(\alpha - 1/2)}{P_{hL}}$, $\{\stgp\}$ also contains infinitely many strategy profiles where $\gamma \le \frac{\Delta(\alpha - 1/2)}{P_{hL}}$ or $1 - \alpha - \gamma \le \frac{\Delta(\alpha - 1/2)}{P_{hL}}$, and Proposition~\ref{prop:biasedgamma_gen} implies that the constraint is violated with $\xi = \frac{\Delta(\alpha - 1/2)}{P_{hL}}$; if $\xi' \ge \frac{\Delta(\alpha - 1/2)}{P_{hL}}$, Proposition~\ref{prop:biasedgamma_gen} implies that the constraint is violated with $\xi = \xi' $. Therefore, $\{\stgp\}$ cannot satisfies all three constraints with $\xi = \max(\xi', \frac{\Delta(\alpha - 1/2)}{P_{hL}})$.

    No we consider the case where $\gamma > \xi'$ and $1 - \alpha - \gamma > \xi'$ holds for all sufficiently large $\ag$. By Lemma~\ref{lem:double}, without loss of generality, we can assume that all type-0 agents play the same strategy and all type-1 agents play the same strategy. For an arbitrary $\ag$, let $\thd_H$ and $\thd_L$ be the expected vote share of $\stgp$. We consider two deviations. In the first deviation $\stgp^1$, a $\xi'$ fraction of type-0 agents switch to always vote for $\bA$. We denote the expected vote share after this deviation as $\thd_H^1$ and $\thd_L^1$. In the second deviation $\stgp^2$, a $\xi'$ fraction of type-1 agents switch to always vote for $\bR$. We denote the expected vote share after this deviation as $\thd_H^2$ and $\thd_L^2$. Note that $\thd_H'$ and $\thd_L'$ defined in Equation~\ref{eq:devthdH} and~\ref{eq:devthdL} is exactly $\thd_H' = \thd_H^1$ and $\thd_L' = \thd_L^2$. We will show that for infinitely many $\ag$, either $\stgp^1$ or $\stgp^2$ is a successful deviation. 

    Now we do case study on the variance of $\stgp$. We first consider the case where $\{\stgp\}$ has infinite $\stgp$ where the majority agent strategies have sub-linear variance. This part of proof resembles that of Proposition~\ref{prop:seg2upper}. Suppose for any $\delta > 0$ there exists infinitely many $\ag$ such that in some world state $\wos$, $Var(\sum_{i\in\maj} \xrv_{i}^{\ag} \mid \wos) < \delta \cdot \ag$. Firstly, $\{\stgp^1\}$ and $\{\stgp^2\}$ also has infinite $\stgp$ where the majority agent strategies have sub-linear variance. Then by Corollary~\ref{coro:var_sublinear}, 
    then for any $\varepsilon > 0$, there exists infinitely many $\ag$ such that in some world state $\wos$, $\thd_H^1 - \thd_L^1 < \varepsilon$ and $\thd_H^2 - \thd_L^2 < \varepsilon$. (Note that for every such $\ag$, $\thd_H^1 - \thd_L^1 < \varepsilon$ and $\thd_H^2 - \thd_L^2 < \varepsilon$ holds simultaneously.)

    Then we show that in this case, at least one of $\{\stgp^1\}$ and $\{\stgp^2\}$ falls into the second case in Theorem~\ref{thm:arbitrary}, which indicates that its fidelity does not converges to 1. If $\liminf_{\ag} \sqrt{\ag}\cdot (\thd_H^1 - \frac12) < 0$, $\{\stgp^1\}$ falls into the second case in Theorem~\ref{thm:arbitrary}. Suppose $\liminf_{\ag} \sqrt{\ag}\cdot (\thd_H^1 - \frac12) \ge 0$, we show that $\liminf_{\ag} \sqrt{\ag}\cdot (\frac12 - \thd_L^2) <0$. $\liminf_{\ag} \sqrt{\ag}\cdot (\thd_H^1 - \frac12) \ge 0$ implies that, for any constant $\delta' > 0$ for all sufficiently large $\ag$, $\thd_H^1 - \frac12 > -\delta'$. Then by Corollary~\ref{coro:var_sublinear}, there exists infinitely many $\ag$ such that in some world state $\wos$, $\thd_H^1 - \thd_L^1 < \varepsilon$ and $\thd_H^2 - \thd_L^2 < \varepsilon$. On the other hand, note that 
    \begin{align*}
        \thd_H^1 =&\ \thd_H^2 - \xi'\cdot (P_{hH} \cdot \bph^1 + P_{\ell H}) - \xi'\cdot (1 - \bpl^0\cdot P_{\ell L})\\
        =&\ \thd_H^2 - \xi'\cdot(1 + P_{\ell H}+ P_{hH} \cdot \bph^1-\bpl^0\cdot P_{\ell L})\\
        \le &\ \thd_H^2 - \xi'\cdot (P_{hL} + P_{\ell H}). 
    \end{align*}
    Combining there three inequalities, and setting the constant such that $\delta' + \varepsilon < \frac{ \xi'\cdot (P_{hL} + P_{\ell H})}{2}$.  we know that there exists infinitely many $\ag$ such that 
    \begin{equation*}
        \thd_L^2 > \thd_H^2 - \varepsilon \ge \thd_H^1 - \varepsilon +  \xi'\cdot (P_{hL} + P_{\ell H}) > \frac12 - \delta' - \varepsilon +  \xi'\cdot (P_{hL} + P_{\ell H}) > \frac12 +\frac{ \xi'\cdot (P_{hL} + P_{\ell H})}{2}.
    \end{equation*}
    This directly implies that $\liminf_{\ag} \sqrt{\ag}\cdot (\frac12 - \thd_L^2) <0$. Therefore, the fidelity of either $\{\stgp^1\}$ or $\{\stgp^2\}$ does not converges to 1. Precisely, there exists an $\delta > 0$ and infinitely many $\ag$ such that $\acc(\stgp^1) \le 1 - \delta$ or $\acc(\stgp^2) \le 1 - \delta$ holds. By Lemma~\ref{lem:double}, the utility of the minority agents increases as the fidelity decreases, and at least on of the deviation is successful. This implies that $\{\stgp\}$ is not a sequence of $\varepsilon$-ex-ante $\kd$-strong BNE with $\varepsilon$ converges to 0. 

    Now we consider the case where $\stgp$ (and consequently $\stgp^1$ and $\stgp^2$) have majority agent strategies with $\Theta(\ag)$ variance of all sufficiently large $\ag$.
    It suffices to show that at least one of $\thd'_L \ge \frac12$ and $\thd'_H \le \frac12$
    holds for infinitely many $\ag$. By the second or the third case (as the variance is linear) of Theorem~\ref{thm:arbitrary}, either the fidelity of $\stgp^1$ or the fidelity of $\stgp^2$ does not converge to 1. This implies that the deviation is successful, and $\{\stgp\}$ is not a sequence of $\varepsilon$-ex-ante $\kd$-strong BNE with $\varepsilon$ converges to 0. 
    
    Recall that $\bph^*$ and $\bpl^*$ be the solution to $\thd_H' = \frac12$ and $\thd_L' = \frac12$. By the definition of $\xi_h$ and $\xi_\ell$, either $\bph^* \ge 1$ or $\bpl^*\le 0$ holds. Without loss of generality, let $\bph^* \ge 1$. Now we consider $\bph^\maj$ and $\bpl^\maj$. Note that $\bph^\maj \le \bph^*$ always holds. 
    \begin{align*}
        \thd_H' =&\ \frac12 + \alpha\cdot (P_{hH}\cdot (\bph^{\maj} - \bph^*) + P_{\ell H}\cdot (\bpl^{\maj} - \bpl^{*})).\\
        \thd_L' = &\ \frac12 + \alpha\cdot (P_{hL}\cdot (\bph^{\maj} - \bph^*) + P_{\ell L}\cdot (\bpl^{\maj} - \bpl^{*})).
    \end{align*}

    We show that $\thd_H' > \frac12$ implies $\thd_L' > \frac12$. By $\bph^\maj \le \bph^*$, $\thd_H' > \frac12$ implies $\bpl^{\maj} > \bpl^{*}$ and $P_{hH}\cdot (\bph^{\maj} - \bph^*) + P_{\ell H}\cdot (\bpl^{\maj} - \bpl^{*}) > 0$. By $P_{hH}> P_{hL}$ and $P_{\ell H} < P_{\ell L}$, $\thd_L' > \thd_H' > \frac12$. 
\end{proof}

\subsection{Segment 2 to 4: Maximizing $\min(\xi_h, \xi_\ell)$. }

Now we try to find the condition that maximizes $\min(\xi_h, \xi_\ell)$. We view $\xi_h$ and $\xi_{\ell}$ as functions of $\bph^1$ and $(1 - \bpl^0)$ and consider their partial derivative on $\bph^1$ and $(1 - \bpl^0)$. It is not hard to verify that $\frac{\partial \xi_h}{\partial (1 - \bpl^0)} < 0$ and $ \frac{\partial \xi_h}{\partial \bph^1} < 0$. On the other two, recall the fact that for a function $f = \frac{u}{v}$, the derivative of $f$ is $f' = \frac{u'v - uv'}{v^2}$. Therefore, we follow $\sign(f') = \sign(u'v - uv')$ and have the following observation on the derivatives. 
\begin{align*}
    \sign(\frac{\partial \xi_h}{\partial \bph^1}) =&\ \sign((1 - \alpha - \gamma)\cdot P_{\ell H}\cdot (1 + P_{\ell L}\cdot (1 - \bpl^0)) - P_{hH} \cdot P_{\ell L} \cdot (\alpha - \frac12)),\\
    \sign(\frac{\partial \xi_\ell}{\partial (1 - \bpl^0)}) =&\ \sign(\gamma \cdot P_{hL}\cdot (1 + P_{hH}\cdot \bph^1) - P_{hH}\cdot P_{\ell L} \cdot (\alpha - \frac12)).
\end{align*}

We first discuss on $\xi_h$. When $(1 - \alpha - \gamma) \le \frac{P_{hH}\cdot P_{\ell L}}{P_{\ell H}} ( \alpha - \frac12)$, $\frac{\partial \xi_h}{\partial \bph^1} \le 0$ at $(1 - \bpl^0)= 0$. Therefore, for any $(\bpl^0, \bph^1)$, $\xi_h(\bph, 1 - \bpl^0) < \xi_h(\bph^1, 0) \le \xi_h(0, 0) = \frac{\Delta(\alpha - \frac12)}{P_{\ell H}}$. On the other hand, when $(1 - \alpha - \gamma) >\frac{P_{hH}\cdot P_{\ell L}}{P_{\ell H}} ( \alpha - \frac12)$, $\frac{\partial \xi_h}{\partial \bph^1} > 0$ for all $\bpl^0 \in [0, 1]$. 
Similarly, when $\gamma \le \frac{P_{hH}\cdot P_{\ell L}}{P_{hL}} ( \alpha - \frac12)$, for any $(\bpl^0, \bph^1)$, $\xi_\ell(\bph, 1 - \bpl^0)< \xi_\ell(0, 0) = \frac{\Delta(\alpha - \frac12)}{P_{hL}}$, and when $\gamma > \frac{P_{hH}\cdot P_{\ell L}}{P_{hL}} ( \alpha - \frac12)$, $\frac{\partial \xi_\ell}{\partial (1 - \bpl^0)} > 0$ holds. 
Therefore, the following proposition holds. 

% \begin{propnb}{prop:mid_biased_gamma_gen} When either $(1 - \alpha - \gamma) \le \frac{P_{hH}\cdot P_{\ell L}}{P_{\ell H}} ( \alpha - \frac12)$ or $\gamma \le\frac{P_{hH}\cdot P_{\ell L}}{P_{hL}} ( \alpha - \frac12)$ holds, either $\xi_h$ or $\xi_{\ell}$ cannot exceed $\frac{\Delta(\alpha - \frac12)}{P_{hL}}$.
% % and $\min(\xi_h, \xi_\ell)$ is maximized at $\frac{\Delta(\alpha - \frac12)}{P_{hL}}$.
% \end{propnb}

\begin{prop}
    \label{prop:mid_biased_gamma_gen} When either $(1 - \alpha - \gamma) \le \frac{P_{hH}\cdot P_{\ell L}}{P_{\ell H}} ( \alpha - \frac12)$ or $\gamma \le\frac{P_{hH}\cdot P_{\ell L}}{P_{hL}} ( \alpha - \frac12)$ holds, either $\xi_h$ or $\xi_{\ell}$ cannot exceed $\frac{\Delta(\alpha - \frac12)}{P_{hL}}$, and $\min(\xi_h, \xi_\ell)$ is maximized at $\frac{\Delta(\alpha - \frac12)}{P_{hL}}$.
\end{prop}

Now we consider the case when $(1 - \alpha - \gamma) > \frac{P_{hH}\cdot P_{\ell L}}{P_{\ell H}} ( \alpha - \frac12)$ and $\gamma > \frac{P_{hH}\cdot P_{\ell L}}{P_{hL}} ( \alpha - \frac12)$. We deal with this case by analyzing the derivatives under a linear constraint of $\bph^1$ and $(1 - \bpl^0)$, i.e. $1 - \bpl^0 = t\cdot \bph^1$ for every $t > 0$. Then we have 
\begin{align*}
    \sign(\frac{\partial \xi_h(\bph^1, t\cdot \bph^1)}{\partial \bph^1}) =&\ \sign((1 - \alpha - \gamma) \cdot P_{\ell H} - (\alpha - \frac12) \cdot P_{hH}\cdot P_{\ell L} - (\alpha  - \frac12) \cdot P_{\ell H} \cdot P_{\ell L}\cdot t)).\\
    \sign(\frac{\partial \xi_\ell(\bph^1, t\cdot \bph^1)}{\partial \bph^1}) =&\ \sign(t \cdot (\gamma \cdot P_{hL} - (\alpha - \frac12)\cdot P_{hH} \cdot P_{\ell L}) - P_{hH} \cdot P_{\ell L} \cdot (\alpha - \frac12)). 
\end{align*}

Therefore, $\frac{\partial \xi_h(\bph^1, t\cdot \bph^1)}{\partial \bph^1} \ge 0$ if and only if $t \le \frac{1 - \alpha -\gamma}{P_{\ell L}\cdot (\alpha - 1/2)} - \frac{P_{hH}}{P_{\ell H}}$, and $\frac{\partial \xi_\ell(\bph^1, t\cdot \bph^1)}{\partial \bph^1} \ge 0 $ if and only if $t \ge \frac{P_{hH} \cdot P_{\ell L} \cdot (\alpha - \frac12)}{(\gamma \cdot P_{hL} - (\alpha - \frac12)\cdot P_{hH} \cdot P_{\ell L})}$. With this observation, we have the following lemma. 

\begin{lem}
    \label{lem:slope_gen}
    When $(1 - \alpha - \gamma) > \frac{P_{hH}\cdot P_{\ell L}}{P_{\ell H}} ( \alpha - \frac12)$ and $\gamma > \frac{P_{hH}\cdot P_{\ell L}}{P_{hL}} ( \alpha - \frac12)$, at least one of the three holds. (1) $\max_{\bph^1, \bpl^0, \gamma} \min(\xi_h, \xi_\ell) \le \frac{\Delta \cdot (\alpha - 1/2)}{P_{hL}}$; (2) $\min(\xi_h, \xi_\ell)$ is maximized when $\bph^1 = 1$, and (3) $\min(\xi_h, \xi_\ell)$ is maximized when $1 - \bpl^0 = 1$. 
    % $\arg\max_{\bph^1, \bpl^0, \gamma}\min(\xi_h, \xi_\ell)$ satisfies one of the following condition holds. (1) $\bph^1 = (1 - \bpl^0) = 0$. (2) $\bph^1 = 1$. (3) $1 - \bpl^0 = 1$. 
\end{lem}

% \begin{lemnb}{prop:border_cases}
%     When $(1 - \alpha - \gamma) > \frac{P_{hH}\cdot P_{\ell L}}{P_{\ell H}} ( \alpha - \frac12)$ and $\gamma > \frac{P_{hH}\cdot P_{\ell L}}{P_{hL}} ( \alpha - \frac12)$, $\arg\max_{\bph^1, \bpl^0, \gamma}\min(\xi_h, \xi_\ell)$ satisfies one of the following condition holds. (1) $\bph^1 = (1 - \bpl^0) = 0$. (2) $\bph^1 = 1$. (3) $1 - \bpl^0 = 1$. 
% \end{lemnb}

\begin{proof}
    We consider a pair of $\bph^1$ and $1 - \bpl^0$ other than (0,0). Let $t = \frac{1 - \bpl^0}{\bph^1}$.
    \begin{enumerate}
        \item If $t \ge  \frac{1 - \alpha -\gamma}{P_{\ell L}\cdot (\alpha - 1/2)} - \frac{P_{hH}}{P_{\ell H}}$, there is $\frac{\partial \xi_h(\bph^1, t\cdot \bph^1)}{\partial \bph^1} \le 0$ which implies $\xi_h(\bph^1, 1 - \bpl^0) < \xi_h(0, 0)$.
        \item If $t \le \frac{P_{hH} \cdot P_{\ell L} \cdot (\alpha - \frac12)}{(\gamma \cdot P_{hL} - (\alpha - \frac12)\cdot P_{hH} \cdot P_{\ell L})}$, there is $\frac{\partial \xi_\ell(\bph^1, t\cdot \bph^1)}{\partial \bph^1} \le 0$ which implies $\xi_\ell(\bph^1, 1 - \bpl^0) < \xi_\ell(0, 0)$.
        \item Finally, suppose $\frac{P_{hH} \cdot P_{\ell L} \cdot (\alpha - \frac12)}{(\gamma \cdot P_{hL} - (\alpha - \frac12)\cdot P_{hH} \cdot P_{\ell L})} <t <\frac{1 - \alpha -\gamma}{P_{\ell L}\cdot (\alpha - 1/2)} - \frac{P_{hH}}{P_{\ell H}}$. In this case, $\frac{\partial \xi_h(\bph^1, t\cdot \bph^1)}{\partial \bph^1} > 0$ and $\frac{\partial \xi_\ell(\bph^1, t\cdot \bph^1)}{\partial \bph^1} >0$. Therefore, if $\bph^1 < 0$ and $(1 - \bpl^0) < 1$, We consider the pair $\bph^1 + \delta\le 1$ and $(1 - \bpl^0) + t\cdot \delta \le 1$, where $\delta > 0$ is a small constant. Then we have $\xi_h(\bph^1, 1 - \bpl^0) < \xi_h(\bph^1 + \delta, 1 - \bpl^0 + t\cdot \delta)$ and $\xi_\ell(\bph^1, 1 - \bpl^0) < \xi_\ell(\bph^1 + \delta, 1 - \bpl^0 + t\cdot \delta)$. 
    \end{enumerate} 
    If the third case exists, $\min(\xi_h, \xi_{\ell})$ is maximized on either $\bph^1 = 1$ or $1 - \bpl^0 = 1$. Otherwise, $\min(\xi_h, \xi_{\ell})$ cannot exceed $\frac{\Delta \cdot (\alpha - 1/2)}{P_{hL}}$. 
\end{proof}

Combining Proposition~\ref{prop:mid_biased_gamma_gen} and Lemma~\ref{lem:slope_gen}, we have the following result. 

% \begin{prop}
%     \label{prop:border_cases}
%     At least one of the following holds. (1) $\max_{\bph^1, \bpl^0, \gamma} \min(\xi_h, \xi_\ell) \le \frac{\Delta \cdot (\alpha - 1/2)}{P_{hL}}$; (2) $\min(\xi_h, \xi_\ell)$ is maximized when $\bph^1 = 1$, and (3) $\min(\xi_h, \xi_\ell)$ is maximized when $1 - \bpl^0 = 1$. 
% \end{prop}

\begin{propnb}{prop:border_cases}
    At least one of the following holds. (1) $\max_{\bph^1, \bpl^0, \gamma} \min(\xi_h, \xi_\ell) \le \frac{\Delta \cdot (\alpha - 1/2)}{P_{hL}}$; (2) $\min(\xi_h, \xi_\ell)$ is maximized when $\bph^1 = 1$, and (3) $\min(\xi_h, \xi_\ell)$ is maximized when $1 - \bpl^0 = 1$. 
\end{propnb}

In the rest of the proof, we will deal the three cases in Proposition~\ref{prop:border_cases}. Case (1) always lead to $\frac{\Delta\cdot (\alpha - 1/2)}{P_{\ell H}}$. Next, we show that in case (3), $\xi$ will upper bounded by either $\frac{\Delta\cdot (\alpha - 1/2)}{P_{hL}}$ (Proposition~\ref{prop:seg3_h}) or $\frac12 \Delta \alpha$ (Proposition~\ref{prop:seg5_h}). Finally, we deal with the most difficult case (2) and discover the non-linear segment  (Lemma~\ref{lem:seg4_upper_fix} and Lemma~\ref{lem:seg4_upper_max}). 

\subsection{Segment 2 to 4: Case $1 - \bpl^0 = 1$. }

Now we are ready to show that, for case (3) $1- \bpl^0 = 1$, when $\frac{\infdif}{1 - \thdmaj} \ge \alpha \ge \frac{1}{1 + P_{\ell L}}$, $\frac{\Delta\cdot (\alpha - 1/2)}{P_{hL}}$ is the upper bound of $\xi$. 

\begin{propnb}{prop:seg3_h} When $(1 - \alpha - \gamma) > \frac{P_{hH}\cdot P_{\ell L}}{P_{\ell H}} ( \alpha - \frac12)$, $\gamma >\frac{P_{hH}\cdot P_{\ell L}}{P_{hL}} ( \alpha - \frac12)$, and $\frac{\infdif}{1 - \thdmaj} \ge \alpha \ge \frac{1}{1 + P_{\ell L}}$, if $1- \bpl^0 = 1$ holds, either $\xi_h$ or $\xi_{\ell}$ cannot exceed $\frac{\Delta(\alpha - \frac12)}{P_{hL}}$.
    % and $\min(\xi_h, \xi_\ell)$ is maximized at $\frac{\Delta(\alpha - \frac12)}{P_{hL}}$.
\end{propnb}

% \begin{prop}
%     \label{prop:seg3_h} When $(1 - \alpha - \gamma) > \frac{P_{hH}\cdot P_{\ell L}}{P_{\ell H}} ( \alpha - \frac12)$, $\gamma >\frac{P_{hH}\cdot P_{\ell L}}{P_{hL}} ( \alpha - \frac12)$, and $\frac{\infdif}{1 - \thdmaj} \ge \alpha \ge \frac{1}{1 + P_{\ell L}}$, if $1- \bpl^0 = 1$ holds, either $\xi_h$ or $\xi_{\ell}$ cannot exceed $\frac{\Delta(\alpha - \frac12)}{P_{hL}}$.
%     % and $\min(\xi_h, \xi_\ell)$ is maximized at $\frac{\Delta(\alpha - \frac12)}{P_{hL}}$.
% \end{prop}

\begin{proof}
    This proof discusses on different case on $\gamma$. Firstly, consider the special case where $\gamma= \frac{1 - (P_{\ell L} + P_{\ell H})\cdot\alpha}{2}$. In this case, and when $\bph^1 = (1 - \bpl^0) = 1$, $\xi_h = \xi_{\ell} = \frac{1}{2}\Delta \alpha$, which is smaller than $\frac{\Delta(\alpha - \frac12)}{P_{hL}}$ when $\alpha \ge \frac{1}{1 + P_{\ell}}$. Now we consider case when $\bph^1 < 1$. We know that when $(1 - \alpha - \gamma) > \frac{P_{hH}\cdot P_{\ell L}}{P_{\ell H}} ( \alpha - \frac12)$ and $\gamma >\frac{P_{hH}\cdot P_{\ell L}}{P_{hL}} ( \alpha - \frac12)$, $\xi_h$ is increasing on $\bph^1$ and decreasing on $(1 - \bpl^0)$, and $\xi_\ell$ is increasing on $(1 - \bpl^0)$ and decreasing on $\bph^1$. Therefore, for all $\bph^1 < 1$, $\xi_h$ decreases and is less than $ \frac{\Delta(\alpha - \frac12)}{P_{hL}}$. 

    Then we consider the case when $\gamma\neq \frac{1 - (P_{\ell L} + P_{\ell H})\cdot\alpha}{2}$. Note that $\xi_h$ decreases and $\xi_\ell$ increases when $\gamma$ increases. Therefore, when $\gamma > \frac{1 - (P_{\ell L} + P_{\ell H})\cdot\alpha}{2}$, $\xi_h$ is even smaller compared to itself when $\gamma= \frac{1 - (P_{\ell L} + P_{\ell H})\cdot\alpha}{2}$, under every same pair of $\bph^1$ and $1 - \bpl^0$. Therefore, in this case, $\xi_h \le \frac12 \Delta\alpha < \frac{\Delta(\alpha - \frac12)}{P_{hL}}$.

    Now we consider $\gamma <\frac{1 - (P_{\ell L} + P_{\ell H})\cdot\alpha}{2}$.  We will find the $\bph^1$ such that $\xi_h(\bph^1, 1) = \frac{\Delta(\alpha - \frac12)}{P_{hL}}$ and $\xi_\ell(\bph^1, 1) = \frac{\Delta(\alpha - \frac12)}{P_{hL}}$ respectively, denoted as $\bph^{1h}$ and $\bph^{1\ell}$. Then we will show that $ \bph^{1h} \ge \bph^{1\ell}$. Given that $\xi_h$ is increasing and $\xi_\ell$ is decreasing on $\bph^1$, this implies that no $\bph^1$ satisfies $\min(\xi_h(\bph^1, 1), \xi_\ell (\bph^1, 1) )> \frac{\Delta(\alpha - \frac12)}{P_{hL}}$. 

    Here we give the explicit form of $\bph^{1h}$ and $\bph^{1\ell}$. 
    \begin{align*}
        \bph^{1h} =&\  \frac{(\alpha - 1/2)\cdot (P_{\ell H}P_{\ell L} + P_{\ell H} - P_{hL})}{P_{hL}\cdot (1- \alpha - \gamma) - (\alpha - 1/2)P_{hH}P_{\ell L}}.\\
        \bph^{1\ell} =&\ \frac{\gamma \cdot P_{hL} - (\alpha - 1/2) P_{hH}P_{\ell L}}{(\alpha - 1/2)\cdot P_{hH}P_{hL}}. 
    \end{align*}

    Note that when $P_{hL}\cdot (1- \alpha - \gamma) - (\alpha - 1/2)P_{hH}P_{\ell L} \le 0$, for any $0 < \bph^1 < 1$, $\xi_h < \frac{\Delta(\alpha - \frac12)}{P_{hL}}$, and the statement holds. Now we consider the case when $P_{hL}\cdot (1- \alpha - \gamma) - (\alpha - 1/2)P_{hH}P_{\ell L} >0$. Note that 

    \begin{align*}
       &\  \frac{(\alpha - 1/2)\cdot (P_{\ell H}P_{\ell L} + P_{\ell H} - P_{hL})}{P_{hL}\cdot (1- \alpha - \gamma) - (\alpha - 1/2)P_{hH}P_{\ell L}} \ge   \frac{\gamma \cdot P_{hL} - (\alpha - 1/2) P_{hH}P_{\ell L}}{(\alpha - 1/2)\cdot P_{hH}P_{hL}}\\
       \Leftrightarrow &\ P_{hL}^2 \cdot \gamma \cdot(1 - \alpha - \gamma)- P_{hL}(\alpha -\frac12) \cdot P_{hH}P_{\ell L} \cdot (\gamma + 1 - \alpha - \gamma) + (\alpha - \frac12)^2 P_{hH}^2 P_{\ell L}^2 \\
       &\ \le (\alpha - \frac12)^2 P_{hH}P_{\ell L}P_{hL}P_{\ell H} + (\alpha - \frac12)^2 \cdot P_{hH} P_{hL} \cdot (P_{\ell H} - P_{hL})\\
       \Leftrightarrow &\ P_{hL}^2 \cdot \gamma \cdot(1 - \alpha - \gamma)- P_{hL}(\alpha -\frac12) \cdot P_{hH}P_{\ell L} \cdot ( 1 - \alpha) + (\alpha - \frac12)^2 P_{hH}^2 P_{\ell L}^2 \\
       &\ \le (\alpha - \frac12)^2 P_{hH}P_{\ell L}P_{hL}P_{\ell H} + (\alpha - \frac12)^2 \cdot P_{hH} P_{hL} \cdot (P_{\ell H} - P_{hL}). 
    \end{align*}
    Note that the first term $ P_{hL}^2 \cdot \gamma \cdot(1 - \alpha - \gamma)$ decreases on $\gamma $ when $\gamma \le \frac12(1 - \alpha)$. Given the constraint that $\gamma  < \frac{1 - (P_{\ell L} + P_{\ell H})\cdot\alpha}{2} \le \frac12(1 - \alpha)$, it suffices to show that the inequality holds when $\gamma  = \frac{1 - (P_{\ell L} + P_{\ell H})\cdot\alpha}{2}$. Recall that when $\bph^1 = 1$, $\xi_h = \xi_\ell = \frac12 \Delta \alpha < \frac{\Delta(\alpha - \frac12)}{P_{hL}}$, $\xi_h$ increases on $\bph^1$, and $\xi_\ell$ decreases on $\bph^1$. Therefore, there must be $\bph^{1h} > \bph^{1\ell}$, which finishes the proof. 
\end{proof}

\begin{propnb}{prop:seg5_h} 
When $(1 - \alpha - \gamma) > \frac{P_{hH}\cdot P_{\ell L}}{P_{\ell H}} ( \alpha - \frac12)$, $\gamma >\frac{P_{hH}\cdot P_{\ell L}}{P_{hL}} ( \alpha - \frac12)$, and $ \frac{1}{1 + P_{\ell L}}\ge \alpha \ge \frac{1}2$, if $1- \bpl^0 = 1$ holds, either $\xi_h$ or $\xi_{\ell}$ cannot exceed $\frac12 \Delta \alpha$, and $\min(\xi_h, \xi_\ell)$ is maximized at $\frac12 \Delta \alpha$. 
\end{propnb}

% \begin{prop}
%     \label{prop:seg5_h} When $(1 - \alpha - \gamma) > \frac{P_{hH}\cdot P_{\ell L}}{P_{\ell H}} ( \alpha - \frac12)$, $\gamma >\frac{P_{hH}\cdot P_{\ell L}}{P_{hL}} ( \alpha - \frac12)$, and $ \frac{1}{1 + P_{\ell L}}\ge \alpha \ge \frac{1}2$, if $1- \bpl^0 = 1$ holds, either $\xi_h$ or $\xi_{\ell}$ cannot exceed $\frac12 \Delta \alpha$, and $\min(\xi_h, \xi_\ell)$ is maximized at $\frac12 \Delta \alpha$. 
% \end{prop}

\begin{proof}
    This proof resembles the proof of Proposition~\ref{prop:seg3_h}. First, when $\gamma= \frac{1 - (P_{\ell L} + P_{\ell H})\cdot\alpha}{2}$, and $\bph^1 = (1 - \bpl^0) = 1$, $\xi_h = \xi_{\ell} = \frac{1}{2}\Delta \alpha$. We know that when $(1 - \alpha - \gamma) > \frac{P_{hH}\cdot P_{\ell L}}{P_{\ell H}} ( \alpha - \frac12)$ and $\gamma >\frac{P_{hH}\cdot P_{\ell L}}{P_{hL}} ( \alpha - \frac12)$, $\xi_h$ is increasing on $\bph^1$ and decreasing on $(1 - \bpl^0)$, and $\xi_\ell$ is increasing on $(1 - \bpl^0)$ and decreasing on $\bph^1$. Therefore, for all $(1 - \bpl^0) = 1$ and $\bph^1 < 1$, $\xi_h < \frac{1}{2}\Delta \alpha$. When $\gamma >\frac{1 - (P_{\ell L} + P_{\ell H})\cdot\alpha}{2}$, $\xi_h$ decreases and is smaller than $\frac12 \Delta \alpha$. 

    Now we consider $\gamma <\frac{1 - (P_{\ell L} + P_{\ell H})\cdot\alpha}{2}$.  We will find the $\bph^1$ such that $\xi_h(\bph^1, 1) = \frac12 \Delta \alpha$ and $\xi_\ell(\bph^1, 1) = \frac12 \Delta \alpha$ respectively, denoted as $\bph^{1h}$ and $\bph^{1\ell}$. Then we will show that $ \bph^{1h} \ge \bph^{1\ell}$. Given that $\xi_h$ is increasing and $\xi_\ell$ is decreasing on $\bph^1$, this implies that no $\bph^1$ satisfies $\min(\xi_h(\bph^1, 1), \xi_\ell (\bph^1, 1)) > \frac12 \Delta \alpha$. 

    Here we give the explicit form of $\bph^{1h}$ and $\bph^{1\ell}$. 

    \begin{align*}
        \bph^{1h} = &\ \frac{\alpha\cdot (P_{\ell  H}P_{\ell L} + P_{\ell H} - 2) +1 }{2(1 - \alpha - \gamma) - \alpha \cdot P_{hH}P_{\ell L}}\\
        \bph^{1\ell} =&\  \frac{2\gamma - 1 + \alpha\cdot (2 - P_{hL} -P_{hH} P_{\ell L})}{\alpha \cdot P_{hH}P_{hL}} = \frac{2\gamma - 1 + \alpha\cdot (1 + P_{\ell L}P_{\ell H})}{\alpha \cdot P_{hH}P_{hL}}
    \end{align*}
     Note that given $\alpha \le \frac{1}{1+P_{\ell L}} < \frac{1}{1+P_{\ell L}P_{hH}}$ and $\gamma <\frac{1 - (P_{\ell L} + P_{\ell H})\cdot\alpha}{2} \le \frac{1-\alpha}{2}$, $2(1 - \alpha - \gamma) - \alpha \cdot P_{hH}P_{\ell L} > 2 - (1+P_{hH}P_{\ell L})\cdot \alpha > 0$. Now we show that $\bph^{1h} \ge \bph^{1\ell }$ holds. 
     
     \begin{align*}
         &\ \frac{\alpha\cdot (P_{\ell  H}P_{\ell L} + P_{\ell H} - 2) +1 }{2(1 - \alpha - \gamma) - \alpha \cdot P_{hH}P_{\ell L}} \ge \frac{2\gamma - 1 + \alpha\cdot (1 + P_{\ell L}P_{\ell H})}{\alpha \cdot P_{hH}P_{hL}}\\
         \Leftrightarrow&\ 4\gamma^2 - 2\gamma\cdot  (3(1 -\alpha) - \alpha \cdot P_{\ell L}) +2 (1 - \alpha)^2  - \alpha \cdot (1 - \alpha) \cdot P_{\ell L}\\
         &\ + \alpha\cdot (2\alpha - 1) \cdot (P_{\ell L} P_{\ell H} - P_{hH}P_{hL}) + \alpha^2 (P_{hH} P_{\ell H} - P_{\ell L} P_{hL}) \ge 0\\
         \Leftrightarrow&\ 4((1 - \alpha - \gamma) -\frac{1 - \alpha}{2})\cdot ((1 - \alpha - \gamma) -\frac{\alpha \cdot P_{\ell L}}{2})\\
         &\ + \alpha\cdot (2\alpha - 1) \cdot (P_{\ell L} P_{\ell H} - P_{hH}P_{hL}) + \alpha^2 (P_{hH} P_{\ell H} - P_{\ell L} P_{hL}) \ge 0. 
     \end{align*}
     We show that every term in the LHS is non-negative. 
     \begin{itemize}
         \item Since we assume $\gamma <\frac{1 - (P_{\ell L} + P_{\ell H})\cdot\alpha}{2}$, there is $1 - \alpha - \gamma  > \frac{1 - (P_{h L} + P_{h H})\cdot\alpha}{2} \ge  \frac{1-\alpha}{2}$. Moreover, when $\alpha < \frac{1}{1+P_{\ell L}}$, $\alpha\cdot P_{\ell L} < \frac{P_{\ell L}}{1+P_{\ell L}} < 1 - \alpha$. Therefore, the first term is non-negative. 
         \item For the second term, we have $\alpha > \frac12$, $P_{\ell L} \ge P_{hH}$, and $P_{\ell H} \ge P_{hL}$. Therefore, the second term is non-negative. 
         \item For the last term, given that $P_{hL} < P_{hH} \le P_{\ell L}$ and $P_{hH}+P_{\ell H} = P_{hL}+P_{\ell L} = 1$, there must be $P_{hH}P_{\ell H} \ge P_{\ell L}P_{hL}$. Therefore, the last term is non-negative. 
     \end{itemize}

     Therefore, we show that $\bph^{1h} \ge \bph^{1\ell}$, which finishes the proof. 
\end{proof}

\subsection{Segment 2 to 4: Case $\bph^1 = 1$. }

For now we start to deal with case (2) in Proposition~\ref{prop:border_cases} where $\bph^1 = 1$. This case contains a non-linear upper bound which cannot be easily characterized. The high-level idea is as follows. Firstly, note that $\xi_h$ decreases while $\xi_{\ell}$ increases on $\gamma$. Therefore, for a fixed $1 - \bpl^0$, there exists a unique $\gamma$ such that $\xi_h = \xi_\ell$. Moreover, if this gamma is valid, it maximized $\min(\xi_h, \xi_{\ell})$ for this fixed $\bpl^0$. Therefore, we first calculate this $\gamma$ and the corresponding $\xi_h$ and $\xi_\ell$ for an arbitrary fixed $1 - \bpl^0$ (Lemma~\ref{lem:seg4_upper_fix}). Then we maximize $\xi_h$ (also equals to $\xi_\ell$) on $1 - \bpl^0$ under the optimal $\gamma$ (Lemma~\ref{lem:seg4_upper_max}). This brings us an upper bound for case (2).  

\begin{lemnb}{lem:seg4_upper_fix} When $(1 - \alpha - \gamma) > \frac{P_{hH}\cdot P_{\ell L}}{P_{\ell H}} ( \alpha - \frac12)$, $\gamma >\frac{P_{hH}\cdot P_{\ell L}}{P_{hL}} ( \alpha - \frac12)$, and $ \frac{1}{2 P_{\ell L}}\ge \alpha \ge \frac{1}2$, for $\bph^1 = 1$ and any fixed $1 - \bpl^0$, let 
    \begin{equation*}
        \hat{\xi} = \frac{\Delta \cdot (\alpha - 1/2 + 1/2 \cdot (1 - \bpl^0))}{(P_{\ell L}\cdot (1 - \bpl^0) + P_{hL})\cdot (P_{\ell H}\cdot (1 - \bpl^0) + P_{hH} +1)}. 
    \end{equation*}
    Then, $\min(\xi_h, \xi_{\ell}) \le \hat{\xi}$ holds for all feasible $\gamma$. 
\end{lemnb}

% \begin{lem}
%     \label{lem:seg4_upper_fix} When $(1 - \alpha - \gamma) > \frac{P_{hH}\cdot P_{\ell L}}{P_{\ell H}} ( \alpha - \frac12)$, $\gamma >\frac{P_{hH}\cdot P_{\ell L}}{P_{hL}} ( \alpha - \frac12)$, and $ \frac{1}{2 P_{\ell L}}\ge \alpha \ge \frac{1}2$, for $\bph^1 = 1$ and any fixed $1 - \bpl^0$, let 
%     \begin{equation*}
%         \hat{\xi} = \frac{\Delta \cdot (\alpha - 1/2 + 1/2 \cdot (1 - \bpl^0))}{(P_{\ell L}\cdot (1 - \bpl^0) + P_{hL})\cdot (P_{\ell H}\cdot (1 - \bpl^0) + P_{hH} +1)}. 
%     \end{equation*}
%     Then, $\min(\xi_h, \xi_{\ell}) \le \hat{\xi}$ holds for all feasible $\gamma$. 
% \end{lem}

\begin{proof}
    First, we claim that $\hat{\xi}$ is reached by fixing $\bph^1 = 1$ and $1 - \bpl^0$ and modifying $\gamma$ such that $\xi_h = \xi_{\ell}$. By solving the equation $\xi_h = \xi_{\ell}$ for $\gamma$, we get the solution
    \begin{equation*}
        \hat{\gamma} = \frac{\frac12 (P_{\ell L}\cdot (1 - \bpl^0) + P_{hL} +1) -\alpha \cdot (P_{hH}P_{\ell L} + P_{\ell H}P_{\ell L} \cdot (1-\bpl^0) +P_{\ell H})}{(P_{\ell L}\cdot (1 - \bpl^0) + P_{hL})\cdot (P_{\ell H}\cdot (1 - \bpl^0) + P_{hH} +1)}. 
    \end{equation*}
    Correspondingly, when $\gamma = \hat{\gamma}$, $\xi_h = \xi_\ell = \hat{\xi}$. 

    Now we show that $\hat{\xi}$ is a upper bound of $\min(\xi_h, \xi_{\ell})$. Recall that $\xi_h$ is decreasing on $\xi$ and $\xi_{\ell}$ is increasing on $\xi$. Therefore, for all $\gamma \ge \hat{\gamma}$, $\xi_h \le \hat{\xi}$; and for all $\gamma \le \hat{\gamma}$, $\xi_\ell \le \hat{\xi}$. 
\end{proof}

\begin{lem}
    \label{lem:seg4_upper_max} Let \begin{equation*}
        \hat{\xi} = \frac{\Delta \cdot (\alpha - 1/2 + 1/2 \cdot (1 - \bpl^0))}{(P_{\ell L}\cdot (1 - \bpl^0) + P_{hL})\cdot (P_{\ell H}\cdot (1 - \bpl^0) + P_{hH} +1)}. 
    \end{equation*}
    For $(1 - \bpl^0) \in [0, 1]$, 
    \begin{itemize}
        \item when $\alpha \ge \frac{1 + P_{hL} + P_{hH}P_{\ell L}}{2(P_{hH}P_{\ell L} + P_{hL} P_{\ell H} +P_{\ell L})}$, $\max_{(1 - \bpl^0)} \hat{\xi} = \frac{\Delta(\alpha - 1/2)}{P_{hL}\cdot (P_{hH} + 1)}$ with $\arg \max_{(1 - \bpl^0)} \hat{\xi}  = 0$;
        \item when $\frac{1 + P_{hL} + P_{hH}P_{\ell L}}{2(P_{hH}P_{\ell L} + P_{hL} P_{\ell H} +P_{\ell L})} > \alpha > \frac{1}{1+P_{\ell L} + (P_{\ell L} - P_{hH})}$, $\max_{(1 - \bpl^0)} \hat{\xi} = \xinl$, and the maximum $\arg \max_{(1 - \bpl^0)} \hat{\xi} \in (0, 1)$;
        \item when $\alpha \le  \frac{1}{1+P_{\ell L} + (P_{\ell L} - P_{hH})}$, $\max_{(1 - \bpl^0)} \hat{\xi} = \frac12 \Delta \alpha$ with $\arg \max_{(1 - \bpl^0)} \hat{\xi} = 1$. 
    \end{itemize}
\end{lem}

\begin{proof}
    Consider the derivative of $\hat{\xi}$ on $(1 - \bpl^0)$. 

    \begin{align*}
        \sign(\frac{\partial \hat{\xi}}{\partial (1 - \bpl^0)}) =&\  \sign( (P_{\ell L}\cdot (1 - \bpl^0) + P_{hL})\cdot (P_{\ell H}\cdot (1 - \bpl^0) + P_{hH} +1) \\
        &\ - (\alpha - 1/2 + 1/2 \cdot (1 - \bpl^0)) \\
        &\ \ \cdot (P_{\ell L}\cdot (P_{\ell H}\cdot (1 - \bpl^0) + P_{hH} +1) + P_{\ell H} \cdot  (P_{\ell L}\cdot (1 - \bpl^0) + P_{hL}))\\
        =&\ \sign(-\frac12 P_{\ell L} P_{\ell H} \cdot (1- \bpl^0)^2 -(2\alpha - 1) P_{\ell L}P_{\ell H}\cdot (1-\bpl^0) \\
        &\ - (\alpha -\frac12) \cdot (P_{hH}P_{\ell L} + P_{hL} P_{\ell H} +P_{\ell L}) + \frac12(P_{hL}P_{hH} + P_{hL})). 
    \end{align*}

    Denote the right-hand side 
    \begin{align*}
        \func(1 - \bpl^0) =&\ -\frac12 P_{\ell L} P_{\ell H} \cdot (1- \bpl^0)^2 -(2\alpha - 1) P_{\ell L}P_{\ell H}\cdot (1-\bpl^0) \\
        &\ - (\alpha -\frac12) \cdot (P_{hH}P_{\ell L} + P_{hL} P_{\ell H} +P_{\ell L}) + \frac12(P_{hL}P_{hH} + P_{hL})
    \end{align*}

    We have the following observations. Firstly, $\func(1 - \bpl^0)$ is decreasing on $(1- \bpl^0)$ for $(1 - \bpl^0) \in [0, 1]$. Secondly, when $1 - \bpl^0 = 0$, 
    \begin{align*}
       \func(1 - \bpl^0) = &\  - (\alpha -\frac12) \cdot (P_{hH}P_{\ell L} + P_{hL} P_{\ell H} +P_{\ell L}) + \frac12(P_{hL}P_{hH} + P_{hL})\\
        =&\ \frac12\cdot (1 + P_{hL} + P_{hH}P_{\ell L}) - \alpha \cdot (P_{hH}P_{\ell L} + P_{hL} P_{\ell H} +P_{\ell L})
    \end{align*}
    Thirdly, when $1 - \bpl^0 = 1$
    \begin{align*}
        \func(1 - \bpl^0)= &\  -\frac12 P_{\ell L} P_{\ell H} -(2\alpha - 1) P_{\ell L}P_{\ell H}\\
        &\ - (\alpha -\frac12) \cdot (P_{hH}P_{\ell L} + P_{hL} P_{\ell H} +P_{\ell L}) + \frac12(P_{hL}P_{hH} + P_{hL})\\
        =&\ 1 - \alpha \cdot (1 + P_{\ell L} +(P_{\ell L} - P_{hH})). 
    \end{align*}

    Therefore, when $\alpha \ge \frac{1 + P_{hL} + P_{hH}P_{\ell L}}{2(P_{hH}P_{\ell L} + P_{hL} P_{\ell H} +P_{\ell L})}$, $\sign(\frac{\partial \hat{\xi}}{\partial (1 - \bpl^0)}) \le 0$ for all $(1 - \bpl^0) \in [0, 1]$, and $\hat{\xi}$ is maximized at $(1 - \bpl^0) = 0$ with $\max_{(1 - \bpl^0)} \hat{\xi} = \frac{\Delta(\alpha - 1/2)}{P_{hL}\cdot (P_{hH} + 1)}$. 
    
    When $\alpha \le  \frac{1}{1+P_{\ell L} + (P_{\ell L} - P_{hH})}$, $\sign(\frac{\partial \hat{\xi}}{\partial (1 - \bpl^0)}) \ge 0$ for all $(1 - \bpl^0) \in [0, 1]$, and $\hat{\xi}$ is maximized at $(1 - \bpl^0) = 1$ with $\max_{(1 - \bpl^0)} \hat{\xi} = \frac12 \Delta \alpha$. 
    
    Finally, when $\frac{1 + P_{hL} + P_{hH}P_{\ell L}}{2(P_{hH}P_{\ell L} + P_{hL} P_{\ell H} +P_{\ell L})} > \frac{1}{1+P_{\ell L} + (P_{\ell L} - P_{hH})}$, $\frac{\partial \hat{\xi}}{\partial (1 - \bpl^0)}$ is positive when $(1 - \bpl^0) = 0$ and negative when $(1 - \bpl^0) = 1$. Therefore, $\hat{\xi}$ is maximized when $(1 - \bpl^0)$ satisfies $\func(1 - \bpl^0) = 0$, which is 
    \begin{equation*}
        (1 - \bpl^0) = (1 - 2\alpha) + \frac{\sqrt{2\cdot (1 - \alpha \cdot P_{\ell H})\cdot (1 - 2\alpha \cdot P_{\ell L}) \cdot P_{\ell L}P_{\ell H}}}{P_{\ell L} \cdot P_{\ell H}}. 
    \end{equation*}

    Note that $\func(1 - \bpl^0) = 0$ is a quadratic equation. By the monotonicity of $\func$ and the intermediate value theorem, a real-number solution is guaranteed. On the other hand, the other solution of the equation is out of $[0, 1]$. 
    Therefore, this is the only solution in $[0, 1]$ and is guaranteed to be in $[0, 1]$. 
    
    The corresponding $\hat{\xi}$ is exactly $\xinl$. 
\end{proof}

\subsection{Segment 2 to 4: Wrap up the Upper Bounds.}

So far, we have shown that the upper and the lower bound of $\xi$ only derive from one of the four cases. (1) When $\gamma$ is biased (either type-0 agents or type-1 agents dominates the minority group). (2) When $1 - \bpl^0 = 1$ (all type-0 agents always vote for $\bR$). (3) When $\bph^1 = 1$ (all type-1 agents always vote for $\bA$). (4) All other cases, upper-bounded with $\frac{\Delta(\alpha-1/2)}{P_{hL}}$ by Proposition~\ref{prop:border_cases}. Each case has its own upper bound, and cases (1), (2), and (3) also have lower bounds. 

Therefore, the last step is to integrate all the cases and reach a global upper and lower bound for each Segment. For the upper bound, we take the maximum of the upper bound for each case. In each case, there does not exist a strategy profile satisfying all three constraints with $\xi$ equal to their corresponding upper bound. Therefore, overall, no strategy profiles can satisfy all three constraints with $\xi$ equal to the maximum of all upper bounds. For the lower bound, we show that the upper bound in each segment binds a tight lower bound. As a consequence, the bounds consist of the threshold curve $\xi^*(\alpha)$, which finishes the proof.

% The lower bound under each cse 

% as the goal is to find a strategy profile that satisfies all the constraints, it is natural to take the maximum among all the cases. For the upper bound, we also take the maximum of each case. As each case

% For the upper bound, we will see that every lower bound we take binds with a tight upper bound. As a consequence, the bounds consist of the threshold curve $\xi^*(\alpha)$, which finishes the proof. 

Now we start to integrate the upper bound of each case. 
\begin{enumerate}
    \item When $\gamma$ is biased, the upper bound is $\frac{\Delta(\alpha - 1/2)}{P_{hL}}$ by Proposition~\ref{prop:biasedgamma_gen} along Segment 2, 3, and 4.
    \item When $1 - \bpl^0 = 1$, the upper bound is given by the maximum of $\frac{\Delta(\alpha - 1/2)}{P_{hL}}$ and $\frac12\Delta\alpha$ by Proposition~\ref{prop:seg3_h} and~\ref{prop:seg5_h}. 
    \item When $\bph^1 = 0$, the upper  bound is no more than $\frac{\Delta(\alpha - 1/2)}{P_{hL}}$ in Segment 2, $\xinl$ in Segment 3, and $\frac12\Delta\alpha$ in Segment 4 by Proposition~\ref{lem:seg4_upper_fix} and~\ref{lem:seg4_upper_max}. 
    \item In all other cases, the upper bound is at most $\frac{\Delta(\alpha - 1/2)}{P_{hL}}$ by Proposition~\ref{prop:border_cases}.
\end{enumerate}
The order of value of these upper bounds varies in different segments, of which the full proof is in Appendix~\ref{apx:lemmas}. In Segment 2, $\frac{\Delta(\alpha - 1/2)}{P_{hL}}$ is the largest. In Segment 3, $\xinl$ is the largest. Finally, in Segment 4, $\frac12 \Delta \alpha \ge \frac{\Delta(\alpha - 1/2)}{P_{hL}}$. Therefore, taking the maximum on each segment, the upper bound for $\xi$ is $\frac{\Delta(\alpha - 1/2)}{P_{hL}}$ in Segment 2, $\xinl$ in Segment 3, and $\frac12 \Delta \alpha$ in Segment 4 (if exists). This is exactly the upper bound given in $\xi^*(\alpha)$.

Here we give the formal proof for the upper bounds. 

\begin{prop}[Segment 2]
    \label{prop:seg3upper}
    When $\frac{1}{2P_{\ell L}} > \alpha \ge \anl$, no  (sequence of) strategy profiles can satisfies all three conditions for $\xi = \frac{\Delta\cdot (\alpha - 1/2)}{P_{hL}}$. 
\end{prop}

\begin{proof}
    By Lemma~\ref{lem:xiupper}, it suffices to show that $\min(\xi_h, \xi_{\ell}) \le  \frac{\Delta\cdot (\alpha - 1/2)}{P_{hL}}$ for all possible $\bph^1$, $1 - \bpl^0$, and $\gamma$. By Proposition~\ref{prop:mid_biased_gamma_gen}, when $(1 - \alpha - \gamma) \le \frac{P_{hH}\cdot P_{\ell L}}{P_{\ell H}} ( \alpha - \frac12)$ or $\gamma \le \frac{P_{hH}\cdot P_{\ell L}}{P_{hL}} ( \alpha - \frac12)$ holds, $\min(\xi_h, \xi_{\ell}) \le  \frac{\Delta\cdot (\alpha - 1/2)}{P_{hL}}$. For the rest of the proof, we assume $(1 - \alpha - \gamma) > \frac{P_{hH}\cdot P_{\ell L}}{P_{\ell H}} ( \alpha - \frac12)$ and $\gamma > \frac{P_{hH}\cdot P_{\ell L}}{P_{hL}} ( \alpha - \frac12)$. Then it suffices to show that in all three cases in Proposition~\ref{prop:border_cases}, $\min(\xi_h, \xi_{\ell}) \le   \frac{\Delta\cdot (\alpha - 1/2)}{P_{hL}}$.

    For case (1) where $\bph^1 = 1 - \bpl^0 = 0$, $\xi_h = \frac{\Delta\cdot (\alpha - 1/2)}{P_{\ell H}} < \frac{\Delta\cdot (\alpha - 1/2)}{P_{hL}}$ and $\xi_\ell = \frac{\Delta\cdot (\alpha - 1/2)}{P_{hL}}$. The statement holds.

    For case (2) where $\bph^1 = 1$, we apply Lemma~\ref{lem:seg4_upper_fix} and Lemma~\ref{lem:seg4_upper_max}.
    Firstly, Lemma~\ref{lem:seg4_upper_fix} and Lemma~\ref{lem:seg4_upper_max} characterize that, for $\alpha \ge \frac{1 + P_{hL} + P_{hH}P_{\ell L}}{2(P_{hH}P_{\ell L} + P_{hL} P_{\ell H} +P_{\ell L})}$, $\min(\xi_h, \xi_{\ell}) \le \frac{\Delta(\alpha - 1/2)}{P_{hL}\cdot (P_{hH} + 1)} < \frac{\Delta\cdot (\alpha - 1/2)}{P_{hL}}$; and when $\frac{1 + P_{hL} + P_{hH}P_{\ell L}}{2(P_{hH}P_{\ell L} + P_{hL} P_{\ell H} +P_{\ell L})} > \alpha > \frac{1}{1+P_{\ell L} + (P_{\ell L} - P_{hH})}$, $\min(\xi_h, \xi_{\ell}) \le\xinl$. Secondly, Lemma~\ref{lem:a4threshold} implies that when $\frac{1}{2P_{\ell L}} \ge \alpha > \anl$, $\xinl \le \frac{\Delta\cdot (\alpha - 1/2)}{P_{hL}}$. Finally, Lemma~\ref{lem:seg4_upper_max} and Lemma~\ref{lem:a4threshold} together shows that that $\anl < \frac{1 + P_{hL} + P_{hH}P_{\ell L}}{2(P_{hH}P_{\ell L} + P_{hL} P_{\ell H} +P_{\ell L})}$. This is because, when $\alpha = \anl$, $\gamma_4 = \frac{\Delta\cdot (\alpha - 1/2)}{P_{hL}}$, while when $\alpha \ge \frac{1 + P_{hL} + P_{hH}P_{\ell L}}{2(P_{hH}P_{\ell L} + P_{hL} P_{\ell H} +P_{\ell L})}$, $\xinl \le \frac{\Delta(\alpha - 1/2)}{P_{hL}\cdot (P_{hH} + 1)} < \frac{\Delta\cdot (\alpha - 1/2)}{P_{hL}}$. Combining these three facts, we know that, 
    \begin{enumerate}
        \item When $\frac{1}{2P_{\ell L}} \ge \alpha \ge \min (\frac{1}{2P_{\ell L}}, \frac{1 + P_{hL} + P_{hH}P_{\ell L}}{2(P_{hH}P_{\ell L} + P_{hL} P_{\ell H} +P_{\ell L})})$, $\min(\xi_h, \xi_{\ell}) \le \frac{\Delta(\alpha - 1/2)}{P_{hL}\cdot (P_{hH} + 1)} < \frac{\Delta\cdot (\alpha - 1/2)}{P_{hL}}$.
        \item When $\min (\frac{1}{2P_{\ell L}}, \frac{1 + P_{hL} + P_{hH}P_{\ell L}}{2(P_{hH}P_{\ell L} + P_{hL} P_{\ell H} +P_{\ell L})}) > \alpha \ge \anl$, $\min(\xi_h, \xi_{\ell})\le \xinl \le \frac{\Delta\cdot (\alpha - 1/2)}{P_{hL}}$.
    \end{enumerate}
    Together this shows that $\min(\xi_h, \xi_{\ell}) \le \frac{\Delta\cdot (\alpha - 1/2)}{P_{hL}}$ when $\frac{1}{2P_{\ell L}} > \alpha \ge \anl$.

    For case (3) where $1 - \bpl^0 = 1$, Proposition~\ref{prop:seg3_h} guarantees that $\min(\xi_h, \xi_{\ell}) \le \frac{\Delta\cdot (\alpha - 1/2)}{P_{hL}}$. Lemma~\ref{lem:a4lower} guarantees that the range of $\alpha$ in Proposition~\ref{prop:seg3_h} covers $\frac{1}{2P_{\ell L}} > \alpha \ge \anl$. 

    Therefore, in all three cases, $\min(\xi_h, \xi_{\ell}) \le   \frac{\Delta\cdot (\alpha - 1/2)}{P_{hL}}$. This finishes the proof. 
\end{proof}

\begin{prop}[Segment 4]
    \label{prop:seg5upper}
    When $ \frac{1}{1+P_{\ell L} + (P_{\ell L} - P_{hH})} \ge \frac12$ and $ \frac{1}{1+P_{\ell L} + (P_{\ell L} - P_{hH})} \ge \alpha \ge \frac12$, no (sequence of) strategy profiles can satisfies all three conditions for $\xi = \frac12 \Delta \alpha$. 
\end{prop}

\begin{proof}
    Note that when $\alpha \le \frac{1}{1 + P_{\ell L}}$, $\frac12 \Delta \alpha \ge \frac{\Delta\cdot (\alpha - 1/2)}{P_{hL}}$. By Lemma~\ref{lem:xiupper}, Proposition~\ref{prop:mid_biased_gamma_gen},  and Proposition~\ref{prop:border_cases}, it suffices to show that in all three cases in Proposition~\ref{prop:border_cases}, $\min(\xi_h, \xi_{\ell}) \le   \frac12 \Delta \alpha$. 

    For case (1) where $\bph^1 = 1 - \bpl^0 = 0$, $\xi_h = \frac{\Delta\cdot (\alpha - 1/2)}{P_{\ell H}} < \frac{\Delta\cdot (\alpha - 1/2)}{P_{hL}}$ and $\xi_\ell = \frac{\Delta\cdot (\alpha - 1/2)}{P_{hL}}$. When $\alpha \le \frac1{1+P_{\ell }}$, $\frac{\Delta\cdot (\alpha - 1/2)}{P_{hL}} \le \frac12 \Delta\alpha$. 

    For case (2) where $\bph^1 = 1$, by Lemma~\ref{lem:seg4_upper_fix} and Lemma~\ref{lem:seg4_upper_max}, for any $\bph^1 = 1$ and $\bpl^0 \in [0, 1]$, $\min(\xi_h, \xi_{\ell}) \le \frac12 \Delta\alpha$. 

    For case (3) where $1 - \bpl^0 = 1$, Proposition~\ref{prop:seg5_h} guarantees that $\min(\xi_h, \xi_{\ell}) \le \frac12 \Delta\alpha$. 

    Therefore, in all three cases, $\min(\xi_h, \xi_{\ell}) \le   \frac12 \Delta\alpha$. This finishes the proof.
\end{proof}

\begin{prop}[Segment 3]
    \label{prop:seg4upper}
    When $ \anl \ge \alpha \ge \max(\frac{1}{1+P_{\ell L} + (P_{\ell L} - P_{hH})}, \frac12)$, no (sequence of) strategy profiles can satisfies all three conditions for $\xi = \xinl$. 
\end{prop}

\begin{proof}
    Note that when $\alpha < \anl$, $\xinl > \frac{\Delta\cdot (\alpha - 1/2)}{P_{hL}}$. By Lemma~\ref{lem:xiupper}, Proposition~\ref{prop:mid_biased_gamma_gen}, and Proposition~\ref{prop:border_cases}, it suffices to show that in all three cases in Proposition~\ref{prop:border_cases}, $\min(\xi_h, \xi_{\ell}) \le   \xinl$. 

    For case (1) where $\bph^1 = 1 - \bpl^0 = 0$, $\xi_h = \frac{\Delta\cdot (\alpha - 1/2)}{P_{\ell H}} < \frac{\Delta\cdot (\alpha - 1/2)}{P_{hL}}$ and $\xi_\ell = \frac{\Delta\cdot (\alpha - 1/2)}{P_{hL}}$. By Lemma~\ref{lem:a4threshold}, for $\anl \ge \alpha> \frac12$, $\frac{\Delta\cdot (\alpha - 1/2)}{P_{hL}} \le \xinl$. 

    For case (2) where $\bph^1 = 1$, by Lemma~\ref{lem:seg4_upper_fix} and Lemma~\ref{lem:seg4_upper_max}, for any $\bph^1 = 1$ and $\bpl^0 \in [0, 1]$, $\min(\xi_h, \xi_{\ell}) \le \xinl$. 

    For case (3) where $1 - \bpl^0 = 1$, Proposition~\ref{prop:seg3_h} guarantees that for $\frac{1}{2P_{\ell L}} \ge \alpha > \frac{1}{1+P_{\ell L}}$, $\min(\xi_h, \xi_{\ell}) \le \frac{\Delta\cdot (\alpha - 1/2)}{P_{hL}}$, and Proposition~\ref{prop:seg5_h} guarantees that $\min(\xi_h, \xi_{\ell}) \le \frac12 \Delta\alpha$. By Lemma~\ref{lem:a4threshold}, for $\anl \ge \alpha> \frac12$, $\frac{\Delta\cdot (\alpha - 1/2)}{P_{hL}} \le \xinl$. Lemma~\ref{lem:seg4_upper_max} implies that when $\alpha \ge \frac{1}{1+P_{\ell L} + (P_{\ell L} - P_{hH})}$, $\xinl \ge \frac12 \Delta\alpha$. Therefore, for all $ \anl \ge \alpha \ge \max(\frac{1}{1+P_{\ell L} + (P_{\ell L} - P_{hH})}, \frac12)$, $\min(\xi_h, \xi_{\ell}) \le \xinl$. 

    Therefore, in all three cases, $\min(\xi_h, \xi_{\ell}) \le   \xinl$. This finishes the proof.
\end{proof}

Before we wrap up the lower bounds, it remains to prove the lower bounds for Segment 3 and 4, which consists of the following subsection

\subsection{Lower bound for Segment 3 and Segment 4.}
We will start from the simpler Segment 4. 
\begin{prop}
    \label{prop:seg5lower}
    When $\frac{1}{1 + P_{\ell L}} \ge \alpha \ge \frac12$, the following strategy $\stgp$ profile satisfies the constraints in Theorem~\ref{thm:thresholdk} with $\xi^* = \frac12 \cdot \Delta \cdot \alpha$: all the majority agents report informatively, a fraction $\gamma^* = \frac{1 - (P_{\ell L} + P_{\ell H})\cdot\alpha}{2}$ of minority agents vote for $\bR$, and the other minority agents ($1 - \alpha -\gamma^* = \frac{1 - (P_{hL} + P_{hH})\cdot\alpha}{2}$ ) always vote for $\bR$. 
\end{prop}

\begin{proof}
    Given that $\gamma^* \le (1 - \alpha)/2$, it suffices to show that $\gamma^* \ge \xi^*$. Note that this equivalents to $\alpha \le \frac{1}{P_{\ell L} +P_{\ell H} + \Delta} = \frac{1}{P_{hH} +P_{\ell L}}$. Given that $\frac{1}{P_{hH} +P_{\ell L}} \ge \frac{1}{2P_{\ell L}}$, this holds for all $\frac{1}{1 + P_{\ell L}} \ge \alpha \ge \frac12$. 

    Now we show that all three constraints are satisfies. 
    Firstly, by Lemma 1, no agents play weakly dominated strategies. For the second constraint of fidelity, we calculated the expected vote share for $\bA$ under two world states respectively. 
    \begin{align*}
        \thd_H =&\ \alpha \cdot P_{hH} + \frac{1 - (P_{hL} + P_{hH})\cdot\alpha}{2} = \frac{1}{2} + \frac{1}{2} \cdot \Delta \cdot \alpha \\
        \thd_L =&\ \alpha\cdot P_{hL} + \frac{1 - (P_{hL} + P_{hH})\cdot\alpha}{2} = \frac{1}{2} - \frac{1}{2} \cdot \Delta \cdot \alpha.
    \end{align*}
    Therefore, $\thd_H > 0.5$ and $\thd_L < 0.5$. Moreover, $\thd_H$ and $\thd_L$ does not depend on $\ag$. Therefore, $\liminf \sqrt{\ag} \cdot (\thd_H - \frac12) = + \infty$ and $\liminf \sqrt{\ag} \cdot (\frac12 - \thd_L) = + \infty$. 
    By applying the first case of Theorem~\ref{thm:arbitrary}, the fidelity of $\stgp_\ag$ converges to 1. 

    Now we show that $\stgp$ is an equilibrium with $\xi^* = \frac12 \cdot \Delta \cdot \alpha$. By Lemma~\ref{lem:acctoeq}, it suffices to consider a deviating group with only minority agents. For any $\xi < \frac{1}{2}\cdot \Delta \cdot \alpha$, $\thd_H > \frac12 + \xi$ and $\thd_L < \frac12 - \xi$. Therefore, for any group with a fraction $\xi$ of deviators and any deviating strategy, $\thd_H' > \frac12$ and $\thd_L < \frac12$. Likewise, applying the first case of Theorem~\ref{thm:arbitrary} and Lemma~\ref{lem:minority_deviate}, the fidelity still converges to 1 after any deviation, and the expected utility gain of minority agents converges to 0. Therefore, the strategy profile proposed is an $\varepsilon$ equilibrium with $\varepsilon$ converges to 0. 
\end{proof}

Then we give the lower bound for Segment 3.

\begin{prop}
    \label{prop:seg4lower} When $\anl > \alpha \ge (\frac12, \frac{1}{1+P_{\ell L} +(P_{\ell L} - P_{hH})})$, there exists a (sequence of) strategy profiles that satisfies the constraints with $\xi = \xinl$. 
\end{prop}

 By Lemma~\ref{lem:xilower}, it suffices to give a valid set of $\bph^1, (1 - \bpl^0)$, and $\gamma$. We specify the following set. 
\begin{align*}
        \bph^{1*} =&\ 1 \\
        1 - \bpl^{0*} = &\ (1 - 2\alpha) + \frac{\sqrt{2\cdot (1 - \alpha \cdot P_{\ell H})\cdot (1 - 2\alpha \cdot P_{\ell L}) \cdot P_{\ell L}P_{\ell H}}}{P_{\ell L} \cdot P_{\ell H}}\\
        \gamma^* = &\ \frac{\frac12 (P_{\ell L}\cdot (1 - \bpl^{0*}) + 2 - P_{\ell L}) -\alpha \cdot (P_{hH}P_{\ell L} + P_{\ell H}P_{\ell L} \cdot (1-\bpl^{0*}) +P_{\ell H})}{(P_{\ell L}\cdot (1 - \bpl^{0*}) + P_{hL})\cdot (P_{\ell H}\cdot (1 - \bpl^{0*}) + P_{hH} +1)}.
\end{align*}
In this case, by Lemma~\ref{lem:seg4_upper_max} and the proof of Lemma~\ref{lem:seg4_upper_fix}, $\min(\xi_h, \xi_\ell) = \xinl$. 
    
We show that these parameters are in a valid scope. Specifically, there are three restriction to be proved. 
\begin{enumerate}
        \item $1 - 1 - \bpl^{0*} \in [0, 1]$. 
        \item $(1 - \alpha - \gamma^*) > \frac{P_{hH}\cdot P_{\ell L}}{P_{\ell H}} (\alpha - \frac12)$ and $\gamma^* >\frac{P_{hH}\cdot P_{\ell L}}{P_{hL}} ( \alpha - \frac12)$.
        \item $\gamma^* > \xinl$ and $1 - \alpha - \gamma^* > \xinl$. 
\end{enumerate}

Lemma~\ref{lem:seg4_upper_max} gaurantees that $1 - \bpl^{0*} \in [0, 1]$. 

\begin{prop}
        \label{prop:seg4lower_gamma1}
        For all $\anl > \alpha > (\frac12, \frac{1}{1+P_{\ell L} +(P_{\ell L} - P_{hH})})$, $(1 - \alpha - \gamma^*) > \frac{P_{hH}\cdot P_{\ell L}}{P_{\ell H}} (\alpha - \frac12)$ and $\gamma^* >\frac{P_{hH}\cdot P_{\ell L}}{P_{hL}} ( \alpha - \frac12)$ holds.
\end{prop}

\begin{proof}
   Recall the definition of $\xi_h$ and $\xi_{\ell}$. 
    \begin{align*}
        &\ \xi_h =  \frac{\Delta(\alpha - 1/2 + (1 - \alpha - \gamma) \cdot \bph^1)}{P_{hH}\cdot P_{\ell L}\cdot \bph^1 + P_{\ell H}\cdot P_{\ell L} \cdot (1 - \bpl^0) + P_{\ell H}}.\\
        &\ \xi_{\ell} = \frac{\Delta (\alpha - 1/2 + \gamma \cdot (1 - \bpl^0))}{P_{hH}\cdot P_{hL}\cdot \bph^1 + P_{hH}\cdot P_{\ell L} \cdot (1 - \bpl^0) + P_{hL}}. 
    \end{align*}

    Now we view $\xi_h$ and $\xi_{\ell}$ as functions of $\gamma \in (-\infty, +\infty)$, denoted as $\xi_h(\gamma)$ and $\xi_{\ell}(\gamma)$. $\xi_h(\gamma)$ is monotonically decreasing on $\gamma$, and $\xi_{\ell}(\gamma)$ is monotonically increasing on $\gamma$. Moreover, in Proposition~\ref{prop:mid_biased_gamma_gen}, we show that when $\gamma \le\frac{P_{hH}\cdot P_{\ell L}}{P_{hL}} ( \alpha - \frac12)$, for any $(\bph^1, 1 - \bpl^0) \in [0, 1]^2$, $\xi_{\ell}(\gamma) \le \frac{\Delta(\alpha - 1/2)}{P_{hL}}$. Similarly, when $\gamma \ge (1 - \alpha) - \frac{P_{hH}\cdot P_{\ell L}}{P_{\ell H}} ( \alpha - \frac12)$, for any $(\bph^1, 1 - \bpl^0) \in [0, 1]^2$, $\xi_{h}(\gamma) \le \frac{\Delta(\alpha - 1/2)}{P_{\ell H}}$. Therefore, for $\gamma \in R \setminus (\frac{P_{hH}\cdot P_{\ell L}}{P_{hL}} ( \alpha - \frac12),(1 - \alpha) - \frac{P_{hH}\cdot P_{\ell L}}{P_{\ell H}} ( \alpha - \frac12))$, for any $(\bph^1, 1 - \bpl^0) \in [0, 1]^2$, $\min(\xi_h, \xi_{\ell}) \le \frac{\Delta(\alpha - 1/2)}{P_{hL}}$. However, given the $\bph^{1*}, 1 - \bpl^{0*}$, and $\gamma^*$ provided, $\hat{\xi}= \xinl > \frac{\Delta(\alpha - 1/2)}{P_{hL}}$. This directly implies that $(1 - \alpha - \gamma) > \frac{P_{hH}\cdot P_{\ell L}}{P_{\ell H}} ( \alpha - \frac12)$, $\gamma >\frac{P_{hH}\cdot P_{\ell L}}{P_{hL}} ( \alpha - \frac12)$ holds.  
\end{proof} 

Finally, we we prove that $\gamma^* > \xinl$ and $1 - \alpha - \gamma^* > \xinl$. We first give two alternative forms of $\gamma^*$. 

\begin{lem}
    \label{lem:seg4_gamma_alt}
    \begin{align*}
        \gamma^* = &\ \frac{(1 - 2\alpha P_{\ell L})\cdot (1 - \alpha P_{\ell H}) + \frac12 \cdot (\frac{1}{P_{\ell H}}-2\alpha)\cdot \sqrt{2\cdot (1 - \alpha \cdot P_{\ell H})\cdot (1 - 2\alpha \cdot P_{\ell L}) \cdot P_{\ell L}P_{\ell H}}}{4(1 - 2\alpha P_{\ell L})\cdot (1 - \alpha P_{\ell H}) +  (\frac{1}{P_{\ell L}} + \frac{2}{P_{\ell H}}-4\alpha)\cdot \sqrt{2\cdot (1 - \alpha \cdot P_{\ell H})\cdot (1 - 2\alpha \cdot P_{\ell L}) \cdot P_{\ell L}P_{\ell H}}}\\
        =&\ \frac{2P_{\ell L}^2 - (3 - 2\alpha \cdot P_{\ell H})\cdot P_{\ell H} \cdot P_{\ell L} + P_{\ell H}\cdot \sqrt{2\cdot (1 - \alpha \cdot P_{\ell H})\cdot (1 - 2\alpha \cdot P_{\ell L}) \cdot P_{\ell L}P_{\ell H}}}{2\cdot (2P_{\ell L} - P_{\ell H})^2}. 
    \end{align*}
\end{lem}

\begin{proof}
    We assign $1 - \bpl^{0*}$ into the numerator and denominator of $\gamma^*$ respectively to get the first alternative form. Recall that 
    \begin{equation*}
       1 - \bpl^{0*} = (1 - 2\alpha) + \frac{\sqrt{2\cdot (1 - \alpha \cdot P_{\ell H})\cdot (1 - 2\alpha \cdot P_{\ell L}) \cdot P_{\ell L}P_{\ell H}}}{P_{\ell L} \cdot P_{\ell H}}. 
    \end{equation*}

    For the numerator, 
    \begin{align*}
        &\ \frac12 (P_{\ell L}\cdot (1 - \bpl^{0*}) + 2 - P_{\ell L}) -\alpha \cdot (P_{hH}P_{\ell L} + P_{\ell H}P_{\ell L} \cdot (1-\bpl^{0*}) +P_{\ell H})\\
        =&\ 1- \alpha P_{\ell L} + \frac{\sqrt{2\cdot (1 - \alpha \cdot P_{\ell H})\cdot (1 - 2\alpha \cdot P_{\ell L}) \cdot P_{\ell L}P_{\ell H}}}{2P_{\ell H}}\\
        &\ - \alpha \cdot \left(P_{hH} P_{\ell L} + P_{\ell H} + (1 - 2\alpha) \cdot P_{\ell L}P_{\ell H} + \sqrt{2\cdot (1 - \alpha \cdot P_{\ell H})\cdot (1 - 2\alpha \cdot P_{\ell L}) \cdot P_{\ell L}P_{\ell H}}\right)\\
        =&\ 1 - \alpha P_{\ell L} + \frac{\sqrt{2\cdot (1 - \alpha \cdot P_{\ell H})\cdot (1 - 2\alpha \cdot P_{\ell L}) \cdot P_{\ell L}P_{\ell H}}}{2P_{\ell H}}\\
        &\ - \alpha \cdot \left(P_{\ell H} + P_{\ell L} - 2 \alpha \cdot P_{\ell L} P_{\ell H} + \sqrt{2\cdot (1 - \alpha \cdot P_{\ell H})\cdot (1 - 2\alpha \cdot P_{\ell L}) \cdot P_{\ell L}P_{\ell H}}\right)\\
        =&\ \frac12 \cdot (\frac{1}{P_{\ell H}}-2\alpha)\cdot \sqrt{2\cdot (1 - \alpha \cdot P_{\ell H})\cdot (1 - 2\alpha \cdot P_{\ell L}) \cdot P_{\ell L}P_{\ell H}}+ 1 -2\alpha P_{\ell L} - \alpha P_{\ell H} + 2\alpha^2 \cdot P_{\ell L} P_{\ell H}\\
        = &\ (1 - 2\alpha P_{\ell L})\cdot (1 - \alpha P_{\ell H}) + \frac12 \cdot (\frac{1}{P_{\ell H}}-2\alpha)\cdot \sqrt{2\cdot (1 - \alpha \cdot P_{\ell H})\cdot (1 - 2\alpha \cdot P_{\ell L}) \cdot P_{\ell L}P_{\ell H}}. 
    \end{align*}

    Similarly, for the denominator, 
    \begin{align*}
        &\ (P_{\ell L}\cdot (1 - \bpl^{0*}) + P_{hL})\cdot (P_{\ell H}\cdot (1 - \bpl^{0*}) + P_{hH} +1)\\
        =&\ \left( \frac{\sqrt{2\cdot (1 - \alpha \cdot P_{\ell H})\cdot (1 - 2\alpha \cdot P_{\ell L}) \cdot P_{\ell L}P_{\ell H}}}{P_{\ell H}} + 1 -2\alpha P_{\ell L}\right)\\\
        &\ \cdot \left( \frac{\sqrt{2\cdot (1 - \alpha \cdot P_{\ell H})\cdot (1 - 2\alpha \cdot P_{\ell L}) \cdot P_{\ell L}P_{\ell H}}}{P_{\ell L}} + 2 -2\alpha P_{\ell H}\right)\\
        =&\ 2\cdot (1 - \alpha \cdot P_{\ell H})\cdot (1 - 2\alpha \cdot P_{\ell L}) + 2(1 - 2\alpha P_{\ell L})\cdot (1 - \alpha P_{\ell H})\\
        &\ + (\frac{1 - 2\alpha P_{\ell L}}{P_{\ell L}} + \frac{2 - 2\alpha P_{\ell H}}{P_{\ell H}})\cdot \sqrt{2\cdot (1 - \alpha \cdot P_{\ell H})\cdot (1 - 2\alpha \cdot P_{\ell L}) \cdot P_{\ell L}P_{\ell H}}\\
        =&\ 4(1 - 2\alpha P_{\ell L})\cdot (1 - \alpha P_{\ell H}) +  (\frac{1}{P_{\ell L}} + \frac{2}{P_{\ell H}}-4\alpha)\cdot \sqrt{2\cdot (1 - \alpha \cdot P_{\ell H})\cdot (1 - 2\alpha \cdot P_{\ell L}) \cdot P_{\ell L}P_{\ell H}}. 
    \end{align*}

    The second form is reached by multiplying 
    $$4(1 - 2\alpha P_{\ell L})\cdot (1 - \alpha P_{\ell H}) - (\frac{1}{P_{\ell L}} + \frac{2}{P_{\ell H}}-4\alpha)\cdot \sqrt{2\cdot (1 - \alpha \cdot P_{\ell H})\cdot (1 - 2\alpha \cdot P_{\ell L}) \cdot P_{\ell L}P_{\ell H}}$$ together to the numerator and denominator to eliminate the square root in the denominator. 

    For the numerator, it becomes
    \begin{align*}
        &\ -\frac{1}{P_{\ell L}}(1 - 2\alpha P_{\ell L})\cdot (1 - \alpha P_{\ell H}) \cdot \sqrt{2\cdot (1 - \alpha \cdot P_{\ell H})\cdot (1 - 2\alpha \cdot P_{\ell L}) \cdot P_{\ell L}P_{\ell H}} \\
        &\ + (1 - 2\alpha P_{\ell L})\cdot (1 - \alpha P_{\ell H}) \cdot ( 4(1 - 2\alpha P_{\ell L})\cdot (1 - \alpha P_{\ell H}) - (1 - 2\alpha P_{\ell H})\cdot (1  + \frac{2P_{\ell L}}{P_{\ell H}}- 4\alpha P_{\ell L}))\\
        =&\ -\frac{1}{P_{\ell L}}(1 - 2\alpha P_{\ell L})\cdot (1 - \alpha P_{\ell H}) \cdot \sqrt{2\cdot (1 - \alpha \cdot P_{\ell H})\cdot (1 - 2\alpha \cdot P_{\ell L}) \cdot P_{\ell L}P_{\ell H}}\\
        &\ \frac{1}{P_{\ell H}} (1 - 2\alpha P_{\ell L})\cdot (1 - \alpha P_{\ell H}) \cdot ((3 - 2\alpha \cdot P_{\ell H})\cdot P_{\ell H} - 2P_{\ell L})
    \end{align*}

    For the denominator it becomes
    \begin{align*}
        &\ 16(1 - 2\alpha P_{\ell L})^2\cdot (1 - \alpha P_{\ell H})^2 -2\cdot (1 - \alpha \cdot P_{\ell H})\cdot (1 - 2\alpha \cdot P_{\ell L}) \cdot P_{\ell L}P_{\ell H} \cdot (\frac{1}{P_{\ell L}} + \frac{2}{P_{\ell H}}-4\alpha)^2\\
        =&\ 2(1 - 2\alpha P_{\ell L})\cdot (1 - \alpha P_{\ell H}) \\
        &\ \cdot ( 8(1 +2P_{\ell L}P_{\ell H}\alpha ^2 - (2P_{\ell L} + {\ell H}) \alpha) - (16P_{\ell L}P_{\ell H}\alpha ^2 + (2P_{\ell L} + {\ell H}) \alpha + \frac{(2P_{\ell L} +P_{\ell H})^2}{P_{\ell L} \cdot P_{\ell H}}))\\
        =&\ 2(1 - 2\alpha P_{\ell L})\cdot (1 - \alpha P_{\ell H}) \cdot \frac{8P_{\ell L} P_{\ell H} - (2P_{\ell L} +P_{\ell H})^2}{P_{\ell L} \cdot P_{\ell H}}\\
        =&\ \frac{2(1 - 2\alpha P_{\ell L})\cdot (1 - \alpha P_{\ell H}) \cdot (2P_{\ell L} -P_{\ell H})^2}{P_{\ell L} \cdot P_{\ell H}}.
    \end{align*}
    Therefore, putting the numerator and the denominator together, we have 
    \begin{equation*}
       \gamma^* =  \frac{2P_{\ell L}^2 - (3 - 2\alpha \cdot P_{\ell H})\cdot P_{\ell H} \cdot P_{\ell L} + P_{\ell H}\cdot \sqrt{2\cdot (1 - \alpha \cdot P_{\ell H})\cdot (1 - 2\alpha \cdot P_{\ell L}) \cdot P_{\ell L}P_{\ell H}}}{2\cdot (2P_{\ell L} - P_{\ell H})^2}. 
    \end{equation*}
\end{proof}

\begin{prop}
    \label{prop:seg4lower_gamma2}
        For all $\anl > \alpha > \frac12$, $\gamma^* > \xinl$. 
\end{prop}

\begin{proof}
    By Lemma~\ref{lem:seg4_upper_fix}, we know that 
    \begin{equation*}
        \xinl = \frac{\Delta \cdot (\alpha - 1/2 + 1/2 \cdot (1 - \bpl^{0*}))}{(P_{\ell L}\cdot (1 - \bpl^{0*}) + P_{hL})\cdot (P_{\ell H}\cdot (1 - \bpl^{0*}) + P_{hH} +1)}. 
    \end{equation*}
    The denominator of $\gamma^*$ and that of $\xinl$ are the same. Therefore, it suffice to compare the numerator and prove 
    \begin{equation*}
        \frac12 (P_{\ell L}\cdot (1 - \bpl^{0*}) + 2 - P_{\ell L}) -\alpha \cdot (P_{hH}P_{\ell L} + P_{\ell H}P_{\ell L} \cdot (1-\bpl^{0*}) +P_{\ell H}) >\Delta \cdot (\alpha - 1/2 + 1/2 \cdot (1 - \bpl^{0*})).
    \end{equation*}

    From Lemma~\ref{lem:seg4_gamma_alt} (the first form), we know the numerator of $\gamma^*$ equals to 
    $$ (1 - 2\alpha P_{\ell L})\cdot (1 - \alpha P_{\ell H}) + \frac12 \cdot (\frac{1}{P_{\ell H}}-2\alpha)\cdot \sqrt{2\cdot (1 - \alpha \cdot P_{\ell H})\cdot (1 - 2\alpha \cdot P_{\ell L}) \cdot P_{\ell L}P_{\ell H}}.$$

    Similarly, for $\xinl$, we have 
    \begin{align*}
        &\ \Delta \cdot (\alpha - 1/2 + 1/2 \cdot (1 - \bpl^{0*})) \\
        =&\ \frac{\Delta}{2P_{\ell L} \cdot P_{\ell H}} \cdot \sqrt{2\cdot (1 - \alpha \cdot P_{\ell H})\cdot (1 - 2\alpha \cdot P_{\ell L}) \cdot P_{\ell L}P_{\ell H}}. 
    \end{align*}
    By subtracting the numerator of $\xinl$ from that of $\gamma^*$, we have 
    \begin{equation*}
        (1 - 2\alpha P_{\ell L})\cdot (1 - \alpha P_{\ell H}) + \frac12 \cdot (\frac{1}{P_{\ell H}}-2\alpha - \frac{\Delta}{P_{\ell L} \cdot P_{\ell H}})\cdot \sqrt{2\cdot (1 - \alpha \cdot P_{\ell H})\cdot (1 - 2\alpha \cdot P_{\ell L} ) \cdot P_{\ell L}P_{\ell H}}
    \end{equation*}
    Note that 
    \begin{align*}
        \frac{1}{P_{\ell H}}-2\alpha - \frac{\Delta}{P_{\ell L} \cdot P_{\ell H}} = \frac{P_{\ell L} - (P_{\ell L} - P_{\ell H})}{P_{\ell L} \cdot P_{\ell H}} - 2\alpha = \frac{1}{P_{\ell L}} - 2\alpha. 
    \end{align*}

    Given the assumption that $\alpha < \anl < \frac{1}{2P_{\ell L}}$, every term in the subtraction is strictly positive. This implies that $\gamma^* > \xinl$. 
\end{proof}

Finally, we will prove $1 - \alpha - \gamma^* > \xinl$ by showing that $1 - \alpha - \gamma^* \ge \gamma^*$. This consists of two steps. First, we characterize the range where $1 - \alpha - \gamma^* \ge \gamma^*$ holds. Second, we show that this range covers $\anl \ge \alpha \ge \frac12$. 

\begin{lem}
    \label{lemma:seg4lower_gamma_half}
    When $(2 P_{hH} - P_{\ell L} - P_{\ell L} P_{hH}) < 0$, for all $\frac{1}{2P_{\ell L}} \ge \alpha > \frac12$, $1 - \alpha - \gamma^* \ge \gamma^*$.
    
    When $(2 P_{hH} - P_{\ell L} - P_{\ell L} P_{hH}) \ge 0$, for all
    $$\frac12 < \alpha \le \frac{P_{\ell H}^2 + 2P_{\ell L}^2 - P_{\ell H} P_{hH}P_{\ell L}}{P_{\ell H}^2 + 4P_{\ell L}^2 - 4 P_{\ell H} P_{hH}P_{\ell L}}-\frac{P_{\ell H} \cdot \sqrt{P_{\ell H} \cdot P_{\ell L} \cdot (2 P_{hH} - P_{\ell L} - P_{\ell L} P_{hH})}}{P_{\ell H}^2 + 4P_{\ell L}^2 - 4 P_{\ell H} P_{hH}P_{\ell L}},$$ $1 - \alpha - \gamma^* \ge \gamma^*$. 
\end{lem}

\begin{proof}
    We solve the inequality $1 - \alpha - \gamma^* \ge \gamma^*$, using the second alternative form of $\gamma^*$ in Lemma~\ref{lem:seg4_gamma_alt}. 
    \begin{align*}
        &\ \frac{2P_{\ell L}^2 - (3 - 2\alpha \cdot P_{\ell H})\cdot P_{\ell H} \cdot P_{\ell L} + P_{\ell H}\cdot \sqrt{2\cdot (1 - \alpha \cdot P_{\ell H})\cdot (1 - 2\alpha \cdot P_{\ell L}) \cdot P_{\ell L}P_{\ell H}}}{2\cdot (2P_{\ell L} - P_{\ell H})^2} \ge \frac{1 - \alpha }{2}\\
        \Leftrightarrow&\ P_{\ell H}\cdot \sqrt{2\cdot (1 - \alpha \cdot P_{\ell H})\cdot (1 - 2\alpha \cdot P_{\ell L}) \cdot P_{\ell L}P_{\ell H}} \\
        &\ \ge (2 P_{\ell L}^2 + P_{\ell H}^2 - P_{\ell L}P_{\ell H}) - \alpha \cdot (2P_{\ell H}^2P_{\ell L} + 4P_{\ell L}^2 - 4P_{\ell H}P_{\ell L} + P_{\ell H}) \\
    \end{align*}

    Given that the left-hand side of the inequality is positive (assuming $\frac{1}{2P_{\ell L}} \ge \alpha > \frac12$), it suffice to consider the inequality by squaring both sides, i.e., 
    \begin{align*}
        &\ P_{\ell H}^2 \cdot 2\cdot (1 - \alpha \cdot P_{\ell H})\cdot (1 - 2\alpha \cdot P_{\ell L}) \cdot P_{\ell L}P_{\ell H} \\
        &\ \ge ((2 P_{\ell L}^2 + P_{\ell H}^2 - P_{\ell L}P_{\ell H}) - \alpha \cdot (2P_{\ell H}^2P_{\ell L} + 4P_{\ell L}^2 - 4P_{\ell H}P_{\ell L} + P_{\ell H}))^2.
    \end{align*}
    By refactoring on $\alpha$, this inequality is equivalent to 
    \begin{align*}
        (2P_{\ell L}-P_{\ell H})^2 \cdot (\alpha^2 \cdot (4P_{\ell L}^2 + P_{\ell H}^2 - 4 P_{\ell L}P_{\ell H}P_{hH}) - 2\alpha \cdot (2P_{\ell L}^2 + P_{\ell H}^2 - P_{\ell L}P_{\ell H}P_{hH}) + (P_{\ell L}^2 +P_{\ell H}^2)) \ge 0. 
    \end{align*}
    Applying the root formula of the quadratic equation on $$\alpha^2 \cdot (4P_{\ell L}^2 + P_{\ell H}^2 - 4 P_{\ell L}P_{\ell H}P_{hH}) - 2\alpha \cdot (2P_{\ell L}^2 + P_{\ell H}^2 - P_{\ell L}P_{\ell H}P_{hH}) + (P_{\ell L}^2 +P_{\ell H}^2) = 0,$$ we get two solutions. 
    \begin{equation*}
        \alpha = \frac{P_{\ell H}^2 + 2P_{\ell L}^2 - P_{\ell H} P_{hH}P_{\ell L}}{P_{\ell H}^2 + 4P_{\ell L}^2 - 4 P_{\ell H} P_{hH}P_{\ell L}} \pm \frac{P_{\ell H} \cdot \sqrt{P_{\ell H} \cdot P_{\ell L} \cdot (2 P_{hH} - P_{\ell L} - P_{\ell L} P_{hH})}}{P_{\ell H}^2 + 4P_{\ell L}^2 - 4 P_{\ell H} P_{hH}P_{\ell L}}
    \end{equation*}
    
    When $(2 P_{hH} - P_{\ell L} - P_{\ell L} P_{hH}) < 0$, the two solutions are complex. This directly implies that 
    \begin{equation*}
        \alpha^2 \cdot (4P_{\ell L}^2 + P_{\ell H}^2 - 4 P_{\ell L}P_{\ell H}P_{hH}) - 2\alpha \cdot (2P_{\ell L}^2 + P_{\ell H}^2 - P_{\ell L}P_{\ell H}P_{hH}) + (P_{\ell L}^2 +P_{\ell H}^2) > 0
    \end{equation*}
    always holds. Therefore, for all $\frac{1}{2P_{\ell L}} \ge \alpha > \frac12$, $1 - \alpha - \gamma^* \ge \gamma^*$.

     When $(2 P_{hH} - P_{\ell L} - P_{\ell L} P_{hH}) \ge 0$, the two solutions are real numbers. This implies that when $$\frac12 < \alpha \le \frac{P_{\ell H}^2 + 2P_{\ell L}^2 - P_{\ell H} P_{hH}P_{\ell L}}{P_{\ell H}^2 + 4P_{\ell L}^2 - 4 P_{\ell H} P_{hH}P_{\ell L}}-\frac{P_{\ell H} \cdot \sqrt{P_{\ell H} \cdot P_{\ell L} \cdot (2 P_{hH} - P_{\ell L} - P_{\ell L} P_{hH})}}{P_{\ell H}^2 + 4P_{\ell L}^2 - 4 P_{\ell H} P_{hH}P_{\ell L}},$$ $1 - \alpha - \gamma^* \ge \gamma^*$. 
\end{proof}

\begin{lem} 
\label{lem:seg4lower_gamma_half2}
When $(2 P_{hH} - P_{\ell L} - P_{\ell L} P_{hH}) \ge 0$, 
    $$\anl \le \frac{P_{\ell H}^2 + 2P_{\ell L}^2 - P_{\ell H} P_{hH}P_{\ell L}}{P_{\ell H}^2 + 4P_{\ell L}^2 - 4 P_{\ell H} P_{hH}P_{\ell L}}-\frac{P_{\ell H} \cdot \sqrt{P_{\ell H} \cdot P_{\ell L} \cdot (2 P_{hH} - P_{\ell L} - P_{\ell L} P_{hH})}}{P_{\ell H}^2 + 4P_{\ell L}^2 - 4 P_{\ell H} P_{hH}P_{\ell L}}.$$
\end{lem}

\begin{proof}
    Recall that 
    \begin{equation*}
        \anl = \frac{2P_{hH}P_{\ell L}^2 +P_{hH}\cdot (P_{hH}+3P_{\ell L})-(3P_{hH} +P_{\ell L}) + 2 - 2P_{hL}\cdot \sqrt{P_{hH}P_{hL}P_{\ell H}P_{\ell L}}}{2P_{\ell H}^2 + 8P_{hH}P_{\ell L}^2}.
    \end{equation*}

    We first use scaling methods to unify the square roots in these formulas. 
    Note that when $(2 P_{hH} - P_{\ell L} - P_{\ell L} P_{hH}) \ge 0$
    \begin{align*}
        \sqrt{P_{\ell H} \cdot P_{\ell L} \cdot (2 P_{hH} - P_{\ell L} - P_{\ell L} P_{hH})} \le  &\ \sqrt{P_{\ell H} \cdot P_{\ell L} \cdot (P_{hH} - P_{\ell L} P_{hH})} = \sqrt{P_{hH}P_{hL}P_{\ell H}P_{\ell L}}.
    \end{align*}

    Therefore, it suffice to show that 
    \begin{equation*}
        \anl \le \frac{P_{\ell H}^2 + 2P_{\ell L}^2 - P_{\ell H} P_{hH}P_{\ell L}}{P_{\ell H}^2 + 4P_{\ell L}^2 - 4 P_{\ell H} P_{hH}P_{\ell L}}-\frac{P_{\ell H} \cdot \sqrt{P_{hH}P_{hL}P_{\ell H}P_{\ell L}}}{P_{\ell H}^2 + 4P_{\ell L}^2 - 4 P_{\ell H} P_{hH}P_{\ell L}}.
    \end{equation*}
    This inequality is equivalent to 
    \begin{align*}
        &\ (P_{\ell H}^2 + 2P_{\ell L}^2 - P_{\ell H} P_{hH}P_{\ell L})\cdot (2P_{\ell H}^2 + 8P_{hH}P_{\ell L}^2)\\
        &\ - (P_{\ell H}^2 + 4P_{\ell L}^2 - 4 P_{\ell H} P_{hH}P_{\ell L})\cdot (2P_{hH}P_{\ell L}^2 +P_{hH}\cdot (P_{hH}+3P_{\ell L})-(3P_{hH} +P_{\ell L}) + 2)\\
        &\ +2 \cdot (P_{hL} \cdot (P_{\ell H}^2 + 4P_{\ell L}^2 - 4 P_{\ell H} P_{hH}P_{\ell L}) - P_{\ell L}\cdot (P_{\ell H}^2 + 4P_{hH}P_{\ell L}^2)) \cdot \sqrt{P_{hH}P_{hL}P_{\ell H}P_{\ell L}}\\
        \ge &\ 0. 
    \end{align*}

    By reorganizing the formulas, the non-square-root term equals to $(P_{\ell L} - P_{hH}) \cdot (2P_{\ell L} - P_{\ell H})^2 \cdot (P_{\ell H} + 2 P_{\ell L}P_{hH})$, and the square-root term equals to $-2(P_{\ell L} - P_{hH}) \cdot (2P_{\ell L} - P_{\ell H})^2 \cdot \sqrt{P_{hH}P_{hL}P_{\ell H}P_{\ell L}}$. Therefore, the inequality is equivalent to
    \begin{equation*}
        (P_{\ell L} - P_{hH}) \cdot (2P_{\ell L} - P_{\ell H})^2 \cdot (P_{\ell H} + 2 P_{\ell L}P_{hH} - 2\sqrt{P_{hH}P_{hL}P_{\ell H}P_{\ell L}}) \ge 0.
    \end{equation*}
    By $P_{hH} > P_{hL}$ and $P_{\ell L} > P_{\ell H}$, $P_{\ell L}P_{hH} >\sqrt{P_{hH}P_{hL}P_{\ell H}P_{\ell L}}$. Together with the assumption $P_{\ell L} \ge P_{hH}$. Therefore, the inequality holds, which finishes the proof. 
\end{proof}

Combining Lemma~\ref{lemma:seg4lower_gamma_half} and ~\ref{lem:seg4lower_gamma_half2}, we have the following proposition. 

\begin{prop}
    \label{prop:seg4lower_gamma3}
    For all $\anl > \alpha > \frac12$, $1 - \alpha - \gamma^* > \xinl$. 
\end{prop}

Therefore, combining Lemma~\ref{lem:seg4_upper_max}, Proposition~\ref{prop:seg4lower_gamma1},\ref{prop:seg4lower_gamma2}, and \ref{prop:seg4lower_gamma3}, Proposition~\ref{prop:seg4lower} is proved. 

\subsection{Segment 2 to 4: Wrap up the Lower Bound.}
Finally, we show the tight lower bound corresponding to the upper bound in each segment. 

\begin{enumerate}
    \item In ``Steep'' Segment 2, Proposition~\ref{prop:seg3lower} in the biased case gives the lower bound $\frac{\Delta(\alpha - 1/2)}{P_{hL}}$, where all minority agents always vote for $\bA$. 
    \item In ``Non-Linear'' Segment 3, we show that the corresponding $\xi^*$ and $1 - \bpl^0$ that maximizes $\hat{\xi}$ serves as a valid set of parameters in Proposition~\ref{prop:seg4lower}, so $\xinl$ is also a lower bound. 
    \item For ``Tail'' Segment 4, the parameters $\bph^1 = 1 - \bpl^0 = 1$, and $\gamma = \frac{1 - (P_{\ell L} + P_{\ell H})\cdot \alpha}{2}$ gives the lower bound $\frac12\Delta\alpha$ by Proposition~\ref{prop:seg5lower}.  
\end{enumerate}

Therefore, the lower bound on $\xi$ is exactly given in $\xi^*(\alpha)$. We finish the proof that $\xi^*(\alpha)$ is exactly the threshold curve for which a strategy profile satisfying all three constraints exists. 

\end{document}